\documentclass[12pt, onecolumn]{IEEEtran} 
\usepackage{url}
\usepackage{amsfonts}
\usepackage{amssymb}
\usepackage{amsmath}
\usepackage{graphicx}
\usepackage{epsfig}
\usepackage{bm} 
\usepackage{bbm} 
\usepackage{enumerate}

\usepackage{color} 
\usepackage[hidelinks]{hyperref} 
\hypersetup{
    colorlinks=false, 
    linktoc=all 
}

\newcommand{\beq}[1]{\begin{equation}\label{#1}}
\newcommand{\eeq}{\end{equation}}

\newcommand{\beqn}[1]{\begin{eqnarray}\label{#1}}
\newcommand{\eeqn}{\end{eqnarray}}

\newtheorem{thmbody}{Theorem}
\newenvironment{thm}{
\begin{thmbody}
	}{
	\end{thmbody} 
	}
\newtheorem{dfnbody}{Definition}

\newtheorem{corbody}{Corollary}
\newenvironment{cor}{
\begin{corbody}
	}{
	\end{corbody} 
	}
\newtheorem{lemmabody}{Lemma}
\newenvironment{lemma}{
\begin{lemmabody}
	}{
	\end{lemmabody} 
	}
\newtheorem{propbody}{Proposition}

\newenvironment{proof}{
	{\it Proof:}
	}{
 $\Box$
	}

\usepackage{dsfont}
\begin{document}
\title{Analogy and duality\\ between \\random channel coding and lossy source coding}

\author{Sergey~Tridenski and Ram~Zamir
}

\maketitle





%
%
%
%

Here 
we
write in a unified fashion (using ``$R(P, Q, D)$'' \cite{ZamirRose01}) the random coding exponents in channel coding and lossy source coding.\footnote{This paper is self-contained, and serves also as an addendum to our paper ``Exponential source/channel duality''. The proofs of the formulas with $R(P, Q, D)$ will be given below.}
We derive their explicit forms and
show, that, for a given random codebook distribution $Q$, the channel {\em decoding error} exponent
can be viewed as an {\em encoding success} exponent in lossy source coding,
and 
the channel {\em correct-decoding} exponent
can be viewed as an {\em encoding failure} exponent in lossy source coding.
We then extend the channel exponents to arbitrary $D$, which corresponds for $D>0$ to erasure decoding
and for $D < 0$ to list decoding.
For comparison,
we also derive the exact random coding exponent for Forney's optimum tradeoff decoder \cite{Forney68}.

In the case of source coding, we assume discrete memoryless sources with a finite alphabet ${\cal X}$ and a finite reproduction alphabet $\hat{\cal X}$.
In the case of channel coding,
we assume discrete memoryless channels with finite input and output alphabets ${\cal X}$ and ${\cal Y}$,
such that for any $(x, y) \, \in\, {\cal X} \times {\cal Y}$ the channel probability is positive $P(y \, | \, x) \, > \, 0$.
For simplicity, let $R$ denote an exponential size of a random codebook,
such that there exist block lengths $n$ for which $e^{nR}$ is integer. We assume the size of the codebook $M \, = \, e^{nR}$ for source coding,
and $M \, = \, e^{nR} \, + \, 1$ for channel coding.
Let $Q$ denote the (i.i.d.) distribution,
according to which the codebook is generated.
We use also the definition:
\begin{equation} \label{eqRTQDDefinition}
R(T, Q, D) \;\; \triangleq \;\; \min_{W(\hat{x} \,|\, x): \;\; d(T\,\circ\, W) \; \leq \; D} D(T \circ W \; \| \; T \times Q),
\end{equation}
where $T(x)$ is a distribution over ${\cal X}$, $Q(\hat{x})$ is a distribution over $\hat{\cal X}$, and $d(T\circ W)$ denotes an average distortion measure $d(x, \hat{x})$. We consider $R(T, Q, D)\,=\,+\infty$, if the set $\left\{{W(\hat{x} \,|\, x): \;\; d(T\,\circ\, W) \; \leq \; D}\right\}$ is empty.

\tableofcontents

\newpage
\section{\bf Encoding success exponent (for sources)} \label{S1}
\begin{thm} \label{thm1}
{\em For a source $P(x)$ and distortion constraint $D$,
the exponent in the probability of successful encoding is given by}
\begin{equation} \label{eqSuccess}
\lim_{n \, \rightarrow \, \infty} \; \left\{-\frac{1}{n}\ln P_{s}\right\} \;\; = \;\;
E_{s}(R, D)
\;\; \triangleq \;\;
\min_{T(x)} \; \Big\{ D(T \| P) \;\; + \;\; {\big| R(T, Q, D) \; - \; R   \big|\mathstrut}_{}^{+} \Big\},
\end{equation}
{\em except possibly for $D \, = \, D_{\min}\, = \, \min_{x, \, \hat{x}}d(x, \hat{x})$, when the RHS is a lower bound.}
\end{thm}
Note that the exponent (\ref{eqSuccess}) is zero for $R \, \geq \, R(P, Q, D)$.
This theorem is proved in Section~\ref{S19}.

\section{\bf Channel decoding error exponent} \label{S2}
For a channel $P(y \,|\, x)$, the exponent in the probability of decoding error is given by
\begin{equation} \label{eqError}
E_{e}(R) \; = \; \min_{T(x, \, y)} \; \Big\{ D(T \; \| \; Q \circ P) \;\; + \;\; {\big| R(T, Q, 0) \; - \; R   \big|\mathstrut}_{}^{+} \Big\},
\end{equation}
where $R(T, Q, D=0)$ is determined with respect to a particular distortion measure defined as
\begin{equation} \label{eqDmeasure}
d\big((x, y), \hat{x}\big) \; \triangleq \; \ln \frac{P(y \,|\, x)}{P(y \,|\, \hat{x})}.
\end{equation}
Note that this exponent is zero for $R \, \geq \, R(Q\circ P, \, Q, \, 0) \, = \, I(Q\circ P)$.

\section{\bf Encoding failure exponent (for sources)} \label{S3}
\begin{thm} \label{thm2}
{\em For a source $P(x)$ and distortion constraint $D$,
the exponent in the probability of encoding failure is given by}
\begin{equation} \label{eqFailure}
\lim_{n \, \rightarrow \, \infty} \; \left\{-\frac{1}{n}\ln P_{f}\right\} \; = \;
E_{f}(R, D) \; \triangleq \; \left\{
\begin{array}{r l}
\displaystyle \min_{T(x): \;\; R(T, \,Q, \,D) \; \geq \; R} \; D(T \| P), & \;\;\; R \; \leq \; R_{\text{max}}(D), \\
+\infty , & \;\;\; R \; > \; R_{\text{max}}(D),
\end{array}
\right.
\end{equation}
{\em where $R_{\text{max}}(D) \, \triangleq \, \max_{T(x)}R(T, Q, D)$,\footnote{$\max_{T(x)}R(T, Q, D)$ may be alternatively expressed as $\max_{x}$, but it can be $+\infty$.} with the possible exception of points of discontinuity of the function $E_{f}(R, D)$.}
\end{thm}
This exponent is zero for $R \, \leq \, R(P, Q, D)$. For $R$ above $R_{\text{max}}(D)$,
the probability of encoding failure tends to zero super-exponentially as $n$ increases, i.e. the limit of its exponent, as $n \, \rightarrow \, \infty$
(which is exactly ``the exponent'' by definition), is infinity.
This theorem is proved in Section~\ref{S21}.

\section{\bf Channel correct-decoding exponent} \label{S4}
For a channel $P(y \,|\, x)$, the exponent in the probability of correct decoding is given by
\begin{equation} \label{eqCorrect}
E_{c}(R) \; = \; \min_{T(x, \, y)} \; \Big\{ D(T \; \| \; Q \circ P) \;\; + \;\; {\big| R \; - \; R(T, Q, 0)\big|\mathstrut}_{}^{+} \Big\},
\end{equation}
where $R(T, Q, D=0)$ is determined with respect to the distortion measure $d\big((x, y), \hat{x}\big)$ (\ref{eqDmeasure}).
This exponent coincides with
\begin{equation} \label{eqCorrectEquiv}
E_{c}(R) \; = \; \left\{
\begin{array}{r l}
\displaystyle \min_{T(x, \, y): \;\; R(T, \,Q, \,0) \; \geq \; R} \; D(T \; \| \; Q \circ P),
\;\;\;\;\;\;\;\;\;\;\;\;\;\;\;\;\;\;\;\;\;\;\;\;\;\;\;\;\;\;\;\;\;\;\;\;\
& \;\;\; R \; \leq \; R_{1}, \\
\displaystyle \min_{T(x, \, y)} \; \big\{ D(T \; \| \; Q \circ P)  \;\; + \;\; R \;\; - \;\; R(T, Q, 0) \big\}, & \;\;\; R \; > \; R_{1},
\end{array}
\right.
\end{equation}
where $R_{1} \, \leq \, \max_{T(x, \, y)}R(T, Q, 0)$.
The exponent is zero for $R \, \leq \, R(Q\circ P, \, Q, \, 0) \, = \, I(Q\circ P)$.
For $R \, > \, R_{1}$, the exponent is a linearly increasing function of $R$ with constant slope $=1$.\footnote{If the exponent is for the natural base $e$, then $R$ here must be accordingly in natural units ({\em nats}).}

\section{\bf Derivation of the explicit encoding success exponent} \label{S5}
We start with a derivation of an explicit formula for $R(T, Q, D)$:
\begin{lemma} \label{lemma1}
\begin{equation} \label{eqRTQD}
R(T, Q, D) \; = \; \sup_{s\,\geq\,0} \; \bigg\{-\sum_{x}T(x) \ln \sum_{\hat{x}}Q(\hat{x})e^{-s[d(x,\, \hat{x})-D]}\bigg\}.
\end{equation}
\end{lemma}
\begin{proof}
\begin{align}
R(T, Q, D) \;\; & \triangleq \;\; \min_{W(\hat{x} \,|\, x): \;\; d(T\,\circ\, W) \; \leq \; D} D(T \circ W \; \| \; T \times Q)
\nonumber \\
& = \footnotemark \;\;
\min_{W(\hat{x} \,|\, x)} \;\sup_{s \,\geq \,0} \;\big\{ D(T \circ W \; \| \; T \times Q) \; + \; s\big[d(T\circ W) \; - \; D\big]\big\}
\label{eqRTQDminsup} \\
& \overset{(*)}{=} \;\;\;\, \sup_{s \,\geq \,0} \; \min_{W(\hat{x} \,|\, x)} \;\big\{ D(T \circ W \; \| \; T \times Q) \; + \; s\big[d(T\circ W) \; - \; D\big]\big\}
\label{eqWs} \\
& = \;\;\;\; \sup_{s \,\geq \,0} \; \min_{W(\hat{x} \,|\, x)} \;\Bigg\{ \sum_{x,\,\hat{x}}T(x)W(\hat{x} \,|\, x)\ln \frac{W(\hat{x} \,|\, x)}{Q(\hat{x})} \; + \; s\bigg[\sum_{x,\,\hat{x}}T(x)W(\hat{x} \,|\, x)d(x, \, \hat{x}) \; - \; D\bigg]\Bigg\}
\nonumber \\
& = \;\;\;\; \sup_{s \,\geq \,0} \;
\min_{W(\hat{x} \,|\, x)} \; \bigg\{ \sum_{x,\,\hat{x}}T(x)W(\hat{x} \,|\, x)\ln \frac{W(\hat{x} \,|\, x)}{Q(\hat{x})e^{-sd(x,\, \hat{x})}} \; - \; sD \bigg\}
\nonumber \\
& = \;\;\;\; \sup_{s\,\geq\,0} \; \bigg\{-\sum_{x}T(x) \ln \sum_{\hat{x}}Q(\hat{x})e^{-sd(x,\, \hat{x})} \; - \; sD\bigg\},
\nonumber
\end{align}\footnotetext{This also includes the case when the set $\{W(\hat{x} \,|\, x): \; d(T\circ W) \, \leq \, D\}$
is empty, then $R(T, Q, D)\,=\,+\infty$.}
where ($*$) follows by the minimax theorem\footnote{The equality can also be verified directly, for different values of $D$, using continuity of the minimizing solution $W_{s}$ and its limit as $s\,\rightarrow\,+\infty$, or, alternatively, showing that $R(T,Q,D)$ is a convex ($\cup$) function of $D$ and (\ref{eqWs}) is the lower convex envelope of $R(T,Q,D)$.}, since the objective function is convex ($\cup$) in $W(\hat{x} \,|\, x)$
and concave (linear) in $s$.
\end{proof}

Before we plug the explicit formula for $R(T, Q, D)$ (\ref{eqRTQD}) into the expression for the encoding success exponent,
we note the following property:
\begin{lemma} \label{lemma2}
{\em $R(T, Q, D)$ is a convex ($\cup$) function of $\,(T, D)$.}
\end{lemma}
\begin{proof}
\begin{align}
& R\big(\lambda T + (1 - \lambda)\widetilde{T}, \, Q, \, \lambda D + (1 - \lambda)\,\widetilde{\!D}\big)
\nonumber \\
& \; = \; \sup_{s\,\geq\,0} \; \bigg\{-\sum_{x}\big(\lambda T(x) + (1 - \lambda)\widetilde{T}(x)\big) \ln \sum_{\hat{x}}Q(\hat{x})e^{-sd(x,\, \hat{x})} \; - \; s\big(\lambda D + (1 - \lambda)\,\widetilde{\!D}\big)\bigg\}
\nonumber \\
& \; = \; \sup_{s\,\geq\,0} \; \Bigg\{\lambda\bigg[-\sum_{x} T(x) \ln \sum_{\hat{x}}Q(\hat{x})e^{-sd(x,\, \hat{x})} \; - \; sD\bigg]\;+
\nonumber \\
& \;\;\;\;\;\;\;\;\;\;\;\;\;\;\;\; + \;
(1 - \lambda)\bigg[-\sum_{x} \widetilde{T}(x) \ln \sum_{\hat{x}}Q(\hat{x})e^{-sd(x,\, \hat{x})} \; - \; s\,\widetilde{\!D}\bigg]
\Bigg\}
\nonumber \\
& \; \leq \; \lambda \sup_{s\,\geq\,0} \; \bigg\{-\sum_{x} T(x) \ln \sum_{\hat{x}}Q(\hat{x})e^{-sd(x,\, \hat{x})} \; - \; sD\bigg\}\;+
\nonumber \\
& \;\;\;\;\;\;\;\;\;\;\;\;\;\;\;\; + \;
(1 - \lambda)\sup_{s\,\geq\,0} \; \bigg\{-\sum_{x} \widetilde{T}(x) \ln \sum_{\hat{x}}Q(\hat{x})e^{-sd(x,\, \hat{x})} \; - \; s\,\widetilde{\!D} \bigg\}
\nonumber \\
& \; = \; \lambda R(T, Q, D) \; + \; (1 - \lambda)R(\widetilde{T}, Q, \,\widetilde{\!D}).
\nonumber
\end{align}
\end{proof}

The encoding success exponent can be rewritten as
\begin{lemma} \label{lemma3}
\begin{equation} \label{eqEs}
\min_{T(x)} \; \Big\{ D(T \| P) \;\; + \;\; {\big| R(T, Q, D) \; - \; R   \big|\mathstrut}_{}^{+} \Big\}
\; = \;
\sup_{0 \, \leq \,\rho \,\leq \,1} \; \min_{T(x)} \;\big\{ D(T \| P) \; + \; \rho\big[R(T, Q, D) \; - \; R\big]\big\}.
\end{equation}
\end{lemma}
\begin{proof}
The expression for the encoding success exponent (\ref{eqSuccess}), which is written with the help of the Csisz\'ar-K\"orner
style brackets $|\cdot|^{+}$ for compactness, translates into the minimum between two exponents:
\begin{align}
& \min_{T(x)} \; \Big\{ D(T \| P) \;\; + \;\; {\big| R(T, Q, D) \; - \; R   \big|\mathstrut}_{}^{+} \Big\}
\; = \; \min \; \big\{E_{A}(R, D), \;\; E_{B}(R, D)\big\} \; =
\label{eqE1E2} \\
& \min \; \bigg\{\min_{T(x): \;\; R(T, Q, D) \; \leq \; R} \; D(T \| P), \;\;\;\;\;\;
\min_{T(x): \;\; R(T, Q, D) \; \geq \; R} \; \big\{ D(T \| P) \; + \; R(T, Q, D) \; - \; R \big\}
\bigg\}.
\nonumber
\end{align}
Using the fact that $R(T, Q, D)$ is convex ($\cup$) in $T$, we can rewrite the left exponent as follows
\begin{align}
E_{A}(R, D) \; = \; \min_{T(x): \;\; R(T, Q, D) \; \leq \; R} \; D(T \| P) \; & = \;\,
\min_{T(x)} \;\sup_{\rho \,\geq \,0} \;\big\{ D(T \| P) \; + \; \rho\big[R(T, Q, D) \; - \; R\big]\big\}
\nonumber \\
& \overset{(*)}{=} \; \sup_{\rho \,\geq \,0} \; \min_{T(x)} \;\big\{ D(T \| P) \; + \; \rho\big[R(T, Q, D) \; - \; R\big]\big\},
\label{eqE1}
\end{align}
where ($*$) follows by the minimax theorem\footnote{Alternatively, check directly that $E_{A}(R, D)$ is convex ($\cup$) in $R$ and observe that (\ref{eqE1}) is the lower convex envelope of $E_{A}(R)$.}, because the objective function is convex ($\cup$) in $T(x)$
and concave (linear) in $\rho$.

For the right exponent we have a lower bound:
\begin{align}
E_{B}(R, D) \; & = \; \min_{T(x): \;\; R(T, Q, D) \; \geq \; R} \; \big\{ D(T \| P) \; + \; R(T, Q, D) \; - \; R \big\} \nonumber \\
& \geq \;
\sup_{\rho \,\geq \,0} \;\; \min_{T(x): \;\; R(T, Q, D) \; \geq \; R} \;\big\{ D(T \| P) \; + \; R(T, Q, D) \; - \; R \; + \; \rho\big[R \; - \; R(T, Q, D)\big]\big\}
\nonumber \\
& \geq \;
\sup_{\rho \,\geq \,0} \;\;\;\;\;\;\;\;\;\;\; \min_{T(x)}
\;\;\;\;\;\;\;\;\;\;
\big\{ D(T \| P) \; + \; R(T, Q, D) \; - \; R \; + \; \rho\big[R \; - \; R(T, Q, D)\big]\big\}
\nonumber \\
& = \;
\sup_{\rho \,\geq \,0} \;\;\;\;\;\;\;\;\;\;\; \min_{T(x)}
\;\;\;\;\;\;\;\;\;\;
\big\{ D(T \| P) \; + \; (1-\rho)\big[R(T, Q, D) \; - \; R\big]\big\}
\nonumber \\
& \geq \;
\sup_{0 \, \leq \,\rho \,\leq \,1} \;\;\;\;\;\;\; \min_{T(x)}
\;\;\;\;\;\;\;\;\;\;
\big\{ D(T \| P) \; + \; (1-\rho)\big[R(T, Q, D) \; - \; R\big]\big\}
\nonumber \\
& = \;
\sup_{0 \, \leq \,\rho \,\leq \,1} \;\;\;\;\;\;\; \min_{T(x)}
\;\;\;\;\;\;\;\;\;\;
\big\{ D(T \| P) \; + \; \rho\big[R(T, Q, D) \; - \; R\big]\big\}.
\label{eqE2}
\end{align}

Observe from (\ref{eqRTQD}), that if $\,d(x,\, \hat{x})-D \, > \, 0$ for all $(x,\, \hat{x})$, then $R(T, Q, D)\,=\,+\infty$
for any choice of $T$.
In this case (\ref{eqEs}) holds trivially.

On the other hand, if there exists at least one pair $(x,\, \hat{x})$,
such that $\,d(x,\, \hat{x})-D \, \leq \, 0$,
then there exists $T$ with finite $R(T, Q, D)$. In this case,
consider the following function of $T$:
\begin{displaymath}
D(T \| P) \; + \; R(T, Q, D).
\end{displaymath}
This is a strictly convex ($\cup$) function of $T$, because $R(T, Q, D)$ is convex and $D(T \| P)$ is strictly convex.
Consequently, there exists a unique $T_{1}$, which attains its minimum:
\begin{displaymath}
D(T_{1} \| P) \; + \; R(T_{1}, Q, D) \; = \; \min_{T(x)} \big\{D(T \| P) \; + \; R(T, Q, D)\big\}.
\end{displaymath}
Note that for $R_{1} \, = \, R(T_{1}, Q, D)$ we obtain:
\begin{displaymath}
E_{A}(R_{1}, D) \; = \; \min_{T(x): \;\; R(T, Q, D) \; \leq \; R_{1}} \; D(T \| P) \; = \;
\min_{T(x)} \;\big\{ D(T \| P) \; + \; R(T, Q, D) \; - \; R_{1}\big\}.
\end{displaymath}
Since $E_{A}(R_{1}, D)$ is finite, we conclude that for $R \geq R_{1}$ the function $E_{A}(R)$ is finite and nonincreasing.
It can be seen from (\ref{eqE1}) that $E_{A}(R, D)$ is a convex ($\cup$) function of $R$.
We conclude, that for $R \geq R_{1}$, in (\ref{eqE1}) it is sufficient to take the supremum over $0 \, \leq \,\rho \,\leq \,1$:
\begin{equation} \label{eqE11}
E_{A}(R, D) \; = \; \sup_{0 \, \leq \,\rho \,\leq \,1} \; \min_{T(x)} \;\big\{ D(T \| P) \; + \; \rho\big[R(T, Q, D) \; - \; R\big]\big\},
\;\;\;\;\;\; R \; \geq \; R_{1}.
\end{equation}
Observe further, that for $R \leq R_{1}$
\begin{align}
E_{B}(R, D) \; & = \; \min_{T(x): \;\; R(T, Q, D) \; \geq \; R} \; \big\{ D(T \| P) \; + \; R(T, Q, D) \; - \; R \big\}
\nonumber \\
& = \; D(T_{1} \| P) \; + \; R(T_{1}, Q, D) \; - \; R \nonumber \\
& = \; \min_{T(x)} \; \big\{ D(T \| P) \; + \; R(T, Q, D) \; - \; R \big\}
\nonumber \\
& \leq \;
\sup_{0 \, \leq \,\rho \,\leq \,1} \; \min_{T(x)} \;
\big\{ D(T \| P) \; + \; \rho\big[R(T, Q, D) \; - \; R\big]\big\}, \;\;\;\;\;\; R \; \leq \; R_{1}.
\label{eqE2Upper}
\end{align}
Comparing (\ref{eqE2}) and (\ref{eqE2Upper}), we conclude that the equality holds
\begin{equation} \label{eqE21}
E_{B}(R, D) \; = \; \sup_{0 \, \leq \,\rho \,\leq \,1} \; \min_{T(x)} \;\big\{ D(T \| P) \; + \; \rho\big[R(T, Q, D) \; - \; R\big]\big\},
\;\;\;\;\;\; R \; \leq \; R_{1}.
\end{equation}
Now, the result of the lemma follows by (\ref{eqE1E2}), when we compare (\ref{eqE1}) with (\ref{eqE21}) for $R \, \leq \, R_{1}$, and (\ref{eqE2}) with (\ref{eqE11}) for $R \, \geq \, R_{1}$, respectively.
\end{proof}

Finally, we are ready to
prove the following formula:
\begin{thm} \label{thm3}
\begin{equation} \label{eqES}
E_{s}(R, D)
\; = \; \sup_{0 \, \leq \,\rho \,\leq \,1} \; \Bigg\{ -\inf_{s\,\geq\,0} \; \ln \; \sum_{x}P(x)\Bigg[\sum_{\hat{x}}Q(\hat{x})e^{-s[d(x,\, \hat{x})-D]}\Bigg]^{\rho} \; - \; \rho R \Bigg\}.
\end{equation}
\end{thm}
\begin{proof}
\begin{align}
& \min_{T(x)} \; \Big\{ D(T \| P) \;\; + \;\; {\big| R(T, Q, D) \; - \; R   \big|\mathstrut}_{}^{+} \Big\}
\; \overset{(a)}{=} \;
\sup_{0 \, \leq \,\rho \,\leq \,1} \; \min_{T(x)} \;\big\{ D(T \| P) \; + \; \rho\big[R(T, Q, D) \; - \; R\big]\big\}
\nonumber \\
& \overset{(b)}{=} \;
\sup_{0 \, \leq \,\rho \,\leq \,1} \; \min_{T(x)} \;\Bigg\{ \sum_{x}T(x)\ln\frac{T(x)}{P(x)} \; + \; \rho\bigg[\sup_{s\,\geq\,0} \; \bigg\{-\sum_{x}T(x) \ln \sum_{\hat{x}}Q(\hat{x})e^{-s[d(x,\, \hat{x})-D]}\bigg\} \; - \; R\bigg]\Bigg\}
\nonumber \\
& = \;
\sup_{0 \, \leq \,\rho \,\leq \,1} \; \min_{T(x)} \;\sup_{s\,\geq\,0} \; \Bigg\{ \sum_{x}T(x)\ln\frac{T(x)}{P(x)} \; + \; \rho\bigg[-\sum_{x}T(x) \ln \sum_{\hat{x}}Q(\hat{x})e^{-s[d(x,\, \hat{x})-D]} \; - \; R\bigg]\Bigg\}
\nonumber \\
& \overset{(c)}{=} \;
\sup_{0 \, \leq \,\rho \,\leq \,1} \; \sup_{s\,\geq\,0} \; \min_{T(x)} \; \Bigg\{ \sum_{x}T(x)\ln\frac{T(x)}{P(x)} \; + \; \rho\bigg[-\sum_{x}T(x) \ln \sum_{\hat{x}}Q(\hat{x})e^{-s[d(x,\, \hat{x})-D]} \; - \; R\bigg]\Bigg\}
\nonumber \\
& = \;
\sup_{0 \, \leq \,\rho \,\leq \,1} \; \sup_{s\,\geq\,0} \; \min_{T(x)} \; \Bigg\{ \sum_{x}T(x)\ln\frac{T(x)}{P(x)\big[\sum_{\hat{x}}Q(\hat{x})e^{-s[d(x,\, \hat{x})-D]}\big]^{\rho}} \; - \; \rho R \Bigg\}
\nonumber \\
& = \;
\sup_{0 \, \leq \,\rho \,\leq \,1} \; \sup_{s\,\geq\,0} \; \Bigg\{ -\ln \sum_{x}P(x)\Bigg[\sum_{\hat{x}}Q(\hat{x})e^{-s[d(x,\, \hat{x})-D]}\Bigg]^{\rho} \; - \; \rho R \Bigg\},
\nonumber
\end{align}
where ($a$) is by (\ref{eqEs}),
in ($b$) we insert the identity (\ref{eqRTQD}) for $R(T, Q, D)$,
and ($c$) follows by the minimax theorem\footnote{Alternatively, the equality can be shown by substituting (\ref{eqRTQDminsup}) for $R(T,Q,D)$ and equating a convex function with its lower convex envelope.},
for the objective function which is convex ($\cup$) in $T(x)$ and concave\footnote{The concavity ($\cap$) is apparent from (\ref{eqWs}), where the function of $s$ is expressed as a {\em minimum} of affine functions of $s$.} ($\cap$) in $s$.
\end{proof}

\newpage
{\em Discussion:}

Let $T_{\rho}$ denote the unique solution of the minimum
\begin{displaymath}
\min_{T(x)} \big\{D(T \| P) \; + \; \rho R(T, Q, D)\big\} \; = \; D(T_{\rho} \| P) \; + \; \rho R(T_{\rho}, Q, D),
\end{displaymath}
for $\rho \, \geq \, 0$,
and define
\begin{displaymath}
R_{\rho} \; \triangleq \; R(T_{\rho}, Q, D), \;\;\;\;\;\; \rho \; \geq \; 0.
\end{displaymath}
Clearly, $T_{0} \, = \, P$, and consequently, by our definition, $R_{0} \, = \, R(P, Q, D)$.
However, note, that $\lim_{\rho \, \rightarrow \, 0}R_{\rho}$
is not necessarily equal to $R_{0}$. In general, it is less than or equal:
\begin{displaymath}
\lim_{\rho \, \rightarrow \, 0}R_{\rho} \; \leq \; R_{0} \; = \; R(P,Q,D).
\end{displaymath}
The inequality arises when $R(P, Q, D) \, = \, +\infty$ and $\lim_{\rho \, \rightarrow \, 0}R_{\rho}$ is still finite\footnote{Note that $R(T, Q, D)$ cannot ``diverge'' to infinity, as a function of $T$, since it is bounded when finite, as divergence is bounded.}.
In this case the exponent $E_{s}(R, D)$ does not decrease all the way to zero,
as $R$ increases,
but stays strictly above zero,
at the height
\begin{displaymath}
\lim_{\rho \, \rightarrow \, 0} D(T_{\rho} \| P)
\; = \; \min_{T(x): \;\; R(T, Q, D) \; < \; +\infty} D(T \| P)
\; > \; 0.
\end{displaymath}
In this particular case,
each one of the straight lines
\begin{displaymath}
D(T_{\rho} \| P) \; + \; \rho \big[R(T_{\rho}, Q, D) \; - \; R\big], \;\;\;\;\;\; \rho \; > \; 0,
\end{displaymath}
touches the curve $E_{s}(R)$, except for the line of slope zero: $E \, = \, D(T_{0} \| P)$, which is equal to zero for all $R$
and runs strictly below $E_{s}(R)$. The range of $D$, for which this behavior occurs, is given by the following

{\em Proposition 1:}
\begin{displaymath}
+\infty \; > \; \lim_{R \, \rightarrow \, \infty} E_{s}(R, D) \; > \; 0
\;\;\;\;\;\; \Longleftrightarrow \;\;\;\;\;\;
\min_{x}\min_{\hat{x}} d(x,\, \hat{x}) \; \leq \; D \; < \; \sum_{x}P(x)\min_{\hat{x}} d(x,\, \hat{x}).
\end{displaymath}

\begin{proof}
Follows from the relations
\begin{align}
R(P, Q, D)\; = \; +\infty
\;\;\;\;\;\; & \Longleftrightarrow \;\;\;\;\;\;
D \; < \; \sum_{x}P(x)\min_{\hat{x}} d(x,\, \hat{x}), \nonumber \\
R(T, Q, D)\; = \; +\infty, \;\;\; \forall \; T
\;\;\;\;\;\; & \Longleftrightarrow \;\;\;\;\;\;
D \; < \; \min_{x}\min_{\hat{x}} d(x,\, \hat{x}).
\nonumber
\end{align}
\end{proof}

\begin{figure}[ht]
\begin{center}
\epsfig{file=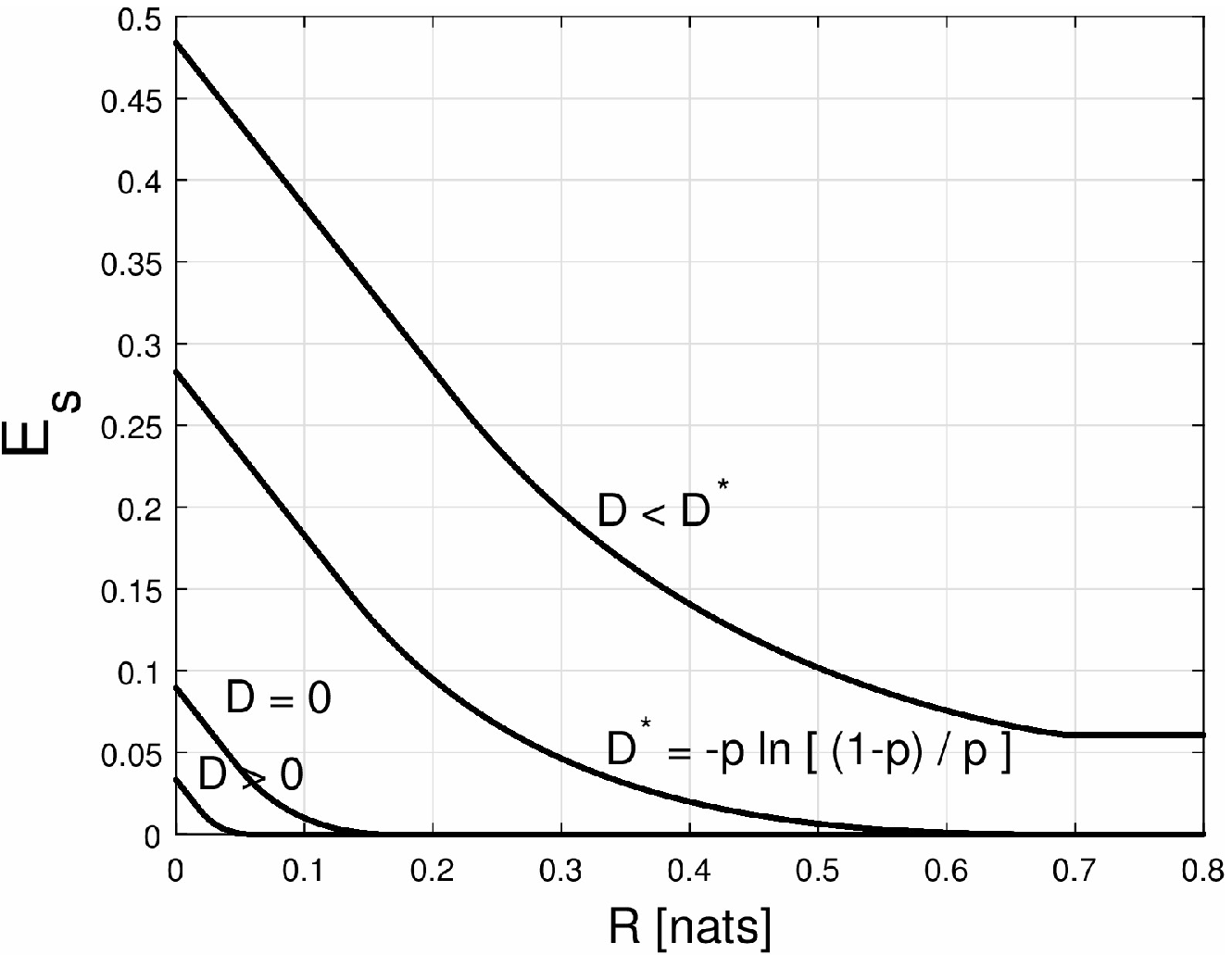, scale = 1}
\caption{
Encoding success exponent (\ref{eqES}) vs. $R$, for various $D=\{0.5, \, 0, -1, -1.7\}\cdot p\cdot \ln\frac{1-p}{p}$. Parameter $p=0.22$.\newline
Source: ${\cal X} = \{a,\,b,\,c,\,d\}$, $P(a) = P(d) = \frac{1-p}{2}$,
$P(b) = P(c) = \frac{p}{2}$.
Reproduction: $\hat{\cal X} \, = \, \{0,\,1\}$, $Q(0) = Q(1) = \frac{1}{2}$.\newline
Distortion measure: $d(a,\,0)=d(b,\,0)=d(c,\,1)=d(d,\,1)=0$, $\;d(a,\,1)=d(d,\,0)=\ln\frac{1-p}{p}$,
$\;d(b,\,1)=d(c,\,0)=-\ln\frac{1-p}{p}$.\newline
The lowest distortion for which $E_{s}(R,D)$ decreases to zero as $R$ increases: $D^{*} = \sum_{x}P(x)\min_{\hat{x}} d(x,\, \hat{x}) = -p\ln\frac{1-p}{p}$.\newline
As $D \searrow D_{\min} = \min_{x, \, \hat{x}} d(x,\, \hat{x}) = -\ln\frac{1-p}{p}$, the curves $E_{s}(R)$ tend to a ``$135^{\circ}$ angle'': $E_{s}(R, D)\nearrow \max\big\{\ln\frac{2}{p}-R,\;\ln\frac{1}{p}\big\}$. \newline
For $D < D_{\min}$ the encoding success exponent is $+\infty$.\newline
As $D \nearrow D_{\max} = \max_{x, \, \hat{x}} d(x,\, \hat{x}) = \ln\frac{1-p}{p}$, the curves $E_{s}(R)$ tend to $0$.\newline
This example corresponds also to the channel error exponent (with the same values of $D$) for the channel
${\rm BSC}(p)$ and $Q(0) = \frac{1}{2}$.}
\label{figExamples}
\end{center}
\end{figure}

\section{\bf Derivation of the explicit channel decoding error exponent} \label{S6}
We make the following substitutions in (\ref{eqES}):
\begin{align}
(x, y) \;\;\;\;\;\;\;\;\; & \longrightarrow \;\;\;\;\;\;\;\;\; x
\label{eqSubstitutions1} \\
Q(x)P(y \,|\, x) \;\;\;\;\;\;\;\;\; & \longrightarrow \;\;\;\;\;\;\;\;\; P(x)
\label{eqSubstitutions2} \\
d\big((x, y), \hat{x}\big) \; = \; \ln \frac{P(y \,|\, x)}{P(y \,|\, \hat{x})}
\;\;\;\;\;\;\;\;\; & \longrightarrow \;\;\;\;\;\;\;\;\;
d(x, \hat{x})
\label{eqSubstitutions3} \\
0 \;\;\;\;\;\;\;\;\; & \longrightarrow \;\;\;\;\;\;\;\;\; D
\label{eqSubstitutions4}
\end{align}
The result is the random coding exponent of Gallager \cite{Gallager73}:
\begin{align}
E_{e}(R) \; & = \; \sup_{0 \, \leq \,\rho \,\leq \,1} \; \Bigg\{ -\inf_{s\,\geq\,0} \;\ln \; \sum_{x, \,y}Q(x)P(y \,|\, x)\Bigg[\sum_{\hat{x}}Q(\hat{x})\left[\frac{P(y \,|\, x)}{P(y \,|\, \hat{x})}\right]^{-s}\Bigg]^{\rho} \; - \; \rho R \Bigg\}
\nonumber \\
& = \; \sup_{0 \, \leq \,\rho \,\leq \,1} \; \Bigg\{ -\inf_{s\,\geq\,0} \; \ln \; \sum_{x, \,y}Q(x)P^{1-s\rho}(y \,|\, x)\Bigg[\sum_{\hat{x}}Q(\hat{x})P^{s}(y \,|\, \hat{x})\Bigg]^{\rho} \; - \; \rho R \Bigg\}
\nonumber \\
& \!\overset{(*)}{=} \; \sup_{0 \, \leq \,\rho \,\leq \,1} \; \Bigg\{ -\ln \; \sum_{x, \, y}Q(x)P^{\frac{1}{1+\rho}}(y \,|\, x)\Bigg[\sum_{\hat{x}}Q(\hat{x})P^{\frac{1}{1+\rho}}(y \,|\, \hat{x})\Bigg]^{\rho} \; - \; \rho R \Bigg\}
\nonumber \\
& = \; \sup_{0 \, \leq \,\rho \,\leq \,1} \; \Bigg\{ -\ln \; \sum_{y}\Bigg[\sum_{\hat{x}}Q(\hat{x})P^{\frac{1}{1+\rho}}(y \,|\, \hat{x})\Bigg]^{1+\rho} \; - \; \rho R \Bigg\},
\label{eqEE}
\end{align}
where ($*$) follows by H\"older's inequality.\footnote{Together with the derivation of (\ref{eqES}) from (\ref{eqSuccess}), this is a lengthy derivation of (\ref{eqEE}). Its purpose is demonstration and a ``sanity check'': that the channel decoding error exponent is indeed a special case of the encoding success exponent for sources. A shorter straightforward derivation of the explicit channel decoding error exponent (\ref{eqEE}) can be made from $\displaystyle\min_{\substack{T(x,\,y),\,W(\hat{x} \,|\, x,\,y):\\d(T\circ W)\,\leq\,0}}\Big\{D(T \; \| \; Q \circ P) \;\; + \;\; {\big| D(T \circ W \; \| \; T \times Q) \; - \; R   \big|\mathstrut}_{}^{+}\Big\}$, which is equivalent to (\ref{eqError}) and uses the same distortion measure (\ref{eqDmeasure}).}

\section{\bf Derivation of the explicit encoding failure exponent} \label{S7}
Here we derive an explicit expression,
which does not always coincide with the encoding failure exponent (\ref{eqFailure}) for all $R$,
but gives the best convex ($\cup$) lower bound for (\ref{eqFailure}), for sufficiently lax distortion constraint $D$.

For the benefit of the next section, we give a number of lemmas first.
\begin{lemma} \label{lemma4}
{\em For any $\rho \geq 0$}
\begin{equation} \label{eqLine1}
E_{f}(R, D) \;\; \geq \;\; \min_{T(x)} \big\{ D(T \| P) \; - \; \rho \big[R(T, Q, D) \; - \; R\big]\big\},
\end{equation}
{\em with equality if $R \, = \, R(T_{\rho}, Q, D)$, where $T_{\rho}$ is a solution of the minimum:}
\begin{displaymath}
\min_{T(x)} \big\{ D(T \| P) \; - \; \rho R(T, Q, D)\big\} \;\; =\footnotemark \; D(T_{\rho} \| P) \; - \;  \rho R(T_{\rho}, Q, D).
\end{displaymath}\footnotetext{In general, $T_{\rho}$ may be not unique, and as a result, the value of $R(T_{\rho}, Q, D)$ may be not unique,
while only the difference $E_{0}(\rho) \, \triangleq \, D(T_{\rho} \|P) \, - \, \rho R(T_{\rho}, Q, D)$ is a function of $\rho$.}
\end{lemma}
\begin{proof}
\begin{align}
E_{f}(R, D) \; & \;\, = \;\; \min_{T(x): \;\; R(T, Q, D) \; \geq \; R} D(T \| P) \;\; \overset{(*)}{=} \;\; D\big(T(R) \; \| \; P\big)
\nonumber \\
& \overset{\rho \, \geq \, 0}{\geq} \;\;
D\big(T(R) \; \| \; P\big) \; - \; \rho\big[R\big(T(R), Q, D\big) \; - \; R\big]
\nonumber \\
& \;\, \geq \;\;
\min_{T(x)} \big\{ D(T \| P) \; - \; \rho\big[R(T, Q, D) \; - \; R\big]\big\}
\nonumber
\end{align}
\begin{align}
& \;\, = \;\;
D(T_{\rho} \|P) \; - \; \rho\big[R(T_{\rho}, Q, D) \; - \; R\big],
\;\;\;\;\;\;\;\;\;\;\;\; \forall \; \rho \; \geq \; 0,
\label{eqLowerBounds1}
\end{align}
where in ($*$) we assumed that the set $\{T(x): \; R(T, Q, D) \, \geq \, R\}$
is nonempty. Otherwise $E_{f}(R, D)$ is considered to be $+\infty$ and
any lower bound is valid.

If $R \, = \, R(T_{\rho}, Q, D)$,
then by (\ref{eqLowerBounds1}) we obtain for this $R$:
\begin{displaymath}
D\big(T(R) \; \| \; P\big) \;\; \overset{(\ref{eqLowerBounds1})}{\geq} \;\; D(T_{\rho} \|P) \; - \; \rho\big[R(T_{\rho}, Q, D) \; - \; R\big]
\;\; = \;\; D(T_{\rho} \|P)
\;\; \geq \;\; D\big(T(R) \; \| \; P\big),
\end{displaymath}
where the second inequality holds because $T_{\rho}$ satisfies (with equality) the minimization constraint $R$.
\end{proof}
\begin{lemma} \label{lemma5}
\footnote{This is the ``only if'' addition to the statement of Lemma~\ref{lemma4}. This lemma will be needed in an example only.}
{\em If}
\begin{displaymath} 
E_{f}(R, D) \;\; = \;\; \min_{T(x)} \big\{ D(T \| P) \; - \; \rho \big[R(T, Q, D) \; - \; R\big]\big\},
\end{displaymath}
{\em for some $\rho > 0$, then necessarily $R \, = \, R(T_{\rho}, Q, D)$ for some $T_{\rho}\,$, such that}
\begin{displaymath}
\min_{T(x)} \big\{ D(T \| P) \; - \; \rho R(T, Q, D)\big\} \;\; = \;\; D(T_{\rho} \| P) \; - \;  \rho R(T_{\rho}, Q, D).
\end{displaymath}
\end{lemma}
\begin{proof}
Since for $R \, \leq \, R(P, Q, D)$ the exponent $E_{f}(R, D)$ is zero,
by the lower bound (\ref{eqLine1}) from the previous lemma we conclude that here necessarily $R \, \geq \, R(P, Q, D)$.
Note also, that the condition of the lemma implies that $R \, \leq \, \max_{T(x)}R(T, Q, D)\, < \, +\infty$.
For $R \, \geq \, R(P, Q, D)$ we can write
\begin{align}
E_{f}(R, D) \; = \; \min_{T(x): \;\; R(T, Q, D) \; \geq \; R} D(T \| P) \; & = \; D\big(T(R) \; \| \; P\big) \nonumber \\
& \overset{(*)}{=} \; D\big(T(R) \; \| \; P\big) \; - \; \rho\big[R\big(T(R), Q, D\big) \; - \; R\big]
\nonumber \\
& \geq \; \min_{T(x)} \big\{ D(T \| P) \; - \; \rho \big[R(T, Q, D) \; - \; R\big]\big\},
\nonumber
\end{align}
where in ($*$) the difference
$\big[R\big(T(R), Q, D\big) -  R\big]$ cannot be positive in the case of $R \, \geq \, R(P, Q, D)$, and must be zero, because $D(T \| P)$ is strictly convex and $R(T, Q, D)$ is a continuous function of $T$.
When we have equality in the above, $T(R)$ is a solution of the last minimum, i.e. $T(R) = T_{\rho}$.
\end{proof}
\begin{lemma} \label{lemma6}
{\em If $\;\displaystyle\max_{T(x)}R(T, Q, D)\, < \, +\infty$, then}
\begin{equation} \label{eqEFCE}
\text{\em lower convex envelope}\;\big(E_{f}(R)\big) \; = \;
\sup_{\rho \, \geq \, 0} \;
\min_{T(x)} \big\{ D(T \| P) \; - \; \rho\big[R(T, Q, D) \; - \; R\big]\big\}.
\end{equation}
{\em If $\;\displaystyle\max_{T(x)}R(T, Q, D)\, = \, +\infty$, then the right-hand side expression gives zero, which is strictly lower than $E_{f}(R)\;$
for $R \, > \, R(P, Q, D)$.}
\end{lemma}
\begin{proof}
By (\ref{eqLine1}) of Lemma~\ref{lemma4} we have a lower bound:
\begin{equation} \label{eqEFLowerBound}
E_{f}(R, D) \; \geq \; \sup_{\rho \, \geq \, 0} \;
\min_{T(x)} \big\{ D(T \| P) \; - \; \rho\big[R(T, Q, D) \; - \; R\big]\big\}.
\end{equation}

Observe, that if
$\;\displaystyle\max_{T(x)}R(T, Q, D)\, = \, +\infty$,
then
the minimum in (\ref{eqEFLowerBound}) is $-\infty$
for all $\rho \, > \, 0$, and for $\rho \, = \, 0$ the minimum is $D(P \| P) \, = \, 0$.
We conclude, that if $\;\displaystyle\max_{T(x)}R(T, Q, D)\, = \, +\infty$,
then the lower bound (\ref{eqEFLowerBound}) is $0$. Apparently, this is not a tight lower bound if $R \, > \, R(P, Q, D)$,
i.e. no ``strong Lagrangian duality'' in this case.

On the other hand, if $\;\displaystyle\max_{T(x)}R(T, Q, D)\, < \, +\infty$, then
for any $\rho \, \geq \, 0$ there exists at least one (i.e. possibly not unique) {\em finite} $R_{\rho} \, = \, R(T_{\rho}, Q, D)$,
where by Lemma~\ref{lemma4} the curve $E_{f}(R)$ touches the straight line lower bound (\ref{eqLine1}).
We conclude, that the curve $E_{f}(R)$ touches the straight line lower bounds (\ref{eqLine1}) for each slope value $\rho \geq 0$.
Since $E_{f}(R)$ is a nondecreasing function, it follows, that the supremum of the straight lines over $\rho \geq 0$ (\ref{eqEFLowerBound}) is the lower convex envelope of $E_{f}(R)$.
\end{proof}
\begin{lemma} \label{lemma7}
\begin{displaymath}
D \; \geq \; \max_{x} \min_{\hat{x}} d(x, \hat{x})
\;\;\;\;\;\; \Longleftrightarrow \;\;\;\;\;\;
\max_{T(x)}R(T, Q, D)\, < \, +\infty.
\end{displaymath}
\end{lemma}
\begin{proof}
Observe, that if
there exists at least one $x$ for which
$\min_{\hat{x}} \{d(x,\, \hat{x})-D\} \, > \, 0$, then, using the expression for $R(T,Q,D)$ (\ref{eqRTQD}), we obtain
\begin{displaymath}
\max_{T(x)} R(T,Q,D) \;\; = \;\; \max_{x} \; \sup_{s\,\geq\,0} \; \bigg\{-\ln \sum_{\hat{x}}Q(\hat{x})e^{-s[d(x,\, \hat{x})-D]}\bigg\}
\;\; = \;\; +\infty.
\end{displaymath}
On the other hand, if for every $x$ holds $\;\min_{\hat{x}} \{d(x,\, \hat{x})-D\} \, \leq \, 0$,
then by the same expression we have $\;\displaystyle\max_{T(x)}R(T, Q, D)\, < \, +\infty$.
\end{proof}
\begin{lemma} \label{lemma8}
{\em For $\rho \, \geq \, 0$}
\begin{equation} \label{eqRho}
\min_{T(x)} \big\{ D(T \| P) \; - \; \rho R(T, Q, D)\big\}
\; = \;
-\sup_{s\,\geq\,0} \; \ln \; \sum_{x}P(x)\Bigg[\sum_{\hat{x}}Q(\hat{x})e^{-s[d(x,\, \hat{x})-D]}\Bigg]^{-\rho}.
\end{equation}
{\em If $D \; \geq \; \max_{x} \min_{\hat{x}} d(x, \hat{x})$, then the minimum on the LHS is achieved by some $T^{*}(x)$ if and only if}
\begin{equation} \label{eqTfin}
T^{*}(x) \; = \; \lim_{s\,\rightarrow \, s^{*}} T_{\rho, s}(x), \;\;\;\;\;\;
T_{\rho, s}(x) \; \propto \; P(x)\Bigg[\sum_{\hat{x}}Q(\hat{x})e^{-s[d(x,\, \hat{x})-D]}\Bigg]^{-\rho},
\end{equation}
{\em where $s^{*}$ is any limit (which may be finite or $+\infty$), which achieves the supremum on the RHS.}
\end{lemma}
\begin{proof}
\begin{align}
D(T \| P) \; - \; \rho R(T, Q, D) \;
& \overset{(a)}{=} \;
\sum_{x}T(x)\ln\frac{T(x)}{P(x)} \; - \; \rho\sup_{s\,\geq\,0} \; \bigg\{-\sum_{x}T(x) \ln \sum_{\hat{x}}Q(\hat{x})e^{-s[d(x,\, \hat{x})-D]}\bigg\}
\nonumber \\
& \overset{(b)}{=} \;
\inf_{s\,\geq\,0} \; \bigg\{ \sum_{x}T(x)\ln\frac{T(x)}{P(x)} \; + \; \rho\sum_{x}T(x) \ln \sum_{\hat{x}}Q(\hat{x})e^{-s[d(x,\, \hat{x})-D]}\bigg\}
\nonumber \\
& = \;
\lim_{s\,\rightarrow\,s^{*}(T)}\;\sum_{x}T(x)\ln\frac{T(x)}{P(x)\big[\sum_{\hat{x}}Q(\hat{x})e^{-s[d(x,\, \hat{x})-D]}\big]^{-\rho}}
\nonumber
\end{align}
\begin{align}
& =
\lim_{s\,\rightarrow\,s^{*}(T)} \left\{\sum_{x}T(x)\ln\frac{T(x)}{\displaystyle\frac{P(x)\big[\sum_{\hat{x}}Q(\hat{x})e^{-s[d(x,\, \hat{x})-D]}\big]^{-\rho}}
{\sum_{a}P(a)\big[\sum_{\hat{x}}Q(\hat{x})e^{-s[d(a,\, \hat{x})-D]}\big]^{-\rho}}}
- \ln \sum_{x}P(x)\Bigg[\sum_{\hat{x}}Q(\hat{x})e^{-s[d(x,\, \hat{x})-D]}\Bigg]^{-\rho}\right\}
\nonumber \\
& \geq \;
-\lim_{s\,\rightarrow\,s^{*}(T)}\;\ln \sum_{x}P(x)\Bigg[\sum_{\hat{x}}Q(\hat{x})e^{-s[d(x,\, \hat{x})-D]}\Bigg]^{-\rho}
\; \geq \;
-\lim_{s\,\rightarrow\,s^{*}}\;\ln \sum_{x}P(x)\Bigg[\sum_{\hat{x}}Q(\hat{x})e^{-s[d(x,\, \hat{x})-D]}\Bigg]^{-\rho}
\nonumber \\
& = \;
\inf_{s\,\geq\,0} \; \Bigg\{ -\ln \sum_{x}P(x)\Bigg[\sum_{\hat{x}}Q(\hat{x})e^{-s[d(x,\, \hat{x})-D]}\Bigg]^{-\rho}\Bigg\},
\label{eqInf}
\end{align}
where in ($a$) we use (\ref{eqRTQD}) for $R(T, Q, D)$, and ($b$) holds for $\rho \, \geq \, 0$.
Observe, that both inequalities above become equalities if
\begin{equation} \label{eqTopt}
T(x) \; = \; \lim_{s\,\rightarrow \, s^{*}} T_{\rho, s}(x), \;\;\;\;\;\;
T_{\rho, s}(x) \; \propto \; P(x)\Bigg[\sum_{\hat{x}}Q(\hat{x})e^{-s[d(x,\, \hat{x})-D]}\Bigg]^{-\rho},
\end{equation}
where $s^{*}$ is a limit, achieving the infimum in (\ref{eqInf}).
We conclude, that the infimum (\ref{eqInf}) coincides with the minimum over $T(x)$,
i.e. obtain the desired result (\ref{eqRho}).
Note, however, that (\ref{eqTopt}) is not a {\em necessary} condition in the case when the infimum (\ref{eqInf}) is $-\infty$.

Note further, that if $D \; \geq \; \max_{x} \min_{\hat{x}} d(x, \hat{x})$,
then (by Lemma~\ref{lemma7}) $\max_{T(x)}R(T, Q, D)\, < \, +\infty$ and the infimum in (\ref{eqInf}) accordingly must be finite (not $-\infty$) for any $\rho \, \geq \, 0$.
In this case, the lower bound (\ref{eqInf}) is attained {\em if and only if} $\;T(x)$ is given by (\ref{eqTopt}).
\end{proof}

Now, by Lemma~\ref{lemma6}, Lemma~\ref{lemma7}, and identity (\ref{eqRho}) of Lemma~\ref{lemma8} we have the following
\begin{thm} \label{thm4}
{\em For distortion constraint $\;D \; \geq \; \max_{x} \min_{\hat{x}} d(x, \hat{x})$,}
\begin{equation} \label{eqEF}
\text{\em lower convex envelope}\;\big(E_{f}(R)\big) \; = \;
\sup_{\rho \, \geq \, 0} \; \Bigg\{-\sup_{s\,\geq\,0} \; \ln \; \sum_{x}P(x)\Bigg[\sum_{\hat{x}}Q(\hat{x})e^{-s[d(x,\, \hat{x})-D]}\Bigg]^{-\rho} \; + \; \rho R\Bigg\}.
\end{equation}
{\em For $\;D \; < \; \max_{x} \min_{\hat{x}} d(x, \hat{x})$, the right-hand side expression gives zero, which is strictly lower than $E_{f}(R)$,
if $R \, > \, R(P, Q, D)$.}
\end{thm}
\begin{figure}[ht]
\begin{center}
\epsfig{file=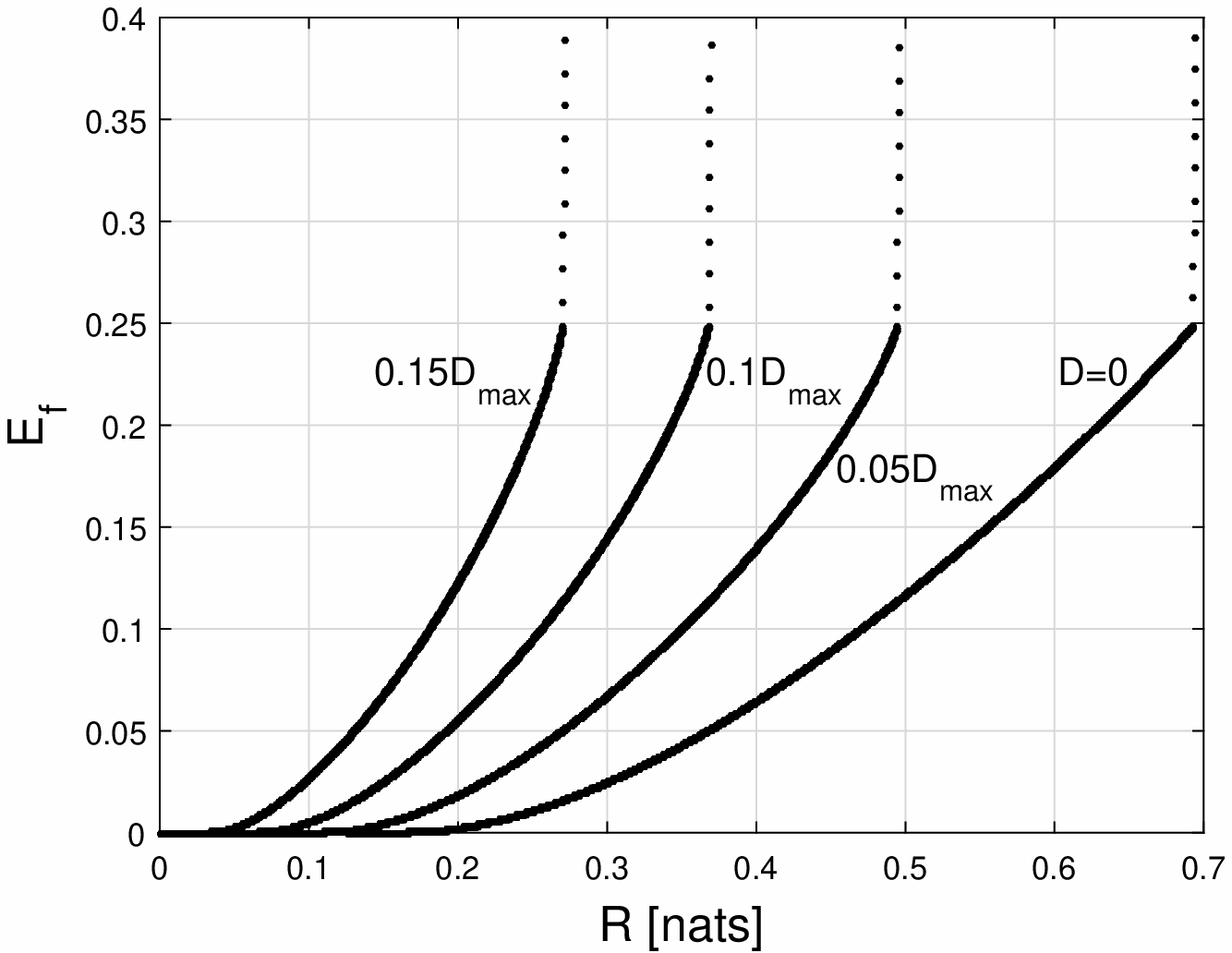, scale = 1}
\caption{Encoding failure exponent (\ref{eqEF}) vs. $R$, for $D\,=\,\{0, \, 0.05, 0.10, 0.15\}\cdot \ln\frac{1-p}{p} \,\geq \, 0 \, = \, \max_{x} \min_{\hat{x}} d(x, \hat{x})$.
Parameter $p=0.22$.\newline
Source: ${\cal X} = \{a,\,b,\,c,\,d\}$, $P(a) = P(d) = \frac{1-p}{2}$,
$P(b) = P(c) = \frac{p}{2}$.
Reproduction: $\hat{\cal X} \, = \, \{0,\,1\}$, $Q(0) = Q(1) = \frac{1}{2}$.\newline
Distortion measure: $d(a,\,0)=d(b,\,0)=d(c,\,1)=d(d,\,1)=0$, $\;d(a,\,1)=d(d,\,0)=\ln\frac{1-p}{p}$,
$\;d(b,\,1)=d(c,\,0)=-\ln\frac{1-p}{p}$.\newline
As $R \, \nearrow \, \max_{T(x)}R(T,Q,D)$, each curve $E_{f}(R) \, \nearrow \, \log\frac{1}{1-p} \, \approx \, 0.249$.\newline
For $R \, > \, \max_{T(x)}R(T,Q,D)$ the encoding failure exponent is $+\infty$.\newline
For $D \, \geq \, D_{max} \, = \, \max_{x,\,\hat{x}}d(x,\hat{x}) \, = \, \ln\frac{1-p}{p}$ the encoding failure exponent is $+\infty$.\newline
This example corresponds also to the channel correct-decoding exponent (with the same values of $D$) for the channel
${\rm BSC}(p)$ and $Q(0) = \frac{1}{2}$.}
\label{figExamples2}
\end{center}
\end{figure}

\begin{figure}[ht]
\begin{center}
\epsfig{file=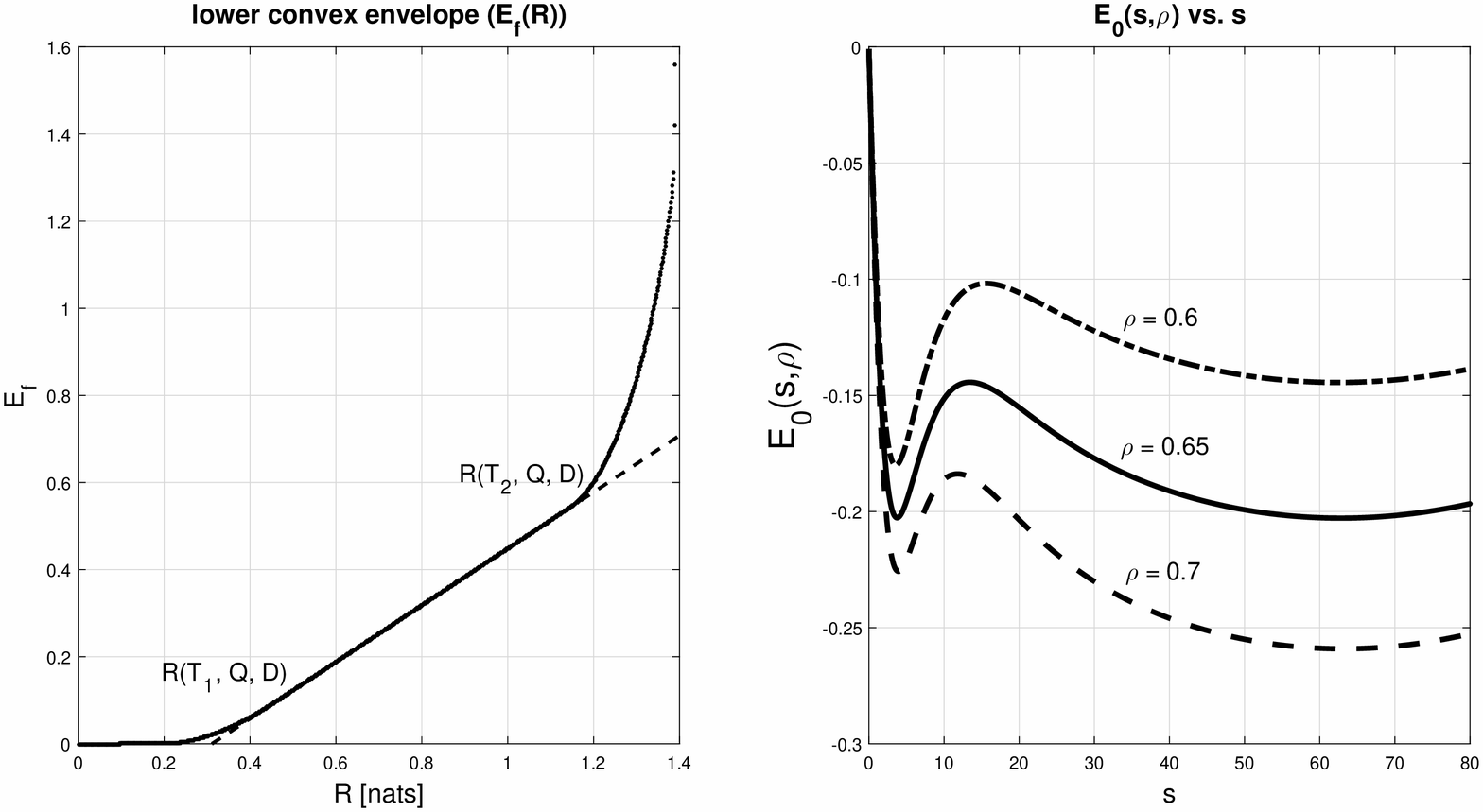, scale = 0.5}
\caption[]{Encoding failure exponent: an example where $E_{f}(R)$ (\ref{eqFailure}) does not always coincide with its lower convex envelope (\ref{eqEF}).\newline
Source / reproduction: $|{\cal X}| = |\hat{\cal X}| = 5$, $P = [0.2923, 0.0142, 0.2673, 0.3210, 0.1051]$,
$Q = [0.2573, 0.0908, 0.2437, 0.0294, 0.3787]$.\newline
Distortion measure and constraint: $\; d(x,\hat{x}) - D \; = \; \left[\begin{array}{r r r r r}
-0.0799 &  0.1580 &  0.0425 &  0.0673 & -0.3449 \\
 0.0815 &  0.2024 & -0.1511 &  0.1030 &  0.4020 \\
 0.0147 & -0.0079 &  0.7994 &  0.6861 &  0.1450 \\
 0.8545 &  0.9160 &  0.9066 &  0.5624 & -0.0015 \\
-0.2179 & -0.4107 & -0.0435 & -0.2367 & -0.2594
\end{array}\right]$.\newline
Observe, that the difference $d(x,\hat{x}) - D$ satisfies the condition of Theorem~\ref{thm4} (each row has negative values).
Therefore (\ref{eqEF}) holds.\newline
The lower convex envelope (\ref{eqEF}) is depicted in the left graph. The envelope has a segment with constant slope $\rho = 0.65$.\newline
The right graph shows, that there are exactly $2$ different values $s^{*}$, achieving the supremum inside $-\ln$ in (\ref{eqEF}) for $\rho = 0.65$:\newline
$E_{0}(\rho = 0.65) \, = \, \inf_{s\,\geq\,0}E_{0}(s,\rho = 0.65)\, = \, E_{0}(s_{i}^{*},\,\rho = 0.65) \, = \, D(T_{i} \| P) - \rho R(T_{i}, Q, D), \;\;\; i \, = \, 1,\,2$.\newline
Note from the right graph, that there are $2$ ``modes'' (local minima) in the curves of $E_{0}(s,\rho)$ vs. $s$, for each $\rho$.\newline
The left mode provides a unique solution
$R_{\rho} \, = \, R(T_{\rho, s^{*}}, Q, D) \, < \, R(T_{1}, Q, D)$ for $\rho < 0.65$.\newline
The right mode provides a unique solution
$R_{\rho} \, = \, R(T_{\rho, s^{*}}, Q, D) \, > \, R(T_{2}, Q, D)$ for $\rho > 0.65$.\newline
$E_{f}(R)$ (\ref{eqFailure}) must run {\em strictly above} its lower convex envelope (\ref{eqEF}) for
$R(T_{1}, Q, D)\, < \, R \, < \, R(T_{2}, Q, D)$, according to Lemmas~\ref{lemma8} and~\ref{lemma5}.
}
\label{figExamples3}
\end{center}
\end{figure}

\section{\bf Derivation of the explicit channel correct-decoding exponent} \label{S8}
We would like to show that the channel correct-decoding exponent (\ref{eqCorrect}) is equivalent to (\ref{eqCorrectEquiv}),
and find an explicit expression for it.

The expression (\ref{eqCorrect}), which is written with the help of the Csisz\'ar-K\"orner
style brackets $|\cdot|^{+}$ for compactness, translates into the minimum between two exponents:
\begin{align}
& \min_{T(x, \, y)} \; \Big\{ D(T \; \| \; Q \circ P) \;\; + \;\; {\big| R \; - \; R(T, Q, 0)\big|\mathstrut}_{}^{+} \Big\}
\; = \; \min \; \big\{E_{A}(R), \;\; E_{B}(R)\big\} \; =
\label{eqChannelE1E2} \\
& \min \; \bigg\{\min_{T(x,\,y): \;\; R(T, \,Q, \,0) \; \geq \; R} \; D(T \; \| \; Q \circ P), \;\;\;
\min_{T(x,\,y): \;\; R(T, \,Q, \,0) \; \leq \; R} \; \big\{ D(T \; \| \; Q \circ P) \; + \; R \; - \; R(T, Q, 0) \big\}
\bigg\}.
\nonumber
\end{align}
Note, that the left exponent $E_{A}(R)$ is the same as (\ref{eqFailure}), after we make the substitutions (\ref{eqSubstitutions1})-(\ref{eqSubstitutions4}).
Therefore, in order to characterize this exponent, we can use Lemmas~\ref{lemma4}, \ref{lemma6}, \ref{lemma7}, \ref{lemma8}, with the substitutions.

With the distortion measure (\ref{eqDmeasure}) we have
\begin{align}
\max_{x, \,y} \;\min_{\hat{x}}\; d\big((x, y), \hat{x}\big) \;\; & = \;\; \max_{x} \; \max_{y} \; \min_{\hat{x}}\; \ln \frac{P(y \,|\, x)}{P(y \,|\, \hat{x})}
\nonumber \\
& \leq \;\; \max_{x} \; \max_{y} \;\;\;\;\;\;\;\;\, \ln \frac{P(y \,|\, x)}{P(y \,|\, x)} \;\; = \;\; 0
\;\; = \;\; \;\;\;\;\;\;\;\; \max_{y} \; \min_{\hat{x}}\; \ln \frac{P(y \,|\, \hat{x})}{P(y \,|\, \hat{x})}
\nonumber \\
& \;\;\;\;\;\;\;\;\;\;\;\;\;\;\;\;\;\;\;\;\;\;\;\;\;\;\;\;\;\;\;\;\;\;\;\;\;\;\;\;\;\;\;\;\;\;\;\;\;\;\;\;\;\;\;\;\;\;\;
\leq \;\; \max_{x} \; \max_{y} \; \min_{\hat{x}}\; \ln \frac{P(y \,|\, x)}{P(y \,|\, \hat{x})},
\nonumber \\
\max_{x, \,y} \;\min_{\hat{x}}\; d\big((x, y), \hat{x}\big) \;\; & = \;\; 0,
\label{eqDistorionZero}
\end{align}
which is precisely the condition of Lemma~\ref{lemma7} with $D\,=\,0$ (satisfied with equality).
Therefore, by Lemma~\ref{lemma7} and Lemma~\ref{lemma6}, for the left exponent $E_{A}(R)$ we can write
\begin{equation} \label{eqLowerConvexEnvelopeE1}
\text{\em lower convex envelope}\;\big(E_{A}(R)\big) \; = \;
\sup_{\rho \, \geq \, 0} \;
\min_{T(x, \, y)} \big\{ D(T \; \| \; Q \circ P) \; - \; \rho\big[R(T, Q, 0) \; - \; R\big]\big\}.
\end{equation}

Similarly to the case of the encoding success (and the channel decoding error) exponent,
in order to compare between $E_{A}(R)$ and $E_{B}(R)$,
it is useful to consider the following function of $T$
(this time with a minus before $R(T,Q,0)$):
\begin{displaymath}
D(T \; \| \; Q \circ P) \; - \; R(T, Q, 0).
\end{displaymath}
Its minimum is given by identity (\ref{eqRho}) of Lemma~\ref{lemma8}:
\begin{align}
\min_{T(x, \, y)} \big\{ D(T \; \| \; Q \circ P) \; - \; R(T, Q, 0)\big\}
& \;\; \overset{(\ref{eqRho})}{=} \;\;
-\sup_{s\,\geq\,0} \; \ln \; \sum_{x, \, y}Q(x)P(y \,|\, x)\Bigg[\sum_{\hat{x}}Q(\hat{x})\left[\frac{P(y \,|\, x)}{P(y \,|\, \hat{x})}\right]^{-s}\Bigg]^{-1}
\nonumber \\
& \;\;\, = \;\;
-\sup_{s\,\geq\,0} \; \ln \; \sum_{y}\bigg[\sum_{x}
\frac
{Q(x)P^{s}(y \,|\, x)}
{\sum_{\hat{x}}Q(\hat{x})P^{s}(y \,|\, \hat{x})}
P(y \,|\, x)\bigg]
\nonumber \\
& \;\;\, \geq \;\;
-\ln \; \sum_{y} \max_{x}P(y \,|\, x)
\nonumber \\
& \;\;\, = \;\;
-\lim_{s\,\rightarrow\,+\infty} \; \ln \; \sum_{y}\bigg[\sum_{x}
\frac
{Q(x)P^{s}(y \,|\, x)}
{\sum_{\hat{x}}Q(\hat{x})P^{s}(y \,|\, \hat{x})}
P(y \,|\, x)\bigg]
\nonumber \\
& \;\;\, \geq \;\;
-\sup_{s\,\geq\,0} \; \ln \; \sum_{y}\bigg[\sum_{x}
\frac
{Q(x)P^{s}(y \,|\, x)}
{\sum_{\hat{x}}Q(\hat{x})P^{s}(y \,|\, \hat{x})}
P(y \,|\, x)\bigg].
\nonumber
\end{align}
That is, the supremum is achieved 
when $s\,\rightarrow\,+\infty$:
\begin{align}
\min_{T(x, \, y)} \big\{ D(T \; \| \; Q \circ P) \; - \; R(T, Q, 0)\big\}
& \;\;\, = \;\;
-\lim_{s\,\rightarrow\,+\infty} \; \ln \; \sum_{x, \, y}Q(x)P(y \,|\, x)\Bigg[\sum_{\hat{x}}Q(\hat{x})\left[\frac{P(y \,|\, x)}{P(y \,|\, \hat{x})}\right]^{-s}\Bigg]^{-1}
\nonumber \\
& \;\;\, = \;\;
D(T_{1} \; \| \; Q \circ P) \; - \;  R(T_{1},\,Q, \,0),
\label{eqDRTQD1}
\end{align}
where by (\ref{eqTfin})
\begin{displaymath}
T_{1}(x, y) \;\; \propto \;\; \lim_{s\,\rightarrow \, +\infty} Q(x)P(y \,|\, x)\Bigg[\sum_{\hat{x}}Q(\hat{x})\left[\frac{P(y \,|\, x)}{P(y \,|\, \hat{x})}\right]^{-s}\Bigg]^{-1}.
\end{displaymath}
By Lemma~\ref{lemma4} we conclude, that $E_{A}(R)$ touches the line
\begin{displaymath}
\min_{T(x, \, y)} \big\{D(T \; \| \; Q \circ P) \; - \; \big[R(T, Q, 0) \; - \; R\big]\big\}
\end{displaymath}
at $R_{1} \, = \, R(T_{1}, Q, 0)$.

For $0 \, \leq \, \rho \, < \, 1$, by Lemma~\ref{lemma8} we obtain
\begin{align}
\min_{T(x, \, y)} \big\{ D(T \; \| \; Q \circ P) \; - \; \rho R(T, Q, 0)\big\}
\; & = \;
-\ln \; \sup_{s\,\geq\,0} \; \sum_{x, \, y}Q(x)P(y \,|\, x)\Bigg[\sum_{\hat{x}}Q(\hat{x})\left[\frac{P(y \,|\, x)}{P(y \,|\, \hat{x})}\right]^{-s}\Bigg]^{-\rho}
\nonumber \\
& = \;
-\ln \; \sup_{s\,\geq\,0} \; \sum_{x, \, y}Q(x)P^{1 + s\rho}(y \,|\, x)\Bigg[\sum_{\hat{x}}Q(\hat{x})P^{s}(y \,|\, \hat{x})\Bigg]^{-\rho}
\nonumber \\
& \!\overset{(*)}{=} \;
-\ln \; \sum_{x, \, y}Q(x)P^{\frac{1}{1-\rho}}(y \,|\, x)\Bigg[\sum_{\hat{x}}Q(\hat{x})P^{\frac{1}{1-\rho}}(y \,|\, \hat{x})\Bigg]^{-\rho}
\nonumber \\
& = \;
-\ln \; \sum_{y}\Bigg[\sum_{\hat{x}}Q(\hat{x})P^{\frac{1}{1-\rho}}(y \,|\, \hat{x})\Bigg]^{1-\rho},
\label{eqExplicitrho}
\end{align}
where ($*$) follows by H\"older's inequality.
For each $0 \, \leq \, \rho \, < \, 1$
\begin{equation} \label{eqDRTQD}
\min_{T(x, \, y)} \big\{ D(T \; \| \; Q \circ P) \; - \; \rho R(T, Q, 0)\big\}
\; = \;
D(T_{\rho} \; \| \; Q \circ P) \; - \; \rho R(T_{\rho}, \,Q, \,0),
\end{equation}
where by (\ref{eqTfin})
\begin{displaymath}
T_{\rho}(x, y) \;\; \propto \;\; Q(x)P(y \,|\, x)\Bigg[\sum_{\hat{x}}Q(\hat{x})\left[\frac{P(y \,|\, x)}{P(y \,|\, \hat{x})}\right]^{-\frac{1}{1-\rho}}\Bigg]^{-\rho}.
\end{displaymath}
By Lemma~\ref{lemma4}, it appears, that $E_{A}(R)$ touches each line
\begin{displaymath}
\min_{T(x, \, y)} \big\{D(T \; \| \; Q \circ P) \; - \; \rho \big[R(T, Q, 0) \; - \; R\big]\big\}, \;\;\;\;\;\; 0 \, \leq \, \rho \, \leq \, 1,
\end{displaymath}
at $R_{\rho} \, = \, R(T_{\rho}, Q, 0)$.
Since $E_{A}(R)$ is nondecreasing,
by continuity of $R(T_{\rho}, Q, 0)$ as a function of $\rho$
for $0 \, \leq \, \rho \, \leq \, 1$,
we conclude that
\begin{equation} \label{eqConvexE1}
E_{A}(R) \; = \;
\sup_{0 \, \leq \, \rho \, \leq \, 1} \;
\min_{T(x, \, y)} \big\{ D(T \; \| \; Q \circ P) \; - \; \rho\big[R(T, Q, 0) \; - \; R\big]\big\},
\;\;\;\;\;\; R \, \leq \, R_{1}.
\end{equation}

Now we proceed to the right exponent $E_{B}(R)$,
which is lower-bounded as follows:
\begin{align}
E_{B}(R) \; & = \; \min_{T(x, \, y): \;\; R(T, \,Q, \,0) \; \leq \; R} \; \big\{ D(T \; \| \; Q \circ P) \; + \; R \; - \;  R(T, Q, 0) \big\} \nonumber \\
& \geq \;
\sup_{\rho \,\geq \,0} \;\; \min_{T(x, \, y): \;\;R(T, \,Q, \,0) \; \leq \; R} \;\big\{ D(T \; \| \; Q \circ P) \; + \; R \; - \; R(T, Q, 0) \; + \; \rho\big[R(T, Q, 0) \; - \; R\big]\big\}
\nonumber \\
& \geq \;
\sup_{\rho \,\geq \,0} \;\;\;\;\;\;\;\;\;\;\; \min_{T(x, \, y)}
\;\;\;\;\;\;\;\;\;\;
\big\{ D(T \; \| \; Q \circ P) \; + \; R \; - \; R(T, Q, 0) \; + \; \rho\big[R(T, Q, 0) \; - \; R\big]\big\}
\nonumber \\
& = \;
\sup_{\rho \,\geq \,0} \;\;\;\;\;\;\;\;\;\;\; \min_{T(x, \, y)}
\;\;\;\;\;\;\;\;\;\;
\big\{ D(T \; \| \; Q \circ P) \; - \; (1-\rho)\big[R(T, Q, 0) \; - \; R\big]\big\}
\nonumber \\
& \geq \;
\sup_{0 \, \leq \,\rho \,\leq \,1} \;\;\;\;\;\;\; \min_{T(x, \, y)}
\;\;\;\;\;\;\;\;\;\;
\big\{ D(T \; \| \; Q \circ P) \; - \; (1-\rho)\big[R(T, Q, 0) \; - \; R\big]\big\}
\nonumber \\
& = \;
\sup_{0 \, \leq \,\rho \,\leq \,1} \;\;\;\;\;\;\; \min_{T(x, \, y)}
\;\;\;\;\;\;\;\;\;\;
\big\{ D(T \; \| \; Q \circ P) \; - \; \rho\big[R(T, Q, 0) \; - \; R\big]\big\}.
\label{eqE2F}
\end{align}
Observe further, that for $R \geq R_{1}$
\begin{align}
E_{B}(R) \; & = \; \min_{T(x, \, y): \;\; R(T, \,Q, \,0) \; \leq \; R} \; \big\{ D(T \; \| \; Q \circ P) \; + \; R \; - \;  R(T, Q, 0) \big\}
\nonumber \\
& = \; D(T_{1} \; \| \; Q \circ P) \; + \; R \; - \;  R(T_{1}, Q, 0)
\nonumber \\
& = \; \min_{T(x, \, y)} \; \big\{ D(T \; \| \; Q \circ P) \; + \; R \; - \;  R(T, Q, 0) \big\}
\label{eqE2Line} \\
& \leq \;
\sup_{0 \, \leq \,\rho \,\leq \,1} \; \min_{T(x, \, y)} \;
\big\{ D(T \; \| \; Q \circ P) \; - \; \rho\big[R(T, Q, 0) \; - \; R\big]\big\}, \;\;\;\;\;\; R \; \geq \; R_{1}.
\label{eqE2Upper2}
\end{align}

Comparing (\ref{eqConvexE1}) with (\ref{eqE2F}), for $R \, \leq \, R_{1}$, and (\ref{eqLowerConvexEnvelopeE1}) with (\ref{eqE2F})-(\ref{eqE2Upper2}), for $R \, \geq \, R_{1}$, respectively, we ascertain the validity of (\ref{eqCorrectEquiv}),
and obtain the explicit expression\footnote{Starting from (\ref{eqFailure})-(\ref{eqCorrect}), it is a lengthy derivation of the explicit channel correct-decoding exponent. Its purpose is to prove (\ref{eqCorrectEquiv}), which shows the relation to the encoding failure exponent for sources (\ref{eqFailure}),
and a sanity check of (\ref{eqCorrect}). A much shorter derivation of the explicit channel correct-decoding exponent can be made from an alternative expression: $\displaystyle\min_{T(x,\,y)}\Big\{D(T \; \| \; Q \,\circ \,P) \;\; + \;\; {\big| R \; - \; D(T \; \| \; Q \times [T]_{y})   \big|\mathstrut}_{}^{+}\Big\}$. This expression, unlike (\ref{eqCorrect}), leads to {\em convex} objective functions,
and has itself a simple derivation/explanation (alternative to the derivation of (\ref{eqCorrect})) as the channel correct-decoding exponent.}:
\begin{align}
E_{c}(R) \; & \; = \;\;
\sup_{0 \, \leq \, \rho \, \leq \, 1} \;
\min_{T(x, \, y)} \big\{ D(T \; \| \; Q \circ P) \; - \; \rho\big[R(T, Q, 0) \; - \; R\big]\big\}
\nonumber \\
& \overset{(*)}{=} \; \sup_{0 \, \leq \,\rho \, < \,1} \; \Bigg\{ -\ln \; \sum_{y}\Bigg[\sum_{\hat{x}}Q(\hat{x})P^{\frac{1}{1-\rho}}(y \,|\, \hat{x})\Bigg]^{1-\rho} \; + \; \rho R \Bigg\},
\label{eqEC}
\end{align}
where ($*$) follows from (\ref{eqDRTQD1})-(\ref{eqDRTQD}), the fact that $T_{\rho} \, \rightarrow\, T_{1}$, as $\rho \, \rightarrow\, 1$, and by continuity of $D(T \; \| \; Q \circ P)$ and $R(T, Q, 0)$, as functions of $T$.
This, together with the Gallager expression (\ref{eqEE}) forms a single convex ($\cup$) ``error/correct-decoding'' exponent curve:
\begin{displaymath}
E_{e-c}(R) \;\; = \;\;
\sup_{-1 \, \leq \,\rho \, < \,1\;\;} \; \Bigg\{ -\ln \; \sum_{y}\Bigg[\sum_{\hat{x}}Q(\hat{x})P^{\frac{1}{1-\rho}}(y \,|\, \hat{x})\Bigg]^{1-\rho} \; + \; \rho R \Bigg\}.
\end{displaymath}

\section{{\bf Extension of the channel decoding error exponent to arbitrary} \texorpdfstring{$D$}{\em D}} \label{S9}
The channel decoding error exponent (\ref{eqError}) can be written for arbitrary $D$ as
\begin{equation} \label{eqErrorD}
\lim_{n \, \rightarrow \, \infty} \; \left\{-\frac{1}{n}\ln P_{e}\right\} \; = \;
E_{e}(R, D)
\; \triangleq \;
\min_{T(x, \, y)} \; \Big\{ D(T \; \| \; Q \circ P) \;\; + \;\; {\big| R(T, Q, D) \; - \; R   \big|\mathstrut}_{}^{+} \Big\},
\end{equation}
where $R(T, Q, D)$ is determined with respect to the distortion measure $d\big((x, y), \hat{x}\big)$ (\ref{eqDmeasure}),
with the possible exception of $D \, = \, D_{\min}\, = \, \min_{(x, \, y), \, \hat{x}}d\big((x, y), \hat{x}\big)$,
when the RHS is a lower bound.
This exponent is exactly the same as the encoding success exponent (\ref{eqSuccess}), after we make the substitutions (\ref{eqSubstitutions1})-(\ref{eqSubstitutions3}).

Note, that the original exponent (\ref{eqSuccess}) corresponds to the ``encoding success'' condition
\begin{equation} \label{eqSuccessCondition}
\sum_{x, \,\hat{x}}T(x)W(\hat{x} \,|\, x)d(x, \, \hat{x}) \;\; \leq \;\; D,
\end{equation}
where the joint distribution $T(x)W(\hat{x} \,|\, x)$ {\em represents} the joint type of a source sequence $\bf x$ and a reproduction sequence $\hat{\bf x}_{m}$.
That is, the encoding success condition (\ref{eqSuccessCondition}) is an extension (to the set of all distributions) of the condition on the joint type
of sequences of length $n$:
\begin{displaymath}
\sum_{x, \,\hat{x}}P_{{\bf x}, \, \hat{\bf x}_{m}}(x, \,\hat{x})d(x, \, \hat{x}) \;\; \leq \;\; D.
\end{displaymath}
This condition, in turn, represents the encoding success event:
\begin{displaymath}
\big\{\exists \, m\,: \;\;\; d({\bf X}, \hat{\bf X}_{m}) \;\; \leq \;\; nD\big\}.
\end{displaymath}

It is obvious from the definition of the encoding success condition, that the decoding error exponent (\ref{eqErrorD})
corresponds to a decoding error condition:
\begin{align}
\sum_{x, \, y, \,\hat{x}}T(x, y)W(\hat{x} \,|\, x, y)d\big((x, y), \, \hat{x}\big) \;\; & \leq \;\; D \nonumber \\
\sum_{x, \, y, \,\hat{x}}T(x, y)W(\hat{x} \,|\, x, y)\ln \frac{P(y \,|\, x)}{P(y \,|\, \hat{x})} \;\; & \leq \;\; D
\nonumber \\
\sum_{x, \, y, \,\hat{x}}T(x, y)W(\hat{x} \,|\, x, y)\ln P(y \,|\, x) \;\; & \leq \;\; D \; + \;
\sum_{x, \, y, \,\hat{x}}T(x, y)W(\hat{x} \,|\, x, y)\ln P(y \,|\, \hat{x}),
\nonumber
\end{align}
where $T(x, y)W(\hat{x} \,|\, x, y)$ represents the joint type of a transmitted codeword ${\bf x}_{m}$, a received vector $\bf y$,
and a competing codeword ${\bf x}_{m'}$. The decoding error condition represents the decoding error event:
\begin{equation} \label{eqErrorEvent}
\left\{\exists \, m'\,\neq \,m\,: \;\;\; \ln\frac{P({\bf Y} \,|\, {\bf X}_{m})}{P({\bf Y} \,|\, {\bf X}_{m'})} \;\; \leq \;\; nD\right\}.
\end{equation}

A positive $D$ amounts to a stricter receiver, which requires a confidence distance greater than $nD$ between the log-likelihoods of the most likely codeword and the second most likely codeword, in order to make a decision. In this case, the decoding error event consists of an erasure and an undetected error.

A negative $D$ amounts to a list decoder. All codewords with log-likelihoods within a distance of less than $-nD$ from the most likely codeword are in the list.
A decoding error occurs when the transmitted codeword is not in the list.

\section{{\bf Explicit channel decoding error exponent with arbitrary} \texorpdfstring{$D$}{\em D}} \label{S10}
The substitutions (\ref{eqSubstitutions1})-(\ref{eqSubstitutions3}) into (\ref{eqES}) give an explicit form of (\ref{eqErrorD}):
\begin{equation} \label{eqDecErrorArbD}
E_{e}(R, D)
\; = \; \sup_{0 \, \leq \,\rho \,\leq \,1} \; \Bigg\{ -\inf_{s\,\geq\,0} \; \ln \; \sum_{x, \,y}Q(x)P(y \,|\, x)\Bigg[\sum_{\hat{x}}Q(\hat{x})\left[\frac{P(y \,|\, x)}{P(y \,|\, \hat{x})}\,e^{-D}\right]^{-s}\Bigg]^{\rho} \; - \; \rho R \Bigg\}.
\end{equation}

\section{{\bf Extension of the channel correct-decoding exponent to arbitrary} \texorpdfstring{$D$}{\em D}} \label{S11}
A natural extension is possible with respect to the decoding error event defined by (\ref{eqErrorEvent}).
In this case, the correct-decoding exponent is given by
\begin{equation} \label{eqCorrectExtended}
E_{c}^{*}(R, D) \; = \; \min_{T(x, \, y): \;\; R(T, \, Q, \, D) \; \geq \; R} \; D(T \; \| \; Q \circ P),
\end{equation}
with the possible exception of points of discontinuity of this function.
This exponent is exactly the same as the encoding failure exponent (\ref{eqFailure}), after we make the substitutions (\ref{eqSubstitutions1})-(\ref{eqSubstitutions3}),
also in the case $D=0$. The superscript $^{*}$ serves to indicate that this exponent is different from (\ref{eqCorrect}) or (\ref{eqCorrectEquiv}), for $D=0$,
as here the receiver declares an error also when there is only an equality in (\ref{eqErrorEvent}), i.e. no tie-breaking\footnote{This distinction is important in the case of the correct-decoding exponent, but not in the case of the decoding error exponent.}.

\section{{\bf Explicit channel correct-decoding exponent with arbitrary} \texorpdfstring{$D$}{\em D}} \label{S12}
Since (\ref{eqCorrectExtended}) is equivalent to the encoding failure exponent (\ref{eqFailure}), we can use Theorem~\ref{thm4} with substitutions (\ref{eqSubstitutions1})-(\ref{eqSubstitutions3}) and (\ref{eqDistorionZero}).
For distortion constraint $D \geq 0$:
\begin{align}
& \text{\em lower convex envelope}\;\big(E_{c}^{*}(R)\big) \; = \;
\nonumber \\
& \sup_{\rho \, \geq \, 0} \; \Bigg\{-\sup_{s\,\geq\,0} \; \ln \; \sum_{x, \,y}Q(x)P(y \,|\, x)\Bigg[\sum_{\hat{x}}Q(\hat{x})\left[\frac{P(y \,|\, x)}{P(y \,|\, \hat{x})}\,e^{-D}\right]^{-s}\Bigg]^{-\rho} \; + \; \rho R\Bigg\}.
\nonumber
\end{align}
For $D < 0$, the right-hand side expression gives zero, which is strictly lower than $E_{c}^{*}(R)$,
if $R > R(Q \circ P, \, Q, \,D)$.

\section{\bf Random coding error exponent of Forney's decoder (lower bound)} \label{S13}
In \cite{Forney68} the decoding error event, given that message $m$ is transmitted, is defined as
\begin{equation} \label{eqForneyErrorEvent}
{\cal E}_{m} \; \triangleq \; \bigg\{\ln\;\frac{P({\bf Y} \,|\, {\bf X\mathstrut}_{m})}{\sum_{m'\,\neq\,m}P({\bf Y} \,|\, {\bf X\mathstrut}_{m'})} \;\; < \;\; nD\bigg\}.
\end{equation}
This is different from the definition of the decoding error event (\ref{eqErrorEvent}) we have used
in order to establish duality between channel decoding and source encoding.
The sum over $m'$, which appears in Forney's metric (\ref{eqForneyErrorEvent}), and consists of an exponentially large number ($e^{nR}\,$) of terms, can be written equivalently as another sum --- of a polynomial number of terms --- over conditional types of different ${\bf X\mathstrut}_{m'}$ given a transmitted-received vector pair\footnote{Conditioning on a transmitted vector is not necessary for our derivation.} $({\bf X\mathstrut}_{m}, {\bf Y})$.
Denote these conditional types as ${P\mathstrut}_{\hat{\bf x} \, | \,{\bf x}, \, {\bf y}}^{}$. Then we can write the sum using indicator functions as
\begin{align}
\sum_{m'\,\neq\,m}P({\bf Y} \,|\, {\bf X\mathstrut}_{m'}) \;\; & = \;\;
\;\sum_{{P\mathstrut}_{\hat{\bf x} \, | \,{\bf x}, \, {\bf y}}^{}}
\;\sum_{m'\,\neq\,m} P({\bf Y} \,|\, {\bf X\mathstrut}_{m'})\cdot
\mathbbm{1}_{\displaystyle\big\{{\bf X\mathstrut}_{m'} \; \in \; T\big({P\mathstrut}_{\hat{\bf x} \, | \,{\bf x}, \, {\bf y}}^{}, \, {\bf X\mathstrut}_{m}, {\bf Y}\big)\big\}}(m')
\nonumber \\
& = \;\; \!\!\sum_{\;\;\;{P\mathstrut}_{\hat{\bf x} \, | \,{\bf x}, \, {\bf y}}^{}}
{e\mathstrut}_{}^{-nE({P\mathstrut}_{\hat{\bf x} \, | \,{\bf x}, \, {\bf y}}^{},\, {\bf X\mathstrut}_{m}, {\bf Y})}
\label{eqTypeExp} \\
& \leq \;\; \!\!\sum_{\;\;\;{P\mathstrut}_{\hat{\bf x} \, | \,{\bf x}, \, {\bf y}}^{}}
{e\mathstrut}_{}^{-nE_{\min}({\bf X\mathstrut}_{m}, {\bf Y})}
\;\; \leq \;\; {{(n\, + \, 1)\mathstrut}^{|{\cal X}|\cdot|{\cal X}|\cdot|{\cal Y}|}}\cdot
{e\mathstrut}_{}^{-nE_{\min}({\bf X\mathstrut}_{m}, {\bf Y})},
\label{eqForneySum}
\end{align}
where some of the exponents $E\big({P\mathstrut}_{\hat{\bf x} \, | \,{\bf x}, \, {\bf y}}^{}, {\bf X\mathstrut}_{m}, {\bf Y}\big)$ may have value $+\infty$ (when the corresponding conditional type class $T\big({P\mathstrut}_{\hat{\bf x} \, | \,{\bf x}, \, {\bf y}}^{}, \, {\bf X\mathstrut}_{m}, {\bf Y}\big)$ is not represented among ${\bf X\mathstrut}_{m'}$),
but not all the exponents are $+\infty$ at the same time,
and
$E_{\min}\big({\bf X\mathstrut}_{m}, {\bf Y}\big)\,\triangleq\,\min_{{P\mathstrut}_{\hat{\bf x} \, | \,{\bf x}, \, {\bf y}}^{}}E\big({P\mathstrut}_{\hat{\bf x} \, | \,{\bf x}, \, {\bf y}}^{}, {\bf X\mathstrut}_{m}, {\bf Y}\big)\,<\,+\infty$.
Note, that here the exponents $E\big({P\mathstrut}_{\hat{\bf x} \, | \,{\bf x}, \, {\bf y}}^{}, {\bf X\mathstrut}_{m}, {\bf Y}\big)$
and their minimum over conditional types $E_{\min}\big({\bf X\mathstrut}_{m}, {\bf Y}\big)$ are random variables,
also given $({\bf X\mathstrut}_{m}, {\bf Y})$.\footnote{It is convenient to think that $E\big({P\mathstrut}_{\hat{\bf x} \, | \,{\bf x}, \, {\bf y}}^{}, {\bf X\mathstrut}_{m}, {\bf Y}\big)$ is a function of all stochastic matrices ${P\mathstrut}_{\hat{\bf x} \, | \,{\bf x}, \, {\bf y}}^{}$ possible for a given block length $n$, regardless of the joint type of $\big({\bf X\mathstrut}_{m}, {\bf Y}\big)$. If a stochastic matrix ${P\mathstrut}_{\hat{\bf x} \, | \,{\bf x}, \, {\bf y}}^{}$ is not compatible with a certain joint type of $\big({\bf X\mathstrut}_{m}, {\bf Y}\big)$, then simply $T\big({P\mathstrut}_{\hat{\bf x} \, | \,{\bf x}, \, {\bf y}}^{}, \, {\bf X\mathstrut}_{m}, {\bf Y}\big)\,=\,\emptyset\;$ and $\;E\big({P\mathstrut}_{\hat{\bf x} \, | \,{\bf x}, \, {\bf y}}^{}, {\bf X\mathstrut}_{m}, {\bf Y}\big)\,=\,+\infty$, i.e. the corresponding term in (\ref{eqTypeExp}) is zero. Given a joint type of $\big({\bf X\mathstrut}_{m}, {\bf Y}\big)$ (and nothing else), the random variables $E\big({P\mathstrut}_{\hat{\bf x} \, | \,{\bf x}, \, {\bf y}}^{}, {\bf X\mathstrut}_{m}, {\bf Y}\big)$ and $E_{\min}\big({\bf X\mathstrut}_{m}, {\bf Y}\big)$ become independent of $\big({\bf X\mathstrut}_{m}, {\bf Y}\big)$, but their proper definitions still require a reference to $\big({\bf X\mathstrut}_{m}, {\bf Y}\big)$.}

Using types, we can upper-bound the ensemble average probability of error, given that message $m$ is transmitted, as follows
\begin{align}
& 
\Pr \, \{{\cal E}_{m}\}
\; \leq \sum_{\;\;\;{P\mathstrut}_{{\bf x}, \,{\bf y}, \, \hat{\bf x}}^{}}
\,\Pr\,\big\{({\bf X\mathstrut}_{m}, {\bf Y}) \, \in \, T({P\mathstrut}_{{\bf x}, \,{\bf y}}^{})\big\}
\,\cdot\,\Pr\,\Big\{E\big({P\mathstrut}_{\hat{\bf x} \, | \,{\bf x}, \, {\bf y}}^{}, {\bf X\mathstrut}_{m}, {\bf Y}\big) = E_{\min}\big({\bf X\mathstrut}_{m}, {\bf Y}\big), \; {\cal E}_{m} \,\Big|\, {P\mathstrut}_{{\bf x}, \,{\bf y}}^{}\Big\}
\nonumber \\
& \overset{(a)}{\leq} \sum_{\;\;\;{P\mathstrut}_{{\bf x}, \,{\bf y}, \, \hat{\bf x}}^{}}
\,\Pr\,\big\{({\bf X\mathstrut}_{m}, {\bf Y}) \, \in \, T({P\mathstrut}_{{\bf x}, \,{\bf y}}^{})\big\} \,\cdot\,
\Pr\,\bigg\{E\big({P\mathstrut}_{\hat{\bf x} \, | \,{\bf x}, \, {\bf y}}^{}, {\bf X\mathstrut}_{m}, {\bf Y}\big) = E_{\min}\big({\bf X\mathstrut}_{m}, {\bf Y}\big),
\nonumber \\
& \;\;\;\;\;\;\;\;\;\;\;\;\;\;\;\;
\;\;\;\;\;\;\;\;\;\;\;\;\;\;\;\;\;\;\;\;\;\;\;\;\;\;\;\;\;\;\;\;\,
\underbrace{E_{\min}\big({\bf X\mathstrut}_{m}, {\bf Y}\big)
\, \leq \, \frac{|{\cal X}||{\cal X}||{\cal Y}|\ln(n + 1)}{n} \, - \, \frac{\ln P({\bf Y} \,|\, {\bf X\mathstrut}_{m})}{n} \, + \, D}_{\supseteq \, {\cal E}_{m}}
 \; \bigg| \; {P\mathstrut}_{{\bf x}, \,{\bf y}}^{}\bigg\}
\nonumber \\
& \overset{(b)}{=} \sum_{\;\;\;{P\mathstrut}_{{\bf x}, \,{\bf y}, \, \hat{\bf x}}^{}}
\,\Pr\,\big\{({\bf X\mathstrut}_{m}, {\bf Y}) \, \in \, T({P\mathstrut}_{{\bf x}, \,{\bf y}}^{})\big\} \,\cdot\,
\Pr\,\bigg\{E\big({P\mathstrut}_{\hat{\bf x} \, | \,{\bf x}, \, {\bf y}}^{}, {\bf X\mathstrut}_{m}, {\bf Y}\big) = E_{\min}\big({\bf X\mathstrut}_{m}, {\bf Y}\big),
\nonumber \\
& \;\;\;\;\;\;\;\;\;\;\;
\;\;\;\;\;\;\;\;\;\;\;\;\;\;\;\;\;\;\;\;\;\;\;\;\;\;\;\;\;\,
E\big({P\mathstrut}_{\hat{\bf x} \, | \,{\bf x}, \, {\bf y}}^{}, {\bf X\mathstrut}_{m}, {\bf Y}\big)
\, \leq \, \frac{|{\cal X}||{\cal X}||{\cal Y}|\ln(n + 1)}{n} \, - \, \mathbb{E}\,[\ln P(Y \,|\, X)] \, + \, D
\; \bigg| \; {P\mathstrut}_{{\bf x}, \,{\bf y}}^{}\bigg\}
\nonumber \\
& \leq \sum_{\;\;\;{P\mathstrut}_{{\bf x}, \,{\bf y}, \, \hat{\bf x}}^{}}
\,\Pr\,\big\{({\bf X\mathstrut}_{m}, {\bf Y}) \, \in \, T({P\mathstrut}_{{\bf x}, \,{\bf y}}^{})\big\} \,\times
\nonumber \\
& \;\;\;\;\;\;\;\;\;\;\;\;\;\;\;\;
\;\;\;\;\;\;\;\;\;\;\;\;\;\;\;\;\,
\Pr\,\bigg\{E\big({P\mathstrut}_{\hat{\bf x} \, | \,{\bf x}, \, {\bf y}}^{}, {\bf X\mathstrut}_{m}, {\bf Y}\big)
\, \leq \, \frac{|{\cal X}||{\cal X}||{\cal Y}|\ln(n + 1)}{n} \, - \, \mathbb{E}\,[\ln P(Y \,|\, X)] \, + \, D
\; \bigg| \; {P\mathstrut}_{{\bf x}, \,{\bf y}}^{}\bigg\}
\nonumber \\
& \overset{(c)}{\leq} \sum_{\;\;\;{P\mathstrut}_{{\bf x}, \,{\bf y}, \, \hat{\bf x}}^{}}
\,\Pr\,\big\{({\bf X\mathstrut}_{m}, {\bf Y}) \, \in \, T({P\mathstrut}_{{\bf x}, \,{\bf y}}^{})\big\} \,\cdot\,
\Pr\,\Big\{E\big({P\mathstrut}_{\hat{\bf x} \, | \,{\bf x}, \, {\bf y}}^{}, {\bf X\mathstrut}_{m}, {\bf Y}\big)
\, \leq \, \epsilon_{1} \, - \, \mathbb{E}\,[\ln P(Y \,|\, X)] \, + \, D
\; \Big| \; {P\mathstrut}_{{\bf x}, \,{\bf y}}^{}\Big\}
\nonumber \\
& \overset{(d)}{=} \; \!\!
\sum_{
\;\;\;{P\mathstrut}_{{\bf x}, \,{\bf y}, \, \hat{\bf x}}^{}\,:\;\;
f({P\mathstrut}_{{\bf x}, \,{\bf y}, \, \hat{\bf x}}^{})
\;\; \leq \;\; R}
\,\Pr\,\big\{({\bf X\mathstrut}_{m}, {\bf Y}) \, \in \, T({P\mathstrut}_{{\bf x}, \,{\bf y}}^{})\big\} \cdot
\Pr\,\Big\{
E\big({P\mathstrut}_{\hat{\bf x} \, | \,{\bf x}, \, {\bf y}}^{}, {\bf X\mathstrut}_{m}, {\bf Y}\big) \, > \, g\big({P\mathstrut}_{{\bf x}, \,{\bf y}, \, \hat{\bf x}}^{}\big),
\nonumber \\
& \;\;\;\;\;\;\;\;\;\;\;\;\;\;\;\;\;\;\;\;\;\;\;\;\;\;\;\;\;\;\;\;\;\;\;\;
\;\;\;\;\;\;\;\;\;\;\;\;\;\;\;\;\;\;\;\;\;\;\;\;\;\;\;\;\;\;\;\;
E\big({P\mathstrut}_{\hat{\bf x} \, | \,{\bf x}, \, {\bf y}}^{}, {\bf X\mathstrut}_{m}, {\bf Y}\big)
\, \leq \, \epsilon_{1} \, - \, \mathbb{E}\,[\ln P(Y \,|\, X)] \, + \, D
\; \Big| \; {P\mathstrut}_{{\bf x}, \,{\bf y}}^{}\Big\}
\nonumber \\
& \;\, + \!\!
\sum_{
\;\;\;{P\mathstrut}_{{\bf x}, \,{\bf y}, \, \hat{\bf x}}^{}\,:\;\;
f({P\mathstrut}_{{\bf x}, \,{\bf y}, \, \hat{\bf x}}^{})
\;\; \leq \;\; R}
\,\Pr\,\big\{({\bf X\mathstrut}_{m}, {\bf Y}) \, \in \, T({P\mathstrut}_{{\bf x}, \,{\bf y}}^{})\big\} \cdot
\Pr\,\Big\{
E\big({P\mathstrut}_{\hat{\bf x} \, | \,{\bf x}, \, {\bf y}}^{}, {\bf X\mathstrut}_{m}, {\bf Y}\big) \, \leq \, g\big({P\mathstrut}_{{\bf x}, \,{\bf y}, \, \hat{\bf x}}^{}\big),
\nonumber \\
& \;\;\;\;\;\;\;\;\;\;\;\;\;\;\;\;\;\;\;\;\;\;\;\;\;\;\;\;\;\;\;\;\;\;\;\;
\;\;\;\;\;\;\;\;\;\;\;\;\;\;\;\;\;\;\;\;\;\;\;\;\;\;\;\;\;\;\;
\underbrace{E\big({P\mathstrut}_{\hat{\bf x} \, | \,{\bf x}, \, {\bf y}}^{}, {\bf X\mathstrut}_{m}, {\bf Y}\big)
\, \leq \, \epsilon_{1} \, - \, \mathbb{E}\,[\ln P(Y \,|\, X)] \, + \, D}_{\text{delete}}
\; \Big| \; {P\mathstrut}_{{\bf x}, \,{\bf y}}^{}\Big\}
\nonumber
\end{align}
\begin{align}
& \;\, + \!\!
\sum_{
\;\;\;{P\mathstrut}_{{\bf x}, \,{\bf y}, \, \hat{\bf x}}^{}\,:\;\;
f({P\mathstrut}_{{\bf x}, \,{\bf y}, \, \hat{\bf x}}^{})
\;\; > \;\; R}
\,\Pr\,\big\{({\bf X\mathstrut}_{m}, {\bf Y}) \, \in \, T({P\mathstrut}_{{\bf x}, \,{\bf y}}^{})\big\} \cdot
\Pr\,\Big\{
E\big({P\mathstrut}_{\hat{\bf x} \, | \,{\bf x}, \, {\bf y}}^{}, {\bf X\mathstrut}_{m}, {\bf Y}\big) \, > \, h\big({P\mathstrut}_{{\bf x}, \,{\bf y}, \, \hat{\bf x}}^{}\big),
\nonumber \\
& \;\;\;\;\;\;\;\;\;\;\;\;\;\;\;\;\;\;\;\;\;\;\;\;\;\;\;\;\;\;\;\;\;\;\;\;
\;\;\;\;\;\;\;\;\;\;\;\;\;\;\;\;\;\;\;\;\;\;\;\;\;\;\;\;\;\;\;\;
E\big({P\mathstrut}_{\hat{\bf x} \, | \,{\bf x}, \, {\bf y}}^{}, {\bf X\mathstrut}_{m}, {\bf Y}\big)
\, \leq \, \epsilon_{1} \, - \, \mathbb{E}\,[\ln P(Y \,|\, X)] \, + \, D
\; \Big| \; {P\mathstrut}_{{\bf x}, \,{\bf y}}^{}\Big\}
\nonumber \\
& \;\, + \!\!
\sum_{
\;\;\;{P\mathstrut}_{{\bf x}, \,{\bf y}, \, \hat{\bf x}}^{}\,:\;\;
f({P\mathstrut}_{{\bf x}, \,{\bf y}, \, \hat{\bf x}}^{})
\;\; > \;\; R}
\,\Pr\,\big\{({\bf X\mathstrut}_{m}, {\bf Y}) \, \in \, T({P\mathstrut}_{{\bf x}, \,{\bf y}}^{})\big\} \cdot
\Pr\,\Big\{
E\big({P\mathstrut}_{\hat{\bf x} \, | \,{\bf x}, \, {\bf y}}^{}, {\bf X\mathstrut}_{m}, {\bf Y}\big) \, \leq \, h\big({P\mathstrut}_{{\bf x}, \,{\bf y}, \, \hat{\bf x}}^{}\big),
\nonumber \\
& \;\;\;\;\;\;\;\;\;\;\;\;\;\;\;\;\;\;\;\;\;\;\;\;\;\;\;\;\;\;\;\;\;\;\;\;
\;\;\;\;\;\;\;\;\;\;\;\;\;\;\;\;\;\;\;\;\;\;\;\;\;\;\;\;\;\;\;
\underbrace{E\big({P\mathstrut}_{\hat{\bf x} \, | \,{\bf x}, \, {\bf y}}^{}, {\bf X\mathstrut}_{m}, {\bf Y}\big)
\, \leq \, \epsilon_{1} \, - \, \mathbb{E}\,[\ln P(Y \,|\, X)] \, + \, D}_{\text{delete}}
\; \Big| \; {P\mathstrut}_{{\bf x}, \,{\bf y}}^{}\Big\}
\nonumber \\
& \overset{(e)}{\leq} \; \!\!
\sum_{
\;\;\;{P\mathstrut}_{{\bf x}, \,{\bf y}, \, \hat{\bf x}}^{}\,:\;\;
f({P\mathstrut}_{{\bf x}, \,{\bf y}, \, \hat{\bf x}}^{})
\;\; \leq \;\; R}
\,\Pr\,\big\{({\bf X\mathstrut}_{m}, {\bf Y}) \, \in \, T({P\mathstrut}_{{\bf x}, \,{\bf y}}^{})\big\} \,\times
\nonumber \\
& \;\;\;\;\;\;\;\;\;\;\;\;\;\;\;\;\;\;\;\;\;\;\;\;\;\;\;\;\;\;\;
\;\;\;\;\;\;\;\;\;\;\;\;\;\;\;\;\;\;\;\;\;\;\;\;\;\;\;\;\;\;\;\;\;\;\;\;\;\;\;\;\;\;\,
\Pr\,\Big\{g\big({P\mathstrut}_{{\bf x}, \,{\bf y}, \, \hat{\bf x}}^{}\big)
\, \leq \, \epsilon_{1} \, - \, \mathbb{E}\,[\ln P(Y \,|\, X)] \, + \, D
\; \Big| \; {P\mathstrut}_{{\bf x}, \,{\bf y}}^{}\Big\}
\nonumber \\
& \;\, + \!\!
\sum_{
\;\;\;{P\mathstrut}_{{\bf x}, \,{\bf y}, \, \hat{\bf x}}^{}\,:\;\;
f({P\mathstrut}_{{\bf x}, \,{\bf y}, \, \hat{\bf x}}^{})
\;\; \leq \;\; R}
\,\Pr\,\big\{({\bf X\mathstrut}_{m}, {\bf Y}) \, \in \, T({P\mathstrut}_{{\bf x}, \,{\bf y}}^{})\big\} \cdot
\Pr\,\Big\{
E\big({P\mathstrut}_{\hat{\bf x} \, | \,{\bf x}, \, {\bf y}}^{}, {\bf X\mathstrut}_{m}, {\bf Y}\big) \, \leq \, g\big({P\mathstrut}_{{\bf x}, \,{\bf y}, \, \hat{\bf x}}^{}\big)
\; \Big| \; {P\mathstrut}_{{\bf x}, \,{\bf y}}^{}\Big\}
\nonumber \\
& \;\, + \!\!
\sum_{
\;\;\;{P\mathstrut}_{{\bf x}, \,{\bf y}, \, \hat{\bf x}}^{}\,:\;\;
f({P\mathstrut}_{{\bf x}, \,{\bf y}, \, \hat{\bf x}}^{})
\;\; > \;\; R}
\,\Pr\,\big\{({\bf X\mathstrut}_{m}, {\bf Y}) \, \in \, T({P\mathstrut}_{{\bf x}, \,{\bf y}}^{})\big\} \cdot
\Pr\,\Big\{
E\big({P\mathstrut}_{\hat{\bf x} \, | \,{\bf x}, \, {\bf y}}^{}, {\bf X\mathstrut}_{m}, {\bf Y}\big) \, < \, +\infty,
\nonumber \\
& \;\;\;\;\;\;\;\;\;\;\;\;\;\;\;\;\;\;\;\;\;\;\;\;\;\;\;\;\;\;\;
\;\;\;\;\;\;\;\;\;\;\;\;\;\;\;\;\;\;\;\;\;\;\;\;\;\;\;\;\;\;\;\;\;\;\;\;\;\;\;\;\;\;\,
\;\;\;\;\;\;\;\,
h\big({P\mathstrut}_{{\bf x}, \,{\bf y}, \, \hat{\bf x}}^{}\big)
\, \leq \, \epsilon_{1} \, - \, \mathbb{E}\,[\ln P(Y \,|\, X)] \, + \, D
\; \Big| \; {P\mathstrut}_{{\bf x}, \,{\bf y}}^{}\Big\}
\nonumber \\
& \;\, + \!\!
\sum_{
\;\;\;{P\mathstrut}_{{\bf x}, \,{\bf y}, \, \hat{\bf x}}^{}\,:\;\;
f({P\mathstrut}_{{\bf x}, \,{\bf y}, \, \hat{\bf x}}^{})
\;\; > \;\; R}
\,\Pr\,\big\{({\bf X\mathstrut}_{m}, {\bf Y}) \, \in \, T({P\mathstrut}_{{\bf x}, \,{\bf y}}^{})\big\} \cdot
\Pr\,\Big\{
E\big({P\mathstrut}_{\hat{\bf x} \, | \,{\bf x}, \, {\bf y}}^{}, {\bf X\mathstrut}_{m}, {\bf Y}\big) \, \leq \, h\big({P\mathstrut}_{{\bf x}, \,{\bf y}, \, \hat{\bf x}}^{}\big)
\; \Big| \; {P\mathstrut}_{{\bf x}, \,{\bf y}}^{}\Big\}
\nonumber \\
& \overset{(f)}{=} \; \!\!
\sum_{
\;\;\;{P\mathstrut}_{{\bf x}, \,{\bf y}, \, \hat{\bf x}}^{}\,:\;\;
f({P\mathstrut}_{{\bf x}, \,{\bf y}, \, \hat{\bf x}}^{})
\;\; \leq \;\; R}
\,\Pr\,\big\{({\bf X\mathstrut}_{m}, {\bf Y}) \, \in \, T({P\mathstrut}_{{\bf x}, \,{\bf y}}^{})\big\} \, \times
\nonumber \\
& \;\;\;\;\;\;\;\;\;\;\;\;\;\;\;\;\;\;\;\;\;\;\;\;\;\;\;\;\;\;\;
\;\;\;\;\;\;\;\;\;\;\;\;\;\;\;\;\;\;\;\;\;\;\;\;\;\;\;\;\;\;\;\;\;\;\;\;\;\;\;\;\;\;\;
\mathbbm{1}_{\displaystyle\big\{g\big({P\mathstrut}_{{\bf x}, \,{\bf y}, \, \hat{\bf x}}^{}\big)
\; \leq \; \epsilon_{1} \, - \, \mathbb{E}\,[\ln P(Y \,|\, X)] \, + \, D\big\}}
({P\mathstrut}_{{\bf x}, \,{\bf y}, \, \hat{\bf x}}^{})
\nonumber \\
&
\nonumber \\
& \;\, + \!\!
\sum_{
\;\;\;{P\mathstrut}_{{\bf x}, \,{\bf y}, \, \hat{\bf x}}^{}\,:\;\;
f({P\mathstrut}_{{\bf x}, \,{\bf y}, \, \hat{\bf x}}^{})
\;\; \leq \;\; R}
\,\Pr\,\big\{({\bf X\mathstrut}_{m}, {\bf Y}) \, \in \, T({P\mathstrut}_{{\bf x}, \,{\bf y}}^{})\big\} \cdot
\Pr\,\Big\{
E\big({P\mathstrut}_{\hat{\bf x} \, | \,{\bf x}, \, {\bf y}}^{}, {\bf X\mathstrut}_{m}, {\bf Y}\big) \, \leq \, g\big({P\mathstrut}_{{\bf x}, \,{\bf y}, \, \hat{\bf x}}^{}\big)
\; \Big| \; {P\mathstrut}_{{\bf x}, \,{\bf y}}^{}\Big\}
\nonumber \\
& \;\, + \!\!
\sum_{
\;\;\;{P\mathstrut}_{{\bf x}, \,{\bf y}, \, \hat{\bf x}}^{}\,:\;\;
f({P\mathstrut}_{{\bf x}, \,{\bf y}, \, \hat{\bf x}}^{})
\;\; > \;\; R}
\,\Pr\,\big\{({\bf X\mathstrut}_{m}, {\bf Y}) \, \in \, T({P\mathstrut}_{{\bf x}, \,{\bf y}}^{})\big\}\cdot
\Pr\,\Big\{E\big({P\mathstrut}_{\hat{\bf x} \, | \,{\bf x}, \, {\bf y}}^{}, {\bf X\mathstrut}_{m}, {\bf Y}\big) < +\infty \; \Big| \; {P\mathstrut}_{{\bf x}, \,{\bf y}}^{}\Big\} \, \times
\nonumber \\
& \;\;\;\;\;\;\;\;\;\;\;\;\;\;\;\;\;\;\;\;\;\;\;\;\;\;\;\;\;\;\;
\;\;\;\;\;\;\;\;\;\;\;\;\;\;\;\;\;\;\;\;\;\;\;\;\;\;\;\;\;\;\;\;\;\;\;\;\;\;\;\;\;\;\;
\mathbbm{1}_{\displaystyle\big\{h\big({P\mathstrut}_{{\bf x}, \,{\bf y}, \, \hat{\bf x}}^{}\big)
\; \leq \; \epsilon_{1} \, - \, \mathbb{E}\,[\ln P(Y \,|\, X)] \, + \, D\big\}}
({P\mathstrut}_{{\bf x}, \,{\bf y}, \, \hat{\bf x}}^{})
\nonumber \\
&
\nonumber \\
& \;\, + \!\!
\sum_{
\;\;\;{P\mathstrut}_{{\bf x}, \,{\bf y}, \, \hat{\bf x}}^{}\,:\;\;
f({P\mathstrut}_{{\bf x}, \,{\bf y}, \, \hat{\bf x}}^{})
\;\; > \;\; R}
\,\Pr\,\big\{({\bf X\mathstrut}_{m}, {\bf Y}) \, \in \, T({P\mathstrut}_{{\bf x}, \,{\bf y}}^{})\big\} \cdot
\Pr\,\Big\{
E\big({P\mathstrut}_{\hat{\bf x} \, | \,{\bf x}, \, {\bf y}}^{}, {\bf X\mathstrut}_{m}, {\bf Y}\big) \, \leq \, h\big({P\mathstrut}_{{\bf x}, \,{\bf y}, \, \hat{\bf x}}^{}\big)
\; \Big| \; {P\mathstrut}_{{\bf x}, \,{\bf y}}^{}\Big\},
\label{eqBoundFGH}
\end{align}
for sufficiently large $n$, which is needed for ($c$) to hold.
Explanation of steps:\newline
($a$) follows by the definition of the error event ${\cal E}_{m}$ (\ref{eqForneyErrorEvent}) and the bound on Forney's sum (\ref{eqForneySum});\newline
($b$) is an identity, with a notation, given $\big\{({\bf X\mathstrut}_{m}, {\bf Y}) \, \in \, T({P\mathstrut}_{{\bf x}, \,{\bf y}}^{})\big\}$:
\begin{equation} \label{eqExpectation}
\frac{\ln P({\bf Y} \,|\, {\bf X\mathstrut}_{m})}{n} \;\; = \;\; \sum_{x, \, y}{P\mathstrut}_{{\bf x}, \,{\bf y}}^{}(x, y)\ln P(y \,|\, x) \;\; \triangleq \;\;
\mathbb{E}\,[\ln P(Y \,|\, X)]\,;
\end{equation}
($c$) holds for any $\epsilon_{1}\, > \, 0$, for sufficiently large $n$, such that
\begin{equation} \label{eqNLarge}
\frac{|{\cal X}||{\cal X}||{\cal Y}|\ln(n + 1)}{n} \; \leq \; \epsilon_{1}\,;
\end{equation}
($d$) is an identity, for arbitrary functions $f\big({P\mathstrut}_{{\bf x}, \,{\bf y}, \, \hat{\bf x}}^{}\big)$,
$g\big({P\mathstrut}_{{\bf x}, \,{\bf y}, \, \hat{\bf x}}^{}\big)$, $h\big({P\mathstrut}_{{\bf x}, \,{\bf y}, \, \hat{\bf x}}^{}\big)$;\newline
($e$) uses {\em if-then} relations between events:
\begin{align}
\Big\{E\big({P\mathstrut}_{\hat{\bf x} \, | \,{\bf x}, \, {\bf y}}^{}, {\bf X\mathstrut}_{m}, {\bf Y}\big) \, > \, g\big({P\mathstrut}_{{\bf x}, \,{\bf y}, \, \hat{\bf x}}^{}\big), \;
E\big({P\mathstrut}_{\hat{\bf x} \, | \,{\bf x}, \, {\bf y}}^{}, {\bf X\mathstrut}_{m}, {\bf Y}\big)
\, & \leq \, \epsilon_{1} \, - \, \mathbb{E}\,[\ln P(Y \,|\, X)] \, + \, D\Big\} \;\; \Rightarrow
\nonumber \\
\Big\{g\big({P\mathstrut}_{{\bf x}, \,{\bf y}, \, \hat{\bf x}}^{}\big)
\, & \leq \, \epsilon_{1} \, - \, \mathbb{E}\,[\ln P(Y \,|\, X)] \, + \, D\Big\},
\nonumber \\
\Big\{
E\big({P\mathstrut}_{\hat{\bf x} \, | \,{\bf x}, \, {\bf y}}^{}, {\bf X\mathstrut}_{m}, {\bf Y}\big) \, > \, h\big({P\mathstrut}_{{\bf x}, \,{\bf y}, \, \hat{\bf x}}^{}\big),\;
E\big({P\mathstrut}_{\hat{\bf x} \, | \,{\bf x}, \, {\bf y}}^{}, {\bf X\mathstrut}_{m}, {\bf Y}\big)
\, & \leq \, \epsilon_{1} \, - \, \mathbb{E}\,[\ln P(Y \,|\, X)] \, + \, D\Big\} \;\; \Rightarrow
\nonumber \\
\Big\{
E\big({P\mathstrut}_{\hat{\bf x} \, | \,{\bf x}, \, {\bf y}}^{}, {\bf X\mathstrut}_{m}, {\bf Y}\big) \, < \, +\infty, \;
h\big({P\mathstrut}_{{\bf x}, \,{\bf y}, \, \hat{\bf x}}^{}\big)
\, & \leq \, \epsilon_{1} \, - \, \mathbb{E}\,[\ln P(Y \,|\, X)] \, + \, D\Big\};
\nonumber
\end{align}
($f$) is an identity, when the functions
$g\big({P\mathstrut}_{{\bf x}, \,{\bf y}, \, \hat{\bf x}}^{}\big)$ and $h\big({P\mathstrut}_{{\bf x}, \,{\bf y}, \, \hat{\bf x}}^{}\big)$
are arbitrary deterministic. In this case, the events
\begin{align}
\big\{g\big({P\mathstrut}_{{\bf x}, \,{\bf y}, \, \hat{\bf x}}^{}\big)
\, & \leq \, \epsilon_{1} \, - \, \mathbb{E}\,[\ln P(Y \,|\, X)] \, + \, D\big\},
\nonumber \\
\big\{h\big({P\mathstrut}_{{\bf x}, \,{\bf y}, \, \hat{\bf x}}^{}\big)
\, & \leq \, \epsilon_{1} \, - \, \mathbb{E}\,[\ln P(Y \,|\, X)] \, + \, D\big\}
\nonumber
\end{align}
are deterministic conditions (i.e., they either hold with probability $1$ or with probability $0$), and indicator functions can be used in place of probabilities.

The upper bound (\ref{eqBoundFGH}) was devised with something like the following lemma in mind:
\begin{lemma} \label{lemma9}
{\em Let $Z_{i}\,\sim\, \text{i.i.d}\;\text{Bernoulli}\left({e\mathstrut}^{-nI}\right)$,
$i \, = \, 1, \, 2, \, ... \, , \, {e\mathstrut}^{nR}$. For $\epsilon \, > \, 0$,}

{\em if $I\,\leq\,R \, + \, \epsilon$, then}
\begin{equation} \label{eqRI}
\Pr\,\Bigg\{\sum_{i \, = \, 1}^{{e\mathstrut}^{nR}}Z_{i} \, \geq \, {e\mathstrut}^{n(R\,-\,I\,+\,2\epsilon)}\Bigg\}
\;\; < \;\; \exp\big\{-[{e\mathstrut}^{n\epsilon}\,-\,(e - 1){e\mathstrut}^{-n\epsilon}]\big\},
\;\;\;\;\;\;\;\; {e\mathstrut}^{2n\epsilon} \; \geq \; e \, - \, 1,
\end{equation}

{\em if $I\,>\,R$, then}
\begin{equation} \label{eqIR}
\Pr\,\Bigg\{\sum_{i \, = \, 1}^{{e\mathstrut}^{nR}}Z_{i} \, \geq \, {e\mathstrut}^{n\epsilon}\Bigg\}
\;\; < \;\; \exp\big\{-[{e\mathstrut}^{n\epsilon}\,-\,e\,+\,1]\big\}.
\end{equation}
\end{lemma}
\begin{proof}
This is an unoptimized Chernoff bound, with the parameter in the exponent $=1$:
\begin{align}
\Pr\,\Bigg\{\sum_{i \, = \, 1}^{{e\mathstrut}^{nR}}Z_{i} \, \geq \, {e\mathstrut}^{n(\Delta\,+\,\epsilon)}\Bigg\}
\;\; & = \;\;
\Pr\,\Bigg\{\exp\Bigg\{\sum_{i \, = \, 1}^{{e\mathstrut}^{nR}}Z_{i}\Bigg\} \, \geq \, \exp\Big\{{e\mathstrut}^{n(\Delta\,+\,\epsilon)}\Big\}\Bigg\}
\nonumber \\
& \overset{(a)}{\leq} \;\;
\exp\Big\{-{e\mathstrut}^{n(\Delta\,+\,\epsilon)}\Big\}\cdot \mathbb{E}\,\left[\exp\Bigg\{\sum_{i \, = \, 1}^{{e\mathstrut}^{nR}}Z_{i}\Bigg\}\right]
\nonumber \\
& = \;\;
\exp\Big\{-{e\mathstrut}^{n(\Delta\,+\,\epsilon)}\Big\}\cdot \prod_{i \, = \, 1}^{{e\mathstrut}^{nR}}\mathbb{E}\,\left[{e\mathstrut}^{Z_{i}}\right]
\nonumber \\
& = \;\;
\exp\Big\{-{e\mathstrut}^{n(\Delta\,+\,\epsilon)}\Big\}\cdot \left[1\, + \, (e-1){e\mathstrut}^{-nI}\right]^{{e\mathstrut}^{nR}}
\nonumber \\
& = \;\;
\exp\Big\{-{e\mathstrut}^{n(\Delta\,+\,\epsilon)}\Big\}\cdot \bigg[\underbrace{\left(1\, + \, (e-1){e\mathstrut}^{-nI}\right)^{\frac{1}{(e-1){e\mathstrut}^{-nI}}}}_{<\,e}\bigg]^{(e-1){e\mathstrut}^{-nI}\cdot\,{e\mathstrut}^{nR}}
\nonumber \\
& \overset{(b)}{<} \;\;
\exp\Big\{-{e\mathstrut}^{n(\Delta\,+\,\epsilon)}\Big\}\cdot \exp\left\{(e-1){e\mathstrut}^{n(R\,-\,I)}\right\}
\nonumber \\
& = \;\;
\left\{
\begin{array}{l l}
\exp\left\{-{e\mathstrut}^{n(R\,-\,I)}\left[{e\mathstrut}^{2n\epsilon}\,-\,e\,+\,1\right]\right\}, & \;\;\; \Delta \; = \; R \, - \, I \, + \, \epsilon,\\
\exp\left\{-{e\mathstrut}^{n\epsilon}\,+\,(e - 1){e\mathstrut}^{n(R\,-\,I)}\right\}, & \;\;\; \Delta \; = \; 0,
\end{array}
\right.
\nonumber
\end{align}
where ($a$) is Markov's inequality (yielding at this step an unoptimized Chernoff bound with parameter $1$), and ($b$) holds because $\;(1 + x)^{1/x} \, < \, e$.

For the case $I\,\leq\,R + \epsilon$, we take the bound with $\Delta \, = \, R  -  I + \epsilon$ and obtain
\begin{align}
...\;\; < \;\; \exp\left\{-{e\mathstrut}^{n(R\,-\,I)}\left[{e\mathstrut}^{2n\epsilon}\,-\,e\,+\,1\right]\right\}
\;\; & \leq \;\; \exp\left\{-\left[{e\mathstrut}^{n\epsilon}\,-\,(e - 1){e\mathstrut}^{-n\epsilon}\right]\right\},
\;\;\;\;\;\; {e\mathstrut}^{2n\epsilon}\,-\,e\,+\,1 \; \geq \; 0.
\nonumber 
\end{align}
For the case $I\,>\,R$, we take the bound with $\Delta \, = \, 0$ and obtain
\begin{align}
\exp\left\{-{e\mathstrut}^{n\epsilon}\,+\,(e - 1){e\mathstrut}^{n(R\,-\,I)}\right\}
\;\; & < \;\; \exp\left\{-{e\mathstrut}^{n\epsilon}\,+\,e \,-\, 1\right\}.
\nonumber
\end{align}
\end{proof}

In order to use Lemma~\ref{lemma9}, recall that the probability of a conditional type is bounded from above and below as
\begin{align}
\exp\big\{-nf\big({P\mathstrut}_{{\bf x}, \,{\bf y}, \, \hat{\bf x}}^{}\big)\big\}
\; & \triangleq \;
\exp\Big\{-nD\big({P\mathstrut}_{{\bf x}, \,{\bf y}, \, \hat{\bf x}}^{}(x, y, \hat{x})\,\big\|\,{P\mathstrut}_{{\bf x}, \,{\bf y}}^{}(x, y)\cdot Q(\hat{x})\big)\Big\}
\label{eqDefF} \\
& \geq \;
\Pr\,
\Big\{{\bf X\mathstrut}_{m'} \; \in \; T\big({P\mathstrut}_{\hat{\bf x} \, | \,{\bf x}, \, {\bf y}}^{}, \, {\bf X\mathstrut}_{m}, {\bf Y}\big)\;\Big|\;
({\bf X\mathstrut}_{m}, {\bf Y}) \, \in \, T({P\mathstrut}_{{\bf x}, \,{\bf y}}^{})\Big\}
\; \triangleq \;
\exp\{-nI\}
\label{eqCondType} \\
& \geq \;
{(n\, + \, 1)\mathstrut}^{-|{\cal X}|\cdot|{\cal X}|\cdot|{\cal Y}|}\cdot
\exp\big\{-nf\big({P\mathstrut}_{{\bf x}, \,{\bf y}, \, \hat{\bf x}}^{}\big)\big\}.
\nonumber
\end{align}
These definitions give
\begin{equation} \label{eqIandF}
f\big({P\mathstrut}_{{\bf x}, \,{\bf y}, \, \hat{\bf x}}^{}\big) \;\; \leq \;\; I
\;\; \leq \;\; f\big({P\mathstrut}_{{\bf x}, \,{\bf y}, \, \hat{\bf x}}^{}\big) \; + \; \frac{|{\cal X}||{\cal X}||{\cal Y}|\ln(n + 1)}{n}.
\end{equation}

If $f\big({P\mathstrut}_{{\bf x}, \,{\bf y}, \, \hat{\bf x}}^{}\big) \, \leq \, R$,
then for $n$ sufficiently large, as in (\ref{eqNLarge}), we get $I \, \leq \, R + \epsilon_{1}$.
If $n$
satisfies (\ref{eqNLarge}), then it
is also large enough to satisfy ${e\mathstrut}^{2n\epsilon_{1}}\, \geq \, e - 1$.
For such $n$, the first part of Lemma~\ref{lemma9} holds for the following:
\begin{align}
& \Pr\,\bigg\{E\big({P\mathstrut}_{\hat{\bf x} \, | \,{\bf x}, \, {\bf y}}^{}, {\bf X\mathstrut}_{m}, {\bf Y}\big)
\; \leq \; \underbrace{- \mathbb{E}_{{P\mathstrut}_{\hat{\bf x}, \,{\bf y}}^{}}[\ln P(Y \,|\, \hat{X})]
 - R + f({P\mathstrut}_{{\bf x}, \,{\bf y}, \, \hat{\bf x}}^{}) - 2\epsilon_{1}}_{g({P\mathstrut}_{{\bf x}, \,{\bf y}, \, \hat{\bf x}}^{})}
\; \bigg| \;
({\bf X\mathstrut}_{m}, {\bf Y}) \, \in \, T({P\mathstrut}_{{\bf x}, \,{\bf y}}^{})\bigg\}
\nonumber \\
\overset{(a)}{=} \;\; & \Pr\,\bigg\{\sum_{m' \, \neq \, m}P({\bf Y} \,|\, {\bf X\mathstrut}_{m'})\cdot
\mathbbm{1}_{\displaystyle\big\{{\bf X\mathstrut}_{m'} \; \in \; T\big({P\mathstrut}_{\hat{\bf x} \, | \,{\bf x}, \, {\bf y}}^{}, \, {\bf X\mathstrut}_{m}, {\bf Y}\big)\big\}}(m')
\; \geq
\nonumber \\
& \;\;\;\;\;\;\;\;\;\;\;\;\;\;\;\;\;\;\;\;\;\;\;\;\,\,
\exp\Big\{n\big(\mathbb{E}_{{P\mathstrut}_{\hat{\bf x}, \,{\bf y}}^{}}[\ln P(Y \,|\, \hat{X})] \, + \, R\,-\,f({P\mathstrut}_{{\bf x}, \,{\bf y}, \, \hat{\bf x}}^{})\,+\,2\epsilon_{1}\big)\Big\}
\; \bigg| \;
({\bf X\mathstrut}_{m}, {\bf Y}) \, \in \, T({P\mathstrut}_{{\bf x}, \,{\bf y}}^{})\bigg\}
\nonumber \\
= \;\; & \Pr\,\bigg\{\sum_{m' \, \neq \, m}
\mathbbm{1}_{\displaystyle\big\{{\bf X\mathstrut}_{m'} \; \in \; T\big({P\mathstrut}_{\hat{\bf x} \, | \,{\bf x}, \, {\bf y}}^{}, \, {\bf X\mathstrut}_{m}, {\bf Y}\big)\big\}}(m')
\; \geq \; {e\mathstrut}^{n\big(R\,-\,f({P\mathstrut}_{{\bf x}, \,{\bf y}, \, \hat{\bf x}}^{})\,+\,2\epsilon_{1}\big)}
\; \bigg| \;
({\bf X\mathstrut}_{m}, {\bf Y}) \, \in \, T({P\mathstrut}_{{\bf x}, \,{\bf y}}^{})\bigg\}
\nonumber \\
\overset{(b)}{\leq} \;\; & \Pr\,\Bigg\{\sum_{m' \, = \, 1}^{{e\mathstrut}^{nR}}
\mathbbm{1}_{\displaystyle\big\{{\bf X\mathstrut}_{m'} \; \in \; T\big({P\mathstrut}_{\hat{\bf x} \, | \,{\bf x}, \, {\bf y}}^{}, \, {\bf X\mathstrut}_{m}, {\bf Y}\big)\big\}}(m')
\; \geq \; {e\mathstrut}^{n(R\,-\,I\,+\,2\epsilon_{1})}
\; \Bigg| \;
({\bf X\mathstrut}_{m}, {\bf Y}) \, \in \, T({P\mathstrut}_{{\bf x}, \,{\bf y}}^{})\Bigg\}
\nonumber \\
\overset{(c)}{<} \;\; & \exp\big\{-[{e\mathstrut}^{n\epsilon_{1}}\,-\,(e - 1){e\mathstrut}^{-n\epsilon_{1}}]\big\},
\;\;\;\;\;\;\;\;\;
\;\;\;\;\;\;\;\;\;
\frac{|{\cal X}||{\cal X}||{\cal Y}|\ln(n + 1)}{n}
\; \leq \; \epsilon_{1},
\label{eqLemmaFirstPart}
\end{align}
where in ($a$) we use the definition of $E\big({P\mathstrut}_{\hat{\bf x} \, | \,{\bf x}, \, {\bf y}}^{}, {\bf X\mathstrut}_{m}, {\bf Y}\big)$
(\ref{eqTypeExp}), and notation (\ref{eqExpectation}) with ${P\mathstrut}_{\hat{\bf x}, \,{\bf y}}^{}$; in ($b$) we
assume the size of the codebook
$M\, = \, {e\mathstrut}^{nR}\,+\, 1$,
and use $I \, \geq \, f\big({P\mathstrut}_{{\bf x}, \,{\bf y}, \, \hat{\bf x}}^{}\big)$;
and ($c$) holds by (\ref{eqCondType}) and (\ref{eqRI}) of the lemma.
If we choose
\begin{equation} \label{eqDefG}
g\big({P\mathstrut}_{{\bf x}, \,{\bf y}, \, \hat{\bf x}}^{}\big) \;\; \triangleq \;\;
- \mathbb{E}_{{P\mathstrut}_{\hat{\bf x}, \,{\bf y}}^{}}[\ln P(Y \,|\, \hat{X})] \, - \, R \, + \,
f\big({P\mathstrut}_{{\bf x}, \,{\bf y}, \, \hat{\bf x}}^{}\big)
\, - \, 2\epsilon_{1},
\end{equation}
then with (\ref{eqLemmaFirstPart}) we obtain that the {\em second sum} in (\ref{eqBoundFGH}) is upper-bounded as
\begin{align}
& \sum_{
\;\;\;{P\mathstrut}_{{\bf x}, \,{\bf y}, \, \hat{\bf x}}^{}\,:\;\;
f({P\mathstrut}_{{\bf x}, \,{\bf y}, \, \hat{\bf x}}^{})
\;\; \leq \;\; R}
\,\Pr\,\big\{({\bf X\mathstrut}_{m}, {\bf Y}) \, \in \, T({P\mathstrut}_{{\bf x}, \,{\bf y}}^{})\big\} \cdot
\Pr\,\Big\{
E\big({P\mathstrut}_{\hat{\bf x} \, | \,{\bf x}, \, {\bf y}}^{}, {\bf X\mathstrut}_{m}, {\bf Y}\big) \, \leq \, g\big({P\mathstrut}_{{\bf x}, \,{\bf y}, \, \hat{\bf x}}^{}\big)
\; \Big| \; {P\mathstrut}_{{\bf x}, \,{\bf y}}^{}\Big\}
\nonumber \\
< \;\;
& \sum_{
\;\;\;{P\mathstrut}_{{\bf x}, \,{\bf y}, \, \hat{\bf x}}^{}\,:\;\;
f({P\mathstrut}_{{\bf x}, \,{\bf y}, \, \hat{\bf x}}^{})
\;\; \leq \;\; R}
\,\Pr\,\big\{({\bf X\mathstrut}_{m}, {\bf Y}) \, \in \, T({P\mathstrut}_{{\bf x}, \,{\bf y}}^{})\big\}\cdot
\exp\big\{-[{e\mathstrut}^{n\epsilon_{1}}\,-\,(e - 1){e\mathstrut}^{-n\epsilon_{1}}]\big\}
\nonumber \\
\leq \;\;
& \;\;\;\;\;\;\;\;
{(n\, + \, 1)\mathstrut}^{|{\cal X}|\cdot|{\cal X}|\cdot|{\cal Y}|}\,\cdot\,
\exp\big\{-[{e\mathstrut}^{n\epsilon_{1}}\,-\,(e - 1){e\mathstrut}^{-n\epsilon_{1}}]\big\},
\;\;\;\;\;\;
\forall\,n: \;\;
\frac{|{\cal X}||{\cal X}||{\cal Y}|\ln(n + 1)}{n}
\; \leq \; \epsilon_{1}.
\label{eqSecondSum}
\end{align}

On the other hand, if $f\big({P\mathstrut}_{{\bf x}, \,{\bf y}, \, \hat{\bf x}}^{}\big) \, > \, R$,
then by (\ref{eqIandF}) also $I \, > \, R$,
and the second part of Lemma~\ref{lemma9} holds for the following:
\begin{align}
& \Pr\,\bigg\{E\big({P\mathstrut}_{\hat{\bf x} \, | \,{\bf x}, \, {\bf y}}^{}, {\bf X\mathstrut}_{m}, {\bf Y}\big)
\; \leq \; \underbrace{- \mathbb{E}_{{P\mathstrut}_{\hat{\bf x}, \,{\bf y}}^{}}[\ln P(Y \,|\, \hat{X})] \, - \, {\epsilon}_{2}}_{h({P\mathstrut}_{{\bf x}, \,{\bf y}, \, \hat{\bf x}}^{})}
\; \bigg| \;
({\bf X\mathstrut}_{m}, {\bf Y}) \, \in \, T({P\mathstrut}_{{\bf x}, \,{\bf y}}^{})\bigg\}
\nonumber \\
\overset{(a)}{=} \;\; & \Pr\,\bigg\{\sum_{m' \, \neq \, m}P({\bf Y} \,|\, {\bf X\mathstrut}_{m'})\cdot
\mathbbm{1}_{\displaystyle\big\{{\bf X\mathstrut}_{m'} \; \in \; T\big({P\mathstrut}_{\hat{\bf x} \, | \,{\bf x}, \, {\bf y}}^{}, \, {\bf X\mathstrut}_{m}, {\bf Y}\big)\big\}}(m')
\; \geq
\nonumber \\
& \;\;\;\;\;\;\;\;\;\;\;\;\;\;\;\;\;\;\;\;\;\;\;\;\;\;\;\;\;\;\;\,
\exp\Big\{n\big(\mathbb{E}_{{P\mathstrut}_{\hat{\bf x}, \,{\bf y}}^{}}[\ln P(Y \,|\, \hat{X})] \, + \, {\epsilon}_{2}\big)\Big\}
\; \bigg| \;
({\bf X\mathstrut}_{m}, {\bf Y}) \, \in \, T({P\mathstrut}_{{\bf x}, \,{\bf y}}^{})\bigg\}
\nonumber
\end{align}
\begin{align}
= \;\; & \Pr\,\bigg\{\sum_{m' \, \neq \, m}
\mathbbm{1}_{\displaystyle\big\{{\bf X\mathstrut}_{m'} \; \in \; T\big({P\mathstrut}_{\hat{\bf x} \, | \,{\bf x}, \, {\bf y}}^{}, \, {\bf X\mathstrut}_{m}, {\bf Y}\big)\big\}}(m')
\; \geq \; {e\mathstrut}^{n{\epsilon}_{2}}
\; \bigg| \;
({\bf X\mathstrut}_{m}, {\bf Y}) \, \in \, T({P\mathstrut}_{{\bf x}, \,{\bf y}}^{})\bigg\}
\nonumber \\
\overset{(b)}{=} \;\; & \Pr\,\Bigg\{\sum_{m' \, = \, 1}^{{e\mathstrut}^{nR}}
\mathbbm{1}_{\displaystyle\big\{{\bf X\mathstrut}_{m'} \; \in \; T\big({P\mathstrut}_{\hat{\bf x} \, | \,{\bf x}, \, {\bf y}}^{}, \, {\bf X\mathstrut}_{m}, {\bf Y}\big)\big\}}(m')
\; \geq \; {e\mathstrut}^{n{\epsilon}_{2}}
\; \Bigg| \;
({\bf X\mathstrut}_{m}, {\bf Y}) \, \in \, T({P\mathstrut}_{{\bf x}, \,{\bf y}}^{})\Bigg\}
\nonumber \\
\overset{(c)}{<} \;\; & \exp\big\{-[{e\mathstrut}^{n{\epsilon}_{2}}\,-\,e\,+\,1]\big\},
\;\;\;\;\;\;\;\;\; \epsilon_{2} \; > \; 0,
\label{eqLemmaSecondPart}
\end{align}
where in ($a$) we use the definition of $E\big({P\mathstrut}_{\hat{\bf x} \, | \,{\bf x}, \, {\bf y}}^{}, {\bf X\mathstrut}_{m}, {\bf Y}\big)$
(\ref{eqTypeExp}), and notation (\ref{eqExpectation}) with ${P\mathstrut}_{\hat{\bf x}, \,{\bf y}}^{}$;
in ($b$) we assume the codebook size $M\, = \, {e\mathstrut}^{nR}\,+\, 1$;
and ($c$) holds by (\ref{eqCondType}) and (\ref{eqIR}) of the lemma.
If we choose
\begin{equation} \label{eqDefH}
h\big({P\mathstrut}_{{\bf x}, \,{\bf y}, \, \hat{\bf x}}^{}\big) \;\; \triangleq \;\;
- \mathbb{E}_{{P\mathstrut}_{\hat{\bf x}, \,{\bf y}}^{}}[\ln P(Y \,|\, \hat{X})] \, - \, {\epsilon}_{2},
\end{equation}
then with (\ref{eqLemmaSecondPart}) we obtain that the {\em fourth sum} in (\ref{eqBoundFGH}) is upper-bounded as
\begin{align}
& \sum_{
\;\;\;{P\mathstrut}_{{\bf x}, \,{\bf y}, \, \hat{\bf x}}^{}\,:\;\;
f({P\mathstrut}_{{\bf x}, \,{\bf y}, \, \hat{\bf x}}^{})
\;\; > \;\; R}
\,\Pr\,\big\{({\bf X\mathstrut}_{m}, {\bf Y}) \, \in \, T({P\mathstrut}_{{\bf x}, \,{\bf y}}^{})\big\} \cdot
\Pr\,\Big\{
E\big({P\mathstrut}_{\hat{\bf x} \, | \,{\bf x}, \, {\bf y}}^{}, {\bf X\mathstrut}_{m}, {\bf Y}\big) \, \leq \, h\big({P\mathstrut}_{{\bf x}, \,{\bf y}, \, \hat{\bf x}}^{}\big)
\; \Big| \; {P\mathstrut}_{{\bf x}, \,{\bf y}}^{}\Big\}
\nonumber \\
< \;\;
& \sum_{
\;\;\;{P\mathstrut}_{{\bf x}, \,{\bf y}, \, \hat{\bf x}}^{}\,:\;\;
f({P\mathstrut}_{{\bf x}, \,{\bf y}, \, \hat{\bf x}}^{})
\;\; > \;\; R}
\,\Pr\,\big\{({\bf X\mathstrut}_{m}, {\bf Y}) \, \in \, T({P\mathstrut}_{{\bf x}, \,{\bf y}}^{})\big\}\cdot
\exp\big\{-[{e\mathstrut}^{n{\epsilon}_{2}}\,-\,e\,+\,1]\big\}
\nonumber \\
\leq \;\;
& \;\;\;\;\;\;\;\;\;\;
{(n\, + \, 1)\mathstrut}^{|{\cal X}|\cdot|{\cal X}|\cdot|{\cal Y}|}\,\cdot\,
\exp\big\{-[{e\mathstrut}^{n{\epsilon}_{2}}\,-\,e\,+\,1]\big\},
\;\;\;\;\;\;\;\;\; \epsilon_{2} \; > \; 0.
\label{eqFourthSum}
\end{align}

With the definitions (\ref{eqDefF}), (\ref{eqDefG}), (\ref{eqDefH}) at hand,
we are ready to bound also the first and the third sums in (\ref{eqBoundFGH}).

The {\em first sum} in (\ref{eqBoundFGH}) is upper-bounded as follows:
\begin{align}
& \sum_{
\;\;\;{P\mathstrut}_{{\bf x}, \,{\bf y}, \, \hat{\bf x}}^{}\,:\;\;
f({P\mathstrut}_{{\bf x}, \,{\bf y}, \, \hat{\bf x}}^{})
\;\; \leq \;\; R}
\,\Pr\,\big\{({\bf X\mathstrut}_{m}, {\bf Y}) \, \in \, T({P\mathstrut}_{{\bf x}, \,{\bf y}}^{})\big\} \, \times
\nonumber \\
& \;\;\;\;\;\;\;\;\;\;\;\;\;\;\;\;\;\;\;\;\;\;\;\;\;\;\;\;\;\;\;
\;\;\;\;\;\;\;\;\;\;\;\;\;\;\;\;\;\;\;\;\;\;\;\;\;\;\;\;\;\;\;\;\;
\mathbbm{1}_{\displaystyle\big\{g\big({P\mathstrut}_{{\bf x}, \,{\bf y}, \, \hat{\bf x}}^{}\big)
\; \leq \; \epsilon_{1} \, - \, \mathbb{E}_{{P\mathstrut}_{{\bf x}, \,{\bf y}}^{}}[\ln P(Y \,|\, X)] \, + \, D\big\}}
({P\mathstrut}_{{\bf x}, \,{\bf y}, \, \hat{\bf x}}^{})
\nonumber \\
&
\nonumber \\
\overset{(a)}{\leq} &
\sum_{
\;\;\;{P\mathstrut}_{{\bf x}, \,{\bf y}, \, \hat{\bf x}}^{}\,:\;\;
f({P\mathstrut}_{{\bf x}, \,{\bf y}, \, \hat{\bf x}}^{})
\;\; \leq \;\; R}
\exp\Big\{-nD\big({P\mathstrut}_{{\bf x}, \,{\bf y}}^{}\;\big\|\;Q \circ P\big)\Big\} \,\times
\nonumber \\
& \;\;\;\;\;\;\;\;\;\;\;\;\;\;\;\;\;\;\;\;\;\;\;\;\;\;\;\;\;\;\;
\;\;\;\;\;\;\;\;\;\;\;\;\;\;\;\;\;\;\;\;\;\;\;\;\;\;\;\;\;\;\;\;\;
\mathbbm{1}_{
\displaystyle\big\{g\big({P\mathstrut}_{{\bf x}, \,{\bf y}, \, \hat{\bf x}}^{}\big)
\; \leq \; \epsilon_{1} \, - \, \mathbb{E}_{{P\mathstrut}_{{\bf x}, \,{\bf y}}^{}}[\ln P(Y \,|\, X)] \, + \, D\big\}}
({P\mathstrut}_{{\bf x}, \,{\bf y}, \, \hat{\bf x}}^{})
\nonumber \\
&
\nonumber \\
\overset{(b)}{=} &
\;\;\;
\sum_{
\;\;\;{P\mathstrut}_{{\bf x}, \,{\bf y}, \, \hat{\bf x}}^{}}
\exp\Big\{-nD\big({P\mathstrut}_{{\bf x}, \,{\bf y}}^{}\;\big\|\;Q \circ P\big)\Big\} \,\times
\nonumber \\
& \;\;\;\;\;\;\;\;\;\;\;\;\;\;\;\;\;\;\;
\mathbbm{1}_{\bigg\{\substack{\displaystyle\mathbb{E}_{{P\mathstrut}_{{\bf x}, \,{\bf y}}^{}}[\ln P(Y \,|\, X)]
\, - \,\mathbb{E}_{{P\mathstrut}_{\hat{\bf x}, \,{\bf y}}^{}}[\ln P(Y \,|\, \hat{X})]
\,+\,
f\big({P\mathstrut}_{{\bf x}, \,{\bf y}, \, \hat{\bf x}}^{}\big)
\; \leq \; R \, + \, D \, + \, 3\epsilon_{1} \\
\;\;\;\;\;\;\;\;\;\;\;\;\;\;\;\;\;\;\;\;\;\;\;\;\;\;\;\;\;\;\;\;\;\;\;\;\;\;\;\;\;\;\;\;\;\;\;\;\;\;\;
\;\;\;\;\;\;\;\;\;\;\;\;\;\;\;\;
\displaystyle f\big({P\mathstrut}_{{\bf x}, \,{\bf y}, \, \hat{\bf x}}^{}\big)
\; \leq \; R}\bigg\}}
({P\mathstrut}_{{\bf x}, \,{\bf y}, \, \hat{\bf x}}^{})
\nonumber
\end{align}
\begin{align}
\overset{(c)}{\leq} &
\;\;\;
\sum_{
\;\;\;{P\mathstrut}_{{\bf x}, \,{\bf y}, \, \hat{\bf x}}^{}}
\exp\big\{-n\,\,\widetilde{\!\!E}{\mathstrut}_{1}^{\,\text{types}}(R, D \, + \, 3\epsilon_{1})\big\}
\nonumber \\
\overset{(d)}{\leq} &
\;\;\;
\sum_{
\;\;\;{P\mathstrut}_{{\bf x}, \,{\bf y}, \, \hat{\bf x}}^{}}
\exp\big\{-n\,\,\widetilde{\!\!E}_{1}(R, D \, + \, 3\epsilon_{1})\big\} \;\; \leq \;\;
{(n\, + \, 1)\mathstrut}^{|{\cal X}|\cdot|{\cal X}|\cdot|{\cal Y}|}\,\cdot\,
\exp\big\{-n\,\,\widetilde{\!\!E}_{1}(R, D \, + \, 3\epsilon_{1})\big\},
\label{eqDmin1}
\end{align}
where in\newline
($a$) we use the bound
\begin{equation} \label{eqTypeProb}
\Pr\,\big\{({\bf X\mathstrut}_{m}, {\bf Y}) \, \in \, T({P\mathstrut}_{{\bf x}, \,{\bf y}}^{})\big\}
\; \leq \;
\exp\Big\{-nD\big({P\mathstrut}_{{\bf x}, \,{\bf y}}^{}(x, y)\;\big\|\;Q(x)\cdot P(y \,|\, x)\big)\Big\};
\end{equation}
($b$) collect all the conditions in the indicator function and substitute the definition of $g\big({P\mathstrut}_{{\bf x}, \,{\bf y}, \, \hat{\bf x}}^{}\big)$
(\ref{eqDefG});\newline
($c$) the minimal exponent $\,\,\widetilde{\!\!E}{\mathstrut}_{1}^{\,\text{types}}(R, D \, + \, 3\epsilon_{1})$ is determined by minimization over types ${P\mathstrut}_{{\bf x}, \,{\bf y}, \, \hat{\bf x}}^{}$, corresponding to block length $n$, subject to the two conditions, which appear in the indicator function:
\begin{align}
\,\,\widetilde{\!\!E}{\mathstrut}_{1}^{\,\text{types}}(R, D) \;\; \triangleq \;\; & \min_{{P\mathstrut}_{{\bf x}, \,{\bf y}, \, \hat{\bf x}}^{}(x, \,y, \,\hat{x})} \, D\big({P\mathstrut}_{{\bf x}, \,{\bf y}}^{} \;\big\|\; Q\circ P\big)
\label{eqFirstSumExponentTypes} \\
& \;\;\;\;\;\;
\text{subject to:} \;\;\;\;\;\;\;\;
\mathbb{E}_{{P\mathstrut}_{{\bf x}, \,{\bf y}, \, \hat{\bf x}}^{}}\bigg[\ln \frac{P(Y \,|\, X)}{P(Y \,|\, \hat{X})}\bigg]
\, + \, D\big({P\mathstrut}_{{\bf x}, \,{\bf y}, \, \hat{\bf x}}^{} \,\big\|\, {P\mathstrut}_{{\bf x}, \,{\bf y}}^{} \!\times Q\big)
\; \leq \; R \, + \, D,
\nonumber \\
& \;\;\;\;\;\;\;\;\;\;\;\;\;\;\;\;\;\;\;\;\;\;\;\;\;\;\;\;\;\;\;\;\;\;\;\;\;\;\;\;\;\;\;\;\;\;\;\;\;\;\;\;\;\;\;
\;\;\;\;\;\;\;\;\;\;\;\;\;
D\big({P\mathstrut}_{{\bf x}, \,{\bf y}, \, \hat{\bf x}}^{} \,\big\|\, {P\mathstrut}_{{\bf x}, \,{\bf y}}^{} \!\times Q\big)
\; \leq \; R,
\nonumber
\end{align}
where we use also the definition of $f\big({P\mathstrut}_{{\bf x}, \,{\bf y}, \, \hat{\bf x}}^{}\big)$ (\ref{eqDefF});\newline
($d$) the minimal exponent $\,\,\widetilde{\!\!E}{\mathstrut}_{1}^{\,\text{types}}(R, D \, + \, 3\epsilon_{1})$ is lower-bounded further by the result of the same minimization,
denoted as $\,\,\widetilde{\!\!E}_{1}(R, D \, + \, 3\epsilon_{1})$,
performed over all possible joint distributions
$T(x,y)\cdot W(\hat{x} \, | \, x, y)\,$:
\begin{align}
\,\,\widetilde{\!\!E}_{1}(R, D) \;\; \triangleq \;\; & \min_{T(x,\,y),\,W(\hat{x}\,|\,x,\,y)} \, D(T \;\|\; Q \circ P)
\label{eqFirstSumExponent} \\
& \;\;\;\;\;\;\;\;
\text{subject to:} \;\;\;\;\;\;\;\;
\mathbb{E}_{\,T\,\circ\, W}\bigg[\ln \frac{P(Y \,|\, X)}{P(Y \,|\, \hat{X})}\bigg]
\, + \, D(T \circ W \,\|\, T \times Q)
\; \leq \; R \, + \, D,
\nonumber \\
& \;\;\;\;\;\;\;\;\;\;\;\;\;\;\;\;\;\;\;\;\;\;\;\;\;\;\;\;\;\;\;\;\;\;\;\;\;\;\;\;\;\;\;\;\;\;\;\;\;\;\;\;\;\;\;
\;\;\;\;\;\;\;\;\;\;\;\;\;\;
D(T \circ W \,\|\, T \times Q)
\; \leq \; R.
\nonumber
\end{align}

Finally, the {\em third sum} in (\ref{eqBoundFGH}) is upper-bounded as follows:
\begin{align}
&
\sum_{
\;\;\;{P\mathstrut}_{{\bf x}, \,{\bf y}, \, \hat{\bf x}}^{}\,:\;\;
f({P\mathstrut}_{{\bf x}, \,{\bf y}, \, \hat{\bf x}}^{})
\;\; > \;\; R}
\,\Pr\,\big\{({\bf X\mathstrut}_{m}, {\bf Y}) \, \in \, T({P\mathstrut}_{{\bf x}, \,{\bf y}}^{})\big\}\cdot
\Pr\,\Big\{E\big({P\mathstrut}_{\hat{\bf x} \, | \,{\bf x}, \, {\bf y}}^{}, {\bf X\mathstrut}_{m}, {\bf Y}\big) < +\infty \; \Big| \; {P\mathstrut}_{{\bf x}, \,{\bf y}}^{}\Big\} \, \times
\nonumber \\
& \;\;\;\;\;\;\;\;\;\;\;\;\;\;\;\;\;\;\;\;\;\;\;\;\;\;\;\;\;\;\;
\;\;\;\;\;\;\;\;\;\;\;\;\;\;\;\;\;\;\;\;\;\;\;\;\;\;\;\;\;\;\;\;\;\;
\mathbbm{1}_{\displaystyle\big\{h\big({P\mathstrut}_{{\bf x}, \,{\bf y}, \, \hat{\bf x}}^{}\big)
\; \leq \; \epsilon_{1} \, - \, \mathbb{E}_{{P\mathstrut}_{{\bf x}, \,{\bf y}}^{}}[\ln P(Y \,|\, X)] \, + \, D\big\}}
({P\mathstrut}_{{\bf x}, \,{\bf y}, \, \hat{\bf x}}^{})
\nonumber \\
&
\nonumber \\
\overset{(a)}{\leq} &
\sum_{
\;\;\;{P\mathstrut}_{{\bf x}, \,{\bf y}, \, \hat{\bf x}}^{}\,:\;\;
f({P\mathstrut}_{{\bf x}, \,{\bf y}, \, \hat{\bf x}}^{})
\;\; > \;\; R}
\exp\Big\{-nD\big({P\mathstrut}_{{\bf x}, \,{\bf y}}^{}\;\big\|\;Q \circ P\big)\Big\}\cdot
\exp\Big\{-n\left[D\big({P\mathstrut}_{{\bf x}, \,{\bf y}, \, \hat{\bf x}}^{} \,\big\|\, {P\mathstrut}_{{\bf x}, \,{\bf y}}^{} \!\times Q\big) \, - \, R\right]\Big\}
\, \times
\nonumber \\
& \;\;\;\;\;\;\;\;\;\;\;\;\;\;\;\;\;\;\;\;\;\;\;\;\;\;\;\;\;\;\;
\;\;\;\;\;\;\;\;\;\;\;\;\;\;\;\;\;\;\;\;\;\;\;\;\;\;\;\;\;\;\;\;\;\;
\mathbbm{1}_{\displaystyle\big\{h\big({P\mathstrut}_{{\bf x}, \,{\bf y}, \, \hat{\bf x}}^{}\big)
\; \leq \; \epsilon_{1} \, - \, \mathbb{E}_{{P\mathstrut}_{{\bf x}, \,{\bf y}}^{}}[\ln P(Y \,|\, X)] \, + \, D\big\}}
({P\mathstrut}_{{\bf x}, \,{\bf y}, \, \hat{\bf x}}^{})
\nonumber
\end{align}
\begin{align}
\overset{(b)}{=} &
\sum_{
\;\;\;{P\mathstrut}_{{\bf x}, \,{\bf y}, \, \hat{\bf x}}^{}}
\exp\Big\{-n\left[D\big({P\mathstrut}_{{\bf x}, \,{\bf y}}^{}\;\big\|\;Q \circ P\big)\, + \,
D\big({P\mathstrut}_{{\bf x}, \,{\bf y}, \, \hat{\bf x}}^{} \,\big\|\, {P\mathstrut}_{{\bf x}, \,{\bf y}}^{} \!\times Q\big) \, - \, R
\right]\Big\}
\, \times
\nonumber \\
& \;\;\;\;\;\;\;\;\;\;\;\;\;\;\;\;\;\;\;\;\;\;\;\;\;\;\;\;\;\;\;
\;\;\;\;\;\;\;\;\;\;\;\;
\mathbbm{1}_{
\bigg\{\substack{\displaystyle
\mathbb{E}_{{P\mathstrut}_{{\bf x}, \,{\bf y}}^{}}[\ln P(Y \,|\, X)] \,
- \, \mathbb{E}_{{P\mathstrut}_{\hat{\bf x}, \,{\bf y}}^{}}[\ln P(Y \,|\, \hat{X})]
\; \leq \; D \, + \, \epsilon_{1} \, + \, {\epsilon}_{2}\\
\;\;\;\;\;\;\;\;\;\;\;\;\;\;
\displaystyle \;D\big({P\mathstrut}_{{\bf x}, \,{\bf y}, \, \hat{\bf x}}^{} \,\big\|\, {P\mathstrut}_{{\bf x}, \,{\bf y}}^{} \!\times Q\big)
\; > \; R
}\bigg\}}
({P\mathstrut}_{{\bf x}, \,{\bf y}, \, \hat{\bf x}}^{})
\nonumber \\
\overset{(c)}{\leq} &
\;\;\;
\sum_{
\;\;\;{P\mathstrut}_{{\bf x}, \,{\bf y}, \, \hat{\bf x}}^{}}
\exp\big\{-nE_{2}^{\,\text{types}}(R, D \, + \, \epsilon_{1} \, + \, {\epsilon}_{2})\big\}
\nonumber \\
\overset{(d)}{\leq} &
\;\;\;
\sum_{
\;\;\;{P\mathstrut}_{{\bf x}, \,{\bf y}, \, \hat{\bf x}}^{}}
\exp\big\{-nE_{2}(R, D \, + \, \epsilon_{1} \, + \, {\epsilon}_{2})\big\} \;\; \leq \;\;
{(n\, + \, 1)\mathstrut}^{|{\cal X}|\cdot|{\cal X}|\cdot|{\cal Y}|}\,\cdot\,
\exp\big\{-nE_{2}(R, D \, + \, \epsilon_{1} \, + \, {\epsilon}_{2})\big\},
\label{eqDmin2}
\end{align}
where\newline
($a$) follows by (\ref{eqTypeProb}) and the union bound
\begin{equation} \label{eqUBTypesR}
\Pr\,\Big\{E\big({P\mathstrut}_{\hat{\bf x} \, | \,{\bf x}, \, {\bf y}}^{}, {\bf X\mathstrut}_{m}, {\bf Y}\big) < +\infty \; \Big| \; {P\mathstrut}_{{\bf x}, \,{\bf y}}^{}\Big\}
\; \leq \;
\exp\Big\{-n\left[D\big({P\mathstrut}_{{\bf x}, \,{\bf y}, \, \hat{\bf x}}^{}(x, y, \hat{x}) \,\big\|\, {P\mathstrut}_{{\bf x}, \,{\bf y}}^{}(x, y) \cdot Q(\hat{x})\big) \, - \, R\right]\Big\};
\end{equation}
($b$) uses the definitions of $h\big({P\mathstrut}_{{\bf x}, \,{\bf y}, \, \hat{\bf x}}^{}\big)$ (\ref{eqDefH})
and $f\big({P\mathstrut}_{{\bf x}, \,{\bf y}, \, \hat{\bf x}}^{}\big)$ (\ref{eqDefF});\newline
($c$) the minimal exponent $E_{2}^{\,\text{types}}(R, D \, + \, \epsilon_{1} \, + \, {\epsilon}_{2})$ is determined by minimization over types ${P\mathstrut}_{{\bf x}, \,{\bf y}, \, \hat{\bf x}}^{}$, corresponding to block length $n$, subject to the two conditions, which appear in the indicator function:
\begin{align}
E_{2}^{\,\text{types}}(R, D) \;\; \triangleq \;\; & \min_{{P\mathstrut}_{{\bf x}, \,{\bf y}, \, \hat{\bf x}}^{}(x, \,y, \,\hat{x})} \, \Big\{D\big({P\mathstrut}_{{\bf x}, \,{\bf y}}^{} \;\big\|\; Q\circ P\big)\, + \,
D\big({P\mathstrut}_{{\bf x}, \,{\bf y}, \, \hat{\bf x}}^{} \,\big\|\, {P\mathstrut}_{{\bf x}, \,{\bf y}}^{} \!\times Q\big) \, - \, R\Big\}
\label{eqE2Types} \\
& \;\;\;\;\;\;
\text{subject to:} \;\;\;\;\;\;\;\;
\;\;\;\;\;\;\;\;\;\;\;\;\;\;\;\;\;\;\;\;\;\;\;\,
\mathbb{E}_{{P\mathstrut}_{{\bf x}, \,{\bf y}, \, \hat{\bf x}}^{}}\bigg[\ln \frac{P(Y \,|\, X)}{P(Y \,|\, \hat{X})}\bigg]
\; \leq \; D,
\nonumber \\
& \;\;\;\;\;\;\;\;\;\;\;\;\;\;\;\;\;\;\;\;\;\;\;\;\;\;\;\;\;\;\;\;\;\;\;\;\;\;\;\;\;\;\;\;\;\;\;\;\;\;\;\;\;
D\big({P\mathstrut}_{{\bf x}, \,{\bf y}, \, \hat{\bf x}}^{} \,\big\|\, {P\mathstrut}_{{\bf x}, \,{\bf y}}^{} \!\times Q\big)
\; > \; R;
\nonumber
\end{align}
($d$) $E_{2}^{\,\text{types}}(R, D \, + \, \epsilon_{1} \, + \, {\epsilon}_{2})$ is lower-bounded further by the result of the same minimization,
denoted as $E_{2}(R, D \, + \, \epsilon_{1} \, + \, {\epsilon}_{2})$,
performed over all possible joint distributions
$T(x,y)\cdot W(\hat{x} \, | \, x, y)\,$:
\begin{align}
E_{2}(R, D) \;\; \triangleq \;\; & \min_{T(x,\,y),\,W(\hat{x}\,|\,x,\,y)}\, \Big\{D(T \;\|\; Q \circ P)\, + \,
D(T \circ W \,\|\, T \times Q) \, - \, R\Big\}
\label{eqThirdSumExponent} \\
& \;\;\;\;\;\;\;\;
\text{subject to:} \;\;\;\;\;\;\;\;\;
\;\;\;\;\;\;\;\;\;\;\;\;\;\;\,
\mathbb{E}_{\,T\,\circ\, W}\bigg[\ln \frac{P(Y \,|\, X)}{P(Y \,|\, \hat{X})}\bigg]
\; \leq \; D,
\nonumber \\
& \;\;\;\;\;\;\;\;\;\;\;\;\;\;\;\;\;\;\;\;\;\;\;\;\;\;\;\;\;\;\;\;\;\;\;\;\;\;\;\;\;\;\;\;\;\;\;\;\;\;\;
D(T \circ W \,\|\, T \times Q)
\; \geq \; R.
\nonumber
\end{align}

Comparing the bounds (\ref{eqSecondSum}), (\ref{eqFourthSum}), (\ref{eqDmin1}), and (\ref{eqDmin2}),
we conclude, that, for $n$ sufficiently large, the exponent in the upper bound (\ref{eqBoundFGH})
is lower-bounded by $\,\min\big\{\,\,\widetilde{\!\!E}_{1}(R, D \, + \, 3\epsilon_{1}), \, E_{2}(R, D \, + \, \epsilon_{1} \, + \, {\epsilon}_{2})\big\} \, - \, \epsilon_{1}$.
Since $\epsilon_{1}$ and $\epsilon_{2}$ are arbitrary,
they can be replaced with zeros,
resulting in the following
\begin{thm} \label{thm5}
\begin{align}
& \liminf_{n \, \rightarrow \, \infty} \; \left\{-\frac{1}{n}\ln \Pr \, \{{\cal E}_{m}\}\right\}
\;\; \geq \;\;
\min\big\{\,\,\widetilde{\!\!E}_{1}(R, D), \, E_{2}(R, D)\big\}
\nonumber \\
& = \;\;
\min \left\{
\min_{\substack{T(x,\,y),\,W(\hat{x} \,|\, x,\,y):\\
\\
d(T\,\circ\, W)\,+\,D(T\, \circ\, W \;\|\; T \,\times\, Q)\;\leq\; D \, + \, R\\
\\
D(T\, \circ\, W \;\|\; T \,\times\, Q)\;\leq\;R
}}
\Big\{D(T \; \| \; Q \circ P)\Big\}
\right.,
\nonumber \\
& \;\;\;\;\;\;\;\;\;\;\;\;\;\;\;\;\;\;\;\;\;\;\;\;\,
\left.
\min_{\substack{T(x,\,y),\,W(\hat{x} \,|\, x,\,y):\\
\\
d(T\,\circ\, W) \; \leq \; D\\
\\
D(T\, \circ\, W \;\|\; T \,\times\, Q) \; \geq \; R
}}
\Big\{D(T \; \| \; Q \circ P)\, + \,
D(T \circ W \;\|\; T \times Q) \, - \, R\Big\}
\right\},
\label{eqLowerBoundForney}
\end{align}
{\em where
$\;d(T\,\circ\, W) \; \triangleq \; \mathbb{E}_{\,T\,\circ\, W}\big[d\big((X, Y), \hat{X}\big)\big]\; \triangleq \;\mathbb{E}_{\,T\,\circ\, W}\Big[\ln \frac{P(Y \,|\, X)}{P(Y \,|\, \hat{X})}\Big]$.}
\end{thm}

\section{\bf Explicit lower bound on the random coding error exponent of Forney's decoder} \label{S14}
We use Theorem~\ref{thm5} to prove the following\footnote{This result is redundant, as below we derive the same expression as the true exponent.}
\begin{thm} \label{thm6}
\begin{align}
& \liminf_{n \, \rightarrow \, \infty} \;\left\{ -\frac{1}{n}\ln \Pr \, \{{\cal E}_{m}\}\right\}
\;\; \geq \;\;
\nonumber \\
& 
\min \left\{
\;\;\;
\sup_{\rho\,\geq\,0}
\;\;
\Bigg\{-\inf_{0\,\leq\,s\,\leq\,1} \; \ln \; \sum_{x,\,y}Q(x)P(y\,|\,x)\left[\sum_{\hat{x}} Q(\hat{x})
\left[\frac{P(y \,|\, x)}{P(y \,|\, \hat{x})}\,e^{-D}\right]^{-s}\right]^{\rho} - \; \rho R \Bigg\}
\right.,
\nonumber \\
& \;\;\;\;\;\;\;\;\;\;
\left.\sup_{0\,\leq\,\rho\,\leq\,1}
\Bigg\{-\;\;
\inf_{s\,\geq\,0}\;
\;\;\ln \;
\sum_{x,\,y}Q(x)P(y\,|\,x)\left[\sum_{\hat{x}} Q(\hat{x})
\left[\frac{P(y \,|\, x)}{P(y \,|\, \hat{x})}\,e^{-D}\right]^{-s}\right]^{\rho} - \; \rho R \Bigg\}
\;\right\}.
\nonumber
\end{align}
\end{thm}
\begin{proof}
\begin{align}
&
\,\,\widetilde{\!\!E}_{1}(R, D) \; = \; \min_{\substack{T(x,\,y),\,W(\hat{x} \,|\, x,\,y):\\
\\
d(T\,\circ\, W)\,+\,D(T\, \circ\, W \;\|\; T \,\times\, Q)\;\leq\; D \, + \, R\\
\\
D(T\, \circ\, W \;\|\; T \,\times\, Q)\;\leq\;R
}}
\Big\{D(T \; \| \; Q \circ P)\Big\}
\nonumber \\
& \geq \;\;
\sup_{\substack{\alpha\,\geq\,0,\;\beta\,\geq\,0\\ \alpha\,+\beta\,> \,0}}\;
\min_{\substack{T(x,\,y),\,W(\hat{x} \,|\, x,\,y):\\
\\
d(T\,\circ\, W)\,+\,D(T\, \circ\, W \;\|\; T \,\times\, Q)\;\leq\; D \, + \, R\\
\\
D(T\, \circ\, W \;\|\; T \,\times\, Q)\;\leq\;R
}}
\Big\{D(T \; \| \; Q \circ P)
\nonumber \\
& \;\;\;\;\;\;\;\;\;\;\;\;\;\;\;\;\;\;\;\;\;\;\;\;\;\;\;\;\;\;\;\;\;\;\;\;\;\;\;\;\;\;\;\;
+\,\alpha\big[D(T\, \circ\, W \;\|\; T \,\times\, Q)\,+\, d(T\,\circ\, W)\, - \, D \, - \, R\big]
\nonumber \\
& \;\;\;\;\;\;\;\;\;\;\;\;\;\;\;\;\;\;\;\;\;\;\;\;\;\;\;\;\;\;\;\;\;\;\;\;\;\;\;\;\;\;\;\;
+\,\beta\big[D(T\, \circ\, W \;\|\; T \,\times\, Q) \, - \, R\big]
\Big\}
\nonumber
\end{align}
\begin{align}
& \geq \;\;
\sup_{\substack{\alpha\,\geq\,0,\;\beta\,\geq\,0\\ \alpha\,+\beta\,> \,0}}
\;\min_{T(x,\,y),\,W(\hat{x} \,|\, x,\,y)}
\Big\{D(T \; \| \; Q \circ P)
\nonumber \\
& \;\;\;\;\;\;\;\;\;\;\;\;\;\;\;\;\;\;\;\;\;\;\;\;\;\;\;\;\;\;\;\;\;\;\;\;\;\;\;\;\;\;\;\;
+\,\alpha\big[D(T\, \circ\, W \;\|\; T \,\times\, Q)\,+\,d(T\,\circ\, W) \, - \, D \, - \, R\big]
\nonumber \\
& \;\;\;\;\;\;\;\;\;\;\;\;\;\;\;\;\;\;\;\;\;\;\;\;\;\;\;\;\;\;\;\;\;\;\;\;\;\;\;\;\;\;\;\;
+\,\beta\big[D(T\, \circ\, W \;\|\; T \,\times\, Q) \, - \, R\big]
\Big\}
\nonumber \\
& = \;\;
\sup_{\substack{\alpha\,\geq\,0,\;\beta\,\geq\,0\\ \alpha\,+\beta\,> \,0}}
\;\min_{T(x,\,y)}
\Big\{D(T \; \| \; Q \circ P) \, - \, \alpha D \, - \,(\alpha+\beta)R
\nonumber \\
& \;\;\;\;\;\;\;\;\;\;\;\;\;\;\;\;\;\;\;\;\;\;\;\;\;\;\;\;\;\;\;\;\;\;\;\;\;\;\;\;\;\;\;\;
+\,\min_{W(\hat{x} \,|\, x,\,y)}\big[(\alpha+\beta)D(T\, \circ\, W \;\|\; T \,\times\, Q)\,+\,\alpha d(T\,\circ\, W)\big]\Big\}
\nonumber \\
& = \;\;
\sup_{\substack{\alpha\,\geq\,0,\;\beta\,\geq\,0\\ \alpha\,+\beta\,> \,0}}
\;\min_{T(x,\,y)}
\Bigg\{\sum_{x,\,y}T(x, y)\ln \frac{T(x, y)}{Q(x)P(y\,|\,x)} \, - \, \alpha D \, - \,(\alpha+\beta)R
\nonumber \\
& \;\;\;
+\,\min_{W(\hat{x} \,|\, x,\,y)}\bigg[(\alpha+\beta) \sum_{x,\,y, \,\hat{x}}T(x, y)W(\hat{x} \,|\, x, \, y)\ln \frac{W(\hat{x} \,|\, x, \, y)}{Q(\hat{x})}\,+\,\alpha \sum_{x, \, y, \,\hat{x}}T(x, y)W(\hat{x} \,|\, x, y)d\big((x, y), \hat{x}\big)\bigg]\Bigg\}
\nonumber \\
& = \;\;
\sup_{\substack{\alpha\,\geq\,0,\;\beta\,\geq\,0\\ \alpha\,+\beta\,> \,0}}
\;\min_{T(x,\,y)}
\Bigg\{\sum_{x,\,y}T(x, y)\ln \frac{T(x, y)}{Q(x)P(y\,|\,x)} \, - \, \alpha D \, - \,(\alpha+\beta)R
\nonumber \\
& \;\;\;\;\;\;\;\;\;\;\;\;\;\;\;\;\;\;\;\;\;\;\;\;\;\;\;\;\;\;\;\;\;\;\;\;\;\;\;\;\;\;\;\;
+\,\min_{W(\hat{x} \,|\, x,\,y)}\bigg[(\alpha+\beta) \sum_{x,\,y, \,\hat{x}}T(x, y)W(\hat{x} \,|\, x, \, y)\ln \frac{W(\hat{x} \,|\, x, \, y)}{Q(\hat{x})
{e\mathstrut}^{-\frac{\alpha}{\alpha\,+\,\beta}d((x, \,y), \,\hat{x})}
}\bigg]\Bigg\}
\nonumber \\
& = \;\;
\sup_{\substack{\alpha\,\geq\,0,\;\beta\,\geq\,0\\ \alpha\,+\beta\,> \,0}}
\;\min_{T(x,\,y)}
\Bigg\{\sum_{x,\,y}T(x, y)\ln \frac{T(x, y)}{Q(x)P(y\,|\,x)} \, - \, \alpha D \, - \,(\alpha+\beta)R
\nonumber \\
& \;\;\;\;\;\;\;\;\;\;\;\;\;\;\;\;\;\;\;\;\;\;\;\;\;\;\;\;\;\;\;\;\;\;\;\;\;\;\;\;\;\;\;\;
-\,(\alpha+\beta) \sum_{x,\,y}T(x, y)\ln \sum_{\hat{x}} Q(\hat{x})
{e\mathstrut}^{-\frac{\alpha}{\alpha\,+\,\beta}d((x, \,y), \,\hat{x})}\Bigg\}
\nonumber \\
& = \;\;
\sup_{\substack{\alpha\,\geq\,0,\;\beta\,\geq\,0\\ \alpha\,+\beta\,> \,0}}
\;\min_{T(x,\,y)}
\left\{\sum_{x,\,y}T(x, y)\ln \frac{T(x, y)}{Q(x)P(y\,|\,x)\left[\sum_{\hat{x}} Q(\hat{x})
{e\mathstrut}^{-\frac{\alpha}{\alpha\,+\,\beta}d((x, \,y), \,\hat{x})}\right]^{\alpha\,+\beta}} \, - \, \alpha D \, - \,(\alpha+\beta)R \right\}
\nonumber \\
& = \;\;
\sup_{\substack{\alpha\,\geq\,0,\;\beta\,\geq\,0\\ \alpha\,+\beta\,> \,0}}
\left\{-\ln \sum_{x,\,y}Q(x)P(y\,|\,x)\bigg[\sum_{\hat{x}} Q(\hat{x})
{e\mathstrut}^{-\frac{\alpha}{\alpha\,+\,\beta}\left[d((x, \,y), \,\hat{x})\,-\,D\right]}\bigg]^{\alpha\,+\beta} - \;\; (\alpha+\beta)R \right\}
\label{eqE1BoundAlphaBeta} \\
& = \;\;
\sup_{\rho\,>\,0}
\left\{-\inf_{0\,\leq\,s\,\leq\,1} \; \ln \; \sum_{x,\,y}Q(x)P(y\,|\,x)\bigg[\sum_{\hat{x}} Q(\hat{x})
{e\mathstrut}^{-s\left[d((x, \,y), \,\hat{x})\,-\,D\right]}\bigg]^{\rho} - \; \rho R \right\},
\label{eqE1Bound}
\end{align}
where we define $\;\rho \, \triangleq \, \alpha\,+\,\beta\;$ and $\;s \, \triangleq \, \frac{\alpha}{\alpha\,+\,\beta}$.
The case $\rho \, = \, 0$ can also be included, because it gives bound zero, which is always true.
\begin{align}
&
E_{2}(R, D) \; = \, \min_{\substack{T(x,\,y),\,W(\hat{x} \,|\, x,\,y):\\
\\
d(T\,\circ\, W) \; \leq \; D\\
\\
D(T\, \circ\, W \;\|\; T \,\times\, Q) \; \geq \; R
}}
\Big\{D(T \; \| \; Q \circ P)\, + \,
D(T \circ W \;\|\; T \times Q) \, - \, R\Big\}
\nonumber
\end{align}
\begin{align}
& \geq \;\;
\sup_{\alpha\,\geq\,0}\;\sup_{\beta\,\geq\,0}
\min_{\substack{T(x,\,y),\,W(\hat{x} \,|\, x,\,y):\\
\\
d(T\,\circ\, W) \; \leq \; D\\
\\
D(T\, \circ\, W \;\|\; T \,\times\, Q) \; \geq \; R
}}
\Big\{D(T \; \| \; Q \circ P)\, + \,
D(T \circ W \;\|\; T \times Q) \, - \, R
\nonumber \\
& \;\;\;\;\;\;\;\;\;\;\;\;\;\;\;\;\;\;\;\;\;\;\;\;\;\;\;\;\;\;\;\;
+\,\alpha\big[d(T\,\circ\, W)\, - \, D\big]
\,-\,\beta\big[D(T \circ W \;\|\; T \times Q) \, - \, R\big]\Big\}
\nonumber \\
& \geq \;\;
\sup_{\alpha\,\geq\,0}\;\sup_{\beta\,\geq\,0}\;
\min_{T(x,\,y),\,W(\hat{x} \,|\, x,\,y)}
\Big\{D(T \; \| \; Q \circ P)\, + \,
D(T \circ W \;\|\; T \times Q) \, - \, R
\nonumber \\
& \;\;\;\;\;\;\;\;\;\;\;\;\;\;\;\;\;\;\;\;\;\;\;\;\;\;\;\;\;\;\;\;
+\,\alpha\big[d(T\,\circ\, W)\, - \, D\big]
\,-\,\beta\big[D(T \circ W \;\|\; T \times Q) \, - \, R\big]\Big\}
\nonumber \\
& = \;\;
\sup_{\alpha\,\geq\,0}\;\sup_{\beta\,\geq\,0}\;
\min_{T(x,\,y),\,W(\hat{x} \,|\, x,\,y)}
\Big\{D(T \; \| \; Q \circ P)
\nonumber \\
& \;\;\;\;\;\;\;\;\;\;\;\;\;\;\;\;\;\;\;\;\;\;\;\;\;\;\;\;\;\;\;\;
+\,\alpha\big[d(T\,\circ\, W)\, - \, D\big]
\,+\,(1-\beta)\big[D(T \circ W \;\|\; T \times Q) \, - \, R\big]\Big\}
\nonumber \\
& \geq \;\;
\sup_{\alpha\,\geq\,0}\;\sup_{0\,<\,\beta\,\leq\,1}\;
\min_{T(x,\,y),\,W(\hat{x} \,|\, x,\,y)}
\Big\{D(T \; \| \; Q \circ P)
\nonumber \\
& \;\;\;\;\;\;\;\;\;\;\;\;\;\;\;\;\;\;\;\;\;\;\;\;\;\;\;\;\;\;\;\;
+\,\alpha\big[d(T\,\circ\, W)\, - \, D\big]
\,+\,\beta\big[D(T \circ W \;\|\; T \times Q) \, - \, R\big]\Big\}
\nonumber \\
& \overset{(a)}{=} \;\;
\sup_{\alpha\,\geq\,0}\;\sup_{0\,<\,\beta\,\leq\,1}
\left\{-\ln \sum_{x,\,y}Q(x)P(y\,|\,x)\bigg[\sum_{\hat{x}} Q(\hat{x})
{e\mathstrut}^{-\frac{\alpha}{\beta}\left[d((x, \,y), \,\hat{x})\,-\,D\right]}\bigg]^{\beta} - \;\; \beta R \right\}
\nonumber \\
& \overset{(b)}{=} \;\;
\sup_{0\,<\,\rho\,\leq\,1}
\left\{-\inf_{s\,\geq\,0} \; \ln \; \sum_{x,\,y}Q(x)P(y\,|\,x)\bigg[\sum_{\hat{x}} Q(\hat{x})
{e\mathstrut}^{-s\left[d((x, \,y), \,\hat{x})\,-\,D\right]}\bigg]^{\rho} - \; \rho R \right\},
\label{eqE2Bound}
\end{align}
where ($a$) is obtained by the same steps as (\ref{eqE1BoundAlphaBeta}),
and in ($b$) we define $\rho \, \triangleq \, \beta$ and $s\, \triangleq \, \frac{\alpha}{\beta}$.
Similarly, the case $\rho \, = \, 0$ can also be included.
It remains to substitute $d((x, y), \hat{x})\, = \, \ln \frac{P(y\,|\,x)}{P(y\,|\,\hat{x})}$
into (\ref{eqE1Bound}) and (\ref{eqE2Bound}), and combine them for the final result.
\end{proof}

\section{\bf Upper bound on the random coding error exponent of Forney's decoder} \label{S15}

The sum in Forney's metric (\ref{eqForneyErrorEvent})
can be lower-bounded as follows
\begin{equation} \label{eqForneySumLower}
\sum_{m'\,\neq\,m}P({\bf Y} \,|\, {\bf X\mathstrut}_{m'})
\;\; \geq \;\;
{e\mathstrut}_{}^{-nE({P\mathstrut}_{\hat{\bf x} \, | \,{\bf x}, \, {\bf y}}^{},\, {\bf X\mathstrut}_{m}, {\bf Y})},
\;\;\;\;\;\;\;\;\; \forall \;{P\mathstrut}_{\hat{\bf x} \, | \,{\bf x}, \, {\bf y}}^{},
\end{equation}
where the exponents
$E\big({P\mathstrut}_{\hat{\bf x} \, | \,{\bf x}, \, {\bf y}}^{}, {\bf X\mathstrut}_{m}, {\bf Y}\big)$ are defined as in (\ref{eqTypeExp}).

Using types, we can lower-bound the ensemble average probability of error, given that message $m$ is transmitted, as follows
\begin{align}
& 
\Pr \, \{{\cal E}_{m}\}
\; \geq \max_{\;\;\;{P\mathstrut}_{{\bf x}, \,{\bf y}}^{}}
\,\Pr\,\big\{({\bf X\mathstrut}_{m}, {\bf Y}) \, \in \, T({P\mathstrut}_{{\bf x}, \,{\bf y}}^{})\big\}
\,\cdot\,\Pr\,\big\{{\cal E}_{m} \,\big|\, {P\mathstrut}_{{\bf x}, \,{\bf y}}^{}\big\}
\nonumber \\
& \overset{(a)}{\geq} \max_{\;\;\;{P\mathstrut}_{{\bf x}, \,{\bf y}, \, \hat{\bf x}}^{}}
\,\Pr\,\big\{({\bf X\mathstrut}_{m}, {\bf Y}) \, \in \, T({P\mathstrut}_{{\bf x}, \,{\bf y}}^{})\big\} \,\cdot\,
\Pr\,\bigg\{
\underbrace{E\big({P\mathstrut}_{\hat{\bf x} \, | \,{\bf x}, \, {\bf y}}^{}, {\bf X\mathstrut}_{m}, {\bf Y}\big)
\, < \,  - \, \frac{\ln P({\bf Y} \,|\, {\bf X\mathstrut}_{m})}{n} \, + \, D}_{\subseteq \, {\cal E}_{m}}
 \; \bigg| \; {P\mathstrut}_{{\bf x}, \,{\bf y}}^{}\bigg\}
\nonumber
\end{align}
\begin{align}
& \overset{(b)}{=} \max_{\;\;\;{P\mathstrut}_{{\bf x}, \,{\bf y}, \, \hat{\bf x}}^{}}
\,\Pr\,\big\{({\bf X\mathstrut}_{m}, {\bf Y}) \, \in \, T({P\mathstrut}_{{\bf x}, \,{\bf y}}^{})\big\} \,\cdot\,
\Pr\,\Big\{
E\big({P\mathstrut}_{\hat{\bf x} \, | \,{\bf x}, \, {\bf y}}^{}, {\bf X\mathstrut}_{m}, {\bf Y}\big)
\, < \,  - \, \mathbb{E}\,[\ln P(Y \,|\, X)] \, + \, D
 \; \Big| \; {P\mathstrut}_{{\bf x}, \,{\bf y}}^{}\Big\}
\nonumber \\
& \overset{(c)}{\geq} \; \!\!
\max_{
\;\;\;{P\mathstrut}_{{\bf x}, \,{\bf y}, \, \hat{\bf x}}^{}\,:\;\;
f({P\mathstrut}_{{\bf x}, \,{\bf y}, \, \hat{\bf x}}^{})
\;\; \leq \;\; R\,-\,2\epsilon_{1}}
\,\Pr\,\big\{({\bf X\mathstrut}_{m}, {\bf Y}) \, \in \, T({P\mathstrut}_{{\bf x}, \,{\bf y}}^{})\big\} \cdot 
\Pr\,\Big\{
E\big({P\mathstrut}_{\hat{\bf x} \, | \,{\bf x}, \, {\bf y}}^{}, {\bf X\mathstrut}_{m}, {\bf Y}\big) \, \leq \, g\big({P\mathstrut}_{{\bf x}, \,{\bf y}, \, \hat{\bf x}}^{}\big),
\nonumber \\
& \;\;\;\;\;\;\;\;\;\;\;\;\;\;\;\;\;\;\;\;\;\;\;\;\;\;\;\;\;\;\;\;\;\;\;\;
\;\;\;\;\;\;\;\;\;\;\;\;\;\;\;\;\;\;\;\;\;\;\;\;\;\;\;\;\;\;\;\;
\;\;\;\;\;\;\;\;\;\;\;\;\;\,
g\big({P\mathstrut}_{{\bf x}, \,{\bf y}, \, \hat{\bf x}}^{}\big)
\, < \, - \, \mathbb{E}\,[\ln P(Y \,|\, X)] \, + \, D
\; \Big| \; {P\mathstrut}_{{\bf x}, \,{\bf y}}^{}\Big\}
\nonumber \\
& = \; \!\!
\max_{
\;\;\;{P\mathstrut}_{{\bf x}, \,{\bf y}, \, \hat{\bf x}}^{}\,:\;\;
f({P\mathstrut}_{{\bf x}, \,{\bf y}, \, \hat{\bf x}}^{})
\;\; \leq \;\; R\,-\,2\epsilon_{1}}
\,\Pr\,\big\{({\bf X\mathstrut}_{m}, {\bf Y}) \, \in \, T({P\mathstrut}_{{\bf x}, \,{\bf y}}^{})\big\}\,\times
\nonumber \\
& \;\;\;\;\;\;\;\;\;\;\;\;\;\;\;\;\;\;\;\;\;\;\;\;\;\;\;\;\;\;\;\;
\;\;\;\;\;\;\;\;\;\;\;\;\;\;\;\;\;\;\;\;\;\;\;\;\;\;\;\;\;\;\;\;\;\;\;\;\;\;
\bigg[1 \; - \; \Pr\,\Big\{
E\big({P\mathstrut}_{\hat{\bf x} \, | \,{\bf x}, \, {\bf y}}^{}, {\bf X\mathstrut}_{m}, {\bf Y}\big) \, > \, g\big({P\mathstrut}_{{\bf x}, \,{\bf y}, \, \hat{\bf x}}^{}\big)\; \Big| \; {P\mathstrut}_{{\bf x}, \,{\bf y}}^{}\Big\}\bigg]\,\times
\nonumber \\
& \;\;\;\;\;\;\;\;\;\;\;\;\;\;\;\;\;\;\;\;\;\;\;\;\;\;\;\;\;\;\;\;\;\;\;\;
\;\;\;\;\;\;\;\;\;\;\;\;\;\;\;\;\;\;\;\;\;\;\;\;\;\;\;\;\;\;\;\;
\;\;\;\;\;\;\;\;\;\,
\mathbbm{1}_{\displaystyle\big\{g\big({P\mathstrut}_{{\bf x}, \,{\bf y}, \, \hat{\bf x}}^{}\big)
\, < \, - \, \mathbb{E}\,[\ln P(Y \,|\, X)] \, + \, D\big\}}
({P\mathstrut}_{{\bf x}, \,{\bf y}, \, \hat{\bf x}}^{}).
\label{eqBoundFG}
\end{align}
Explanation of steps:\newline
($a$) follows by the definition of the error event ${\cal E}_{m}$ (\ref{eqForneyErrorEvent}) and the lower bound on Forney's sum (\ref{eqForneySumLower});\newline
($b$) uses notation (\ref{eqExpectation});\newline
($c$) holds for any functions $f\big({P\mathstrut}_{{\bf x}, \,{\bf y}, \, \hat{\bf x}}^{}\big)$ and $g\big({P\mathstrut}_{{\bf x}, \,{\bf y}, \, \hat{\bf x}}^{}\big)$, because
\begin{align}
\Big\{
E\big({P\mathstrut}_{\hat{\bf x} \, | \,{\bf x}, \, {\bf y}}^{}, {\bf X\mathstrut}_{m}, {\bf Y}\big) \, \leq \, g\big({P\mathstrut}_{{\bf x}, \,{\bf y}, \, \hat{\bf x}}^{}\big), \;\;
g\big({P\mathstrut}_{{\bf x}, \,{\bf y}, \, \hat{\bf x}}^{}\big)
\, & < \, - \, \mathbb{E}\,[\ln P(Y \,|\, X)] \, + \, D\Big\}
\;\; \Rightarrow
\nonumber \\
\Big\{
E\big({P\mathstrut}_{\hat{\bf x} \, | \,{\bf x}, \, {\bf y}}^{}, {\bf X\mathstrut}_{m}, {\bf Y}\big)
\, & < \, - \, \mathbb{E}\,[\ln P(Y \,|\, X)] \, + \, D\Big\}.
\nonumber
\end{align}
We use also another version of ($c$),
written with functions
$f\big({P\mathstrut}_{{\bf x}, \,{\bf y}, \, \hat{\bf x}}^{}\big)$ and $h\big({P\mathstrut}_{{\bf x}, \,{\bf y}, \, \hat{\bf x}}^{}\big)$
as
\begin{align}
& 
\Pr \, \{{\cal E}_{m}\} \; \geq
\nonumber \\
& \max_{
\;\;\;{P\mathstrut}_{{\bf x}, \,{\bf y}, \, \hat{\bf x}}^{}\,:\;\;
f({P\mathstrut}_{{\bf x}, \,{\bf y}, \, \hat{\bf x}}^{})
\;\; \geq \;\; R}
\,\Pr\,\big\{({\bf X\mathstrut}_{m}, {\bf Y}) \, \in \, T({P\mathstrut}_{{\bf x}, \,{\bf y}}^{})\big\} \cdot 
\Pr\,\Big\{
E\big({P\mathstrut}_{\hat{\bf x} \, | \,{\bf x}, \, {\bf y}}^{}, {\bf X\mathstrut}_{m}, {\bf Y}\big) \, \leq \, h\big({P\mathstrut}_{{\bf x}, \,{\bf y}, \, \hat{\bf x}}^{}\big),
\nonumber \\
& \;\;\;\;\;\;\;\;\;\;\;\;\;\;\;\;\;\;\;\;\;\;\;\;\;\;\;\;\;\;\;\;\;\;\;\;
\;\;\;\;\;\;\;\;\;\;\;\;\;\;\;\;\;\;\;\;\;\;\;\;\;\;\;\;
\;\;\;\;\;\;\;\;\;\;\;\;\;\,
h\big({P\mathstrut}_{{\bf x}, \,{\bf y}, \, \hat{\bf x}}^{}\big)
\, < \, - \, \mathbb{E}\,[\ln P(Y \,|\, X)] \, + \, D
\; \Big| \; {P\mathstrut}_{{\bf x}, \,{\bf y}}^{}\Big\} \; =
\nonumber \\
& \max_{
\;\;\;{P\mathstrut}_{{\bf x}, \,{\bf y}, \, \hat{\bf x}}^{}\,:\;\;
f({P\mathstrut}_{{\bf x}, \,{\bf y}, \, \hat{\bf x}}^{})
\;\; \geq \;\; R}
\,\Pr\,\big\{({\bf X\mathstrut}_{m}, {\bf Y}) \, \in \, T({P\mathstrut}_{{\bf x}, \,{\bf y}}^{})\big\}\cdot
\Pr\,\Big\{
E\big({P\mathstrut}_{\hat{\bf x} \, | \,{\bf x}, \, {\bf y}}^{}, {\bf X\mathstrut}_{m}, {\bf Y}\big) \, \leq \, h\big({P\mathstrut}_{{\bf x}, \,{\bf y}, \, \hat{\bf x}}^{}\big)\; \Big| \; {P\mathstrut}_{{\bf x}, \,{\bf y}}^{}\Big\}\,\times
\nonumber \\
& \;\;\;\;\;\;\;\;\;\;\;\;\;\;\;\;\;\;\;\;\;\;\;\;\;\;\;\;\;\;\;\;\;\;\;\;
\;\;\;\;\;\;\;\;\;\;\;\;\;\;\;\;\;\;\;\;\;\;\;\;\;\;\;\;
\;\;\;\;\;\;\;\;\;\,
\mathbbm{1}_{\displaystyle\big\{h\big({P\mathstrut}_{{\bf x}, \,{\bf y}, \, \hat{\bf x}}^{}\big)
\, < \, - \, \mathbb{E}\,[\ln P(Y \,|\, X)] \, + \, D\big\}}
({P\mathstrut}_{{\bf x}, \,{\bf y}, \, \hat{\bf x}}^{}).
\label{eqBoundFH}
\end{align}
The lower bounds (\ref{eqBoundFG}), (\ref{eqBoundFH}) were constructed for the use with the
following lemma
\begin{lemma} \label{lemma10}
{\em Let $Z_{i}\,\sim\, \text{i.i.d}\;\text{Bernoulli}\left({e\mathstrut}^{-nI}\right)$,
$i \, = \, 1, \, 2, \, ... \, , \, {e\mathstrut}^{nR}$.}

{\em If $I\,\leq\,R \, - \, \epsilon$, with $\epsilon \, > \, 0$, then}
\begin{equation} \label{eqSPpart}
\Pr\,\Bigg\{\sum_{i \, = \, 1}^{{e\mathstrut}^{nR}}Z_{i} \; < \; {e\mathstrut}^{n(R\,-\,I\,-\,\epsilon)}\Bigg\}
\;\; < \;\; \frac{{e\mathstrut}^{-n\epsilon}}{{(1\,-\,{e\mathstrut}^{-n\epsilon})\mathstrut}^{2}}.
\end{equation}

{\em If $I\,\geq\,R$, then}
\begin{equation} \label{eqLowSlopePart}
\Pr\,\Bigg\{\sum_{i \, = \, 1}^{{e\mathstrut}^{nR}}Z_{i} \; \geq \; 1\Bigg\}
\;\; > \;\; {e\mathstrut}^{-n(I\,-\,R)}\,\cdot\,
\underbrace{
{\Big(1\,-\,{e\mathstrut}^{-nR}\Big)\mathstrut}^{{e\mathstrut}^{nR}}}_{ \rightarrow \, 1/e}.
\end{equation}
\end{lemma}
\begin{proof}
For $I\,\leq\,R \, - \, \epsilon$
\begin{align}
\Pr\,\Bigg\{\sum_{i \, = \, 1}^{{e\mathstrut}^{nR}}Z_{i} \; < \; {e\mathstrut}^{n(R\,-\,I\,-\,\epsilon)}\Bigg\}
\;\; & = \;\;
\Pr\,\Bigg\{\sum_{i \, = \, 1}^{{e\mathstrut}^{nR}}\big(Z_{i}\,-\,{e\mathstrut}^{-nI}\big) \; < \; {e\mathstrut}^{n(R\,-\,I\,-\,\epsilon)}\,-\,{e\mathstrut}^{n(R\,-\,I)}\Bigg\}
\nonumber \\
& = \;\;
\Pr\,\Bigg\{\sum_{i \, = \, 1}^{{e\mathstrut}^{nR}}\big(Z_{i}\,-\,{e\mathstrut}^{-nI}\big) \; < \;
{e\mathstrut}^{n(R\,-\,I)}(
{e\mathstrut}^{-n\epsilon}-\,1)\Bigg\}
\nonumber \\
& \leq \;\;
\Pr\,\Bigg\{\Bigg|\sum_{i \, = \, 1}^{{e\mathstrut}^{nR}}\big(Z_{i}\,-\,{e\mathstrut}^{-nI}\big)\Bigg| \; > \;
{e\mathstrut}^{n(R\,-\,I)}(
1 \, - \, {e\mathstrut}^{-n\epsilon})\Bigg\}
\nonumber \\
& \overset{(*)}{\leq} \;\;
\frac{1 \, - \, {e\mathstrut}^{-nI}}{{(1\,-\,{e\mathstrut}^{-n\epsilon})\mathstrut}^{2}}\cdot{e\mathstrut}^{-n(R\,-\,I)}
\;\; < \;\;
\frac{{e\mathstrut}^{-n(R\,-\,I)}}{{(1\,-\,{e\mathstrut}^{-n\epsilon})\mathstrut}^{2}}
\;\; \leq \;\;
\frac{{e\mathstrut}^{-n\epsilon}}{{(1\,-\,{e\mathstrut}^{-n\epsilon})\mathstrut}^{2}},
\nonumber
\end{align}
where
($*$)
is Chebyshev's inequality.

For $I\,\geq\,R$
\begin{align}
\Pr\,\Bigg\{\sum_{i \, = \, 1}^{{e\mathstrut}^{nR}}Z_{i} \; \geq \; 1\Bigg\}
\;\; & > \;\;
\sum_{i \, = \, 1}^{{e\mathstrut}^{nR}}
\Pr\,\big\{Z_{i} \; = \; 1\big\}
\prod_{j \, \neq \, i}
\Pr\,\big\{Z_{j} \; = \; 0\big\}
\nonumber \\
& = \;\;
{e\mathstrut}^{nR}
{e\mathstrut}^{-nI}
{\Big(1\,-\,{e\mathstrut}^{-nI}\Big)\mathstrut}^{{e\mathstrut}^{nR}-\,1}
\;\;
> \;\;
{e\mathstrut}^{nR}
{e\mathstrut}^{-nI}
{\Big(1\,-\,{e\mathstrut}^{-nR}\Big)\mathstrut}^{{e\mathstrut}^{nR}}.
\nonumber
\end{align}
\end{proof}

Let $f\big({P\mathstrut}_{{\bf x}, \,{\bf y}, \, \hat{\bf x}}^{}\big)$ and $I$ be defined as in
(\ref{eqDefF}) and (\ref{eqCondType}).
If $f\big({P\mathstrut}_{{\bf x}, \,{\bf y}, \, \hat{\bf x}}^{}\big) \, \leq \, R \, - \, 2\epsilon_{1}$,
then for $n$ sufficiently large, as in (\ref{eqNLarge}),  we obtain by (\ref{eqIandF}):
$I \; \leq \; f\big({P\mathstrut}_{{\bf x}, \,{\bf y}, \, \hat{\bf x}}^{}\big)\,+\,\epsilon_{1} \; \leq \; R \, - \, \epsilon_{1}$.
For such $n$, the first part of Lemma~\ref{lemma10} holds for the following:
\begin{align}
& \Pr\,\bigg\{E\big({P\mathstrut}_{\hat{\bf x} \, | \,{\bf x}, \, {\bf y}}^{}, {\bf X\mathstrut}_{m}, {\bf Y}\big)
\; > \; \underbrace{- \mathbb{E}_{{P\mathstrut}_{\hat{\bf x}, \,{\bf y}}^{}}[\ln P(Y \,|\, \hat{X})]
 - R + f({P\mathstrut}_{{\bf x}, \,{\bf y}, \, \hat{\bf x}}^{}) + 2\epsilon_{1}}_{g({P\mathstrut}_{{\bf x}, \,{\bf y}, \, \hat{\bf x}}^{})}
\; \bigg| \;
({\bf X\mathstrut}_{m}, {\bf Y}) \, \in \, T({P\mathstrut}_{{\bf x}, \,{\bf y}}^{})\bigg\}
\nonumber \\
\overset{(a)}{=} \;\; & \Pr\,\bigg\{\sum_{m' \, \neq \, m}P({\bf Y} \,|\, {\bf X\mathstrut}_{m'})\cdot
\mathbbm{1}_{\displaystyle\big\{{\bf X\mathstrut}_{m'} \; \in \; T\big({P\mathstrut}_{\hat{\bf x} \, | \,{\bf x}, \, {\bf y}}^{}, \, {\bf X\mathstrut}_{m}, {\bf Y}\big)\big\}}(m')
\; <
\nonumber \\
& \;\;\;\;\;\;\;\;\;\;\;\;\;\;\;\;\;\;\;\;\;\;\;\;\,\,
\exp\Big\{n\big(\mathbb{E}_{{P\mathstrut}_{\hat{\bf x}, \,{\bf y}}^{}}[\ln P(Y \,|\, \hat{X})] \, + \, R\,-\,f({P\mathstrut}_{{\bf x}, \,{\bf y}, \, \hat{\bf x}}^{})\,-\,2\epsilon_{1}\big)\Big\}
\; \bigg| \;
({\bf X\mathstrut}_{m}, {\bf Y}) \, \in \, T({P\mathstrut}_{{\bf x}, \,{\bf y}}^{})\bigg\}
\nonumber \\
= \;\; & \Pr\,\bigg\{\sum_{m' \, \neq \, m}
\mathbbm{1}_{\displaystyle\big\{{\bf X\mathstrut}_{m'} \; \in \; T\big({P\mathstrut}_{\hat{\bf x} \, | \,{\bf x}, \, {\bf y}}^{}, \, {\bf X\mathstrut}_{m}, {\bf Y}\big)\big\}}(m')
\; < \; {e\mathstrut}^{n\big(R\,-\,f({P\mathstrut}_{{\bf x}, \,{\bf y}, \, \hat{\bf x}}^{})\,-\,2\epsilon_{1}\big)}
\; \bigg| \;
({\bf X\mathstrut}_{m}, {\bf Y}) \, \in \, T({P\mathstrut}_{{\bf x}, \,{\bf y}}^{})\bigg\}
\nonumber \\
\overset{(b)}{\leq} \;\; & \Pr\,\Bigg\{\sum_{m' \, = \, 1}^{{e\mathstrut}^{nR}}
\mathbbm{1}_{\displaystyle\big\{{\bf X\mathstrut}_{m'} \; \in \; T\big({P\mathstrut}_{\hat{\bf x} \, | \,{\bf x}, \, {\bf y}}^{}, \, {\bf X\mathstrut}_{m}, {\bf Y}\big)\big\}}(m')
\; < \; {e\mathstrut}^{n(R\,-\,I\,-\,\epsilon_{1})}
\; \Bigg| \;
({\bf X\mathstrut}_{m}, {\bf Y}) \, \in \, T({P\mathstrut}_{{\bf x}, \,{\bf y}}^{})\Bigg\}
\nonumber \\
\overset{(c)}{<} \;\; &  \frac{{e\mathstrut}^{-n\epsilon_{1}}}{{(1\,-\,{e\mathstrut}^{-n\epsilon_{1}})\mathstrut}^{2}},
\;\;\;\;\;\;\;\;\;
\;\;\;\;\;\;\;\;\;
\frac{|{\cal X}||{\cal X}||{\cal Y}|\ln(n + 1)}{n}
\; \leq \; \epsilon_{1},
\label{eqLemmaPartOne}
\end{align}
where in \newline
($a$) the definition of $E\big({P\mathstrut}_{\hat{\bf x} \, | \,{\bf x}, \, {\bf y}}^{}, {\bf X\mathstrut}_{m}, {\bf Y}\big)$
(\ref{eqTypeExp}) is used, and notation (\ref{eqExpectation}) with ${P\mathstrut}_{\hat{\bf x}, \,{\bf y}}^{}$;\newline
($b$) the codebook size is assumed to be
$M\,=\,{e\mathstrut}^{nR}\,+\,1$;
the RHS of the inequality is increased 
by substitution of $I \, \leq \, f\big({P\mathstrut}_{{\bf x}, \,{\bf y}, \, \hat{\bf x}}^{}\big) \, + \, \epsilon_{1}$;
\newline
($c$) holds by the definition of $I$ (\ref{eqCondType}) and the first statement of the lemma (\ref{eqSPpart}).\newline
If we choose
\begin{equation} \label{eqDeG}
g\big({P\mathstrut}_{{\bf x}, \,{\bf y}, \, \hat{\bf x}}^{}\big) \;\; \triangleq \;\;
- \mathbb{E}_{{P\mathstrut}_{\hat{\bf x}, \,{\bf y}}^{}}[\ln P(Y \,|\, \hat{X})] \, - \, R \, + \,
f\big({P\mathstrut}_{{\bf x}, \,{\bf y}, \, \hat{\bf x}}^{}\big)
\, + \, 2\epsilon_{1},
\end{equation}
then with (\ref{eqLemmaPartOne}) the lower bound (\ref{eqBoundFG}) becomes
\begin{align}
& \max_{
\;\;\;{P\mathstrut}_{{\bf x}, \,{\bf y}, \, \hat{\bf x}}^{}}
\,\Pr\,\big\{({\bf X\mathstrut}_{m}, {\bf Y}) \, \in \, T({P\mathstrut}_{{\bf x}, \,{\bf y}}^{})\big\}\,\cdot\,
\bigg[1 \; - \; \Pr\,\Big\{
E\big({P\mathstrut}_{\hat{\bf x} \, | \,{\bf x}, \, {\bf y}}^{}, {\bf X\mathstrut}_{m}, {\bf Y}\big) \, > \, g\big({P\mathstrut}_{{\bf x}, \,{\bf y}, \, \hat{\bf x}}^{}\big)\; \Big| \; {P\mathstrut}_{{\bf x}, \,{\bf y}}^{}\Big\}\bigg]\,\times
\nonumber \\
& \;\;\;\;\;\;\;\;\;\;\;\;\;\;\;\;\;\;\;\;\;\;\;\;\;\;\;\;\;\;\;\;\;\;\;\;\;\;\;\;\;\;\;\;\;\;\;\;
\;\;\;\;\;\;\;\;\;\;\;\;\;\;\;\;\;
\mathbbm{1}_{\bigg\{\substack{\displaystyle g\big({P\mathstrut}_{{\bf x}, \,{\bf y}, \, \hat{\bf x}}^{}\big)
\; < \; - \, \mathbb{E}_{{P\mathstrut}_{{\bf x}, \,{\bf y}}^{}}[\ln P(Y \,|\, X)] \, + \, D \\
\displaystyle f\big({P\mathstrut}_{{\bf x}, \,{\bf y}, \, \hat{\bf x}}^{}\big)
\; \leq \;  R\,-\,2\epsilon_{1}
\;\;\;\;\;\;\;\;\;\;\;\;\;\;\;\;\;\;\;\;\;\;\;\;\;\;\;
}\bigg\}}
({P\mathstrut}_{{\bf x}, \,{\bf y}, \, \hat{\bf x}}^{})
\nonumber \\
& \overset{(a)}{\geq} \; \!\!
\max_{
\;\;\;{P\mathstrut}_{{\bf x}, \,{\bf y}, \, \hat{\bf x}}^{}}
{(n\, + \, 1)\mathstrut}^{-|{\cal X}|\cdot|{\cal Y}|}\cdot
\exp\Big\{-nD\big({P\mathstrut}_{{\bf x}, \,{\bf y}}^{}\;\big\|\;Q \circ P\big)\Big\}\,\cdot \,
\left[1 \; - \; \frac{{e\mathstrut}^{-n\epsilon_{1}}}{{(1\,-\,{e\mathstrut}^{-n\epsilon_{1}})\mathstrut}^{2}}\right]\,\times
\nonumber \\
& \;\;\;\;\;\;\;\;\;\;\;\;\;\;\;\;\;
\mathbbm{1}_{\bigg\{\substack{\displaystyle\mathbb{E}_{{P\mathstrut}_{{\bf x}, \,{\bf y}}^{}}[\ln P(Y \,|\, X)]
\, - \,\mathbb{E}_{{P\mathstrut}_{\hat{\bf x}, \,{\bf y}}^{}}[\ln P(Y \,|\, \hat{X})]
\,+\,
f\big({P\mathstrut}_{{\bf x}, \,{\bf y}, \, \hat{\bf x}}^{}\big)
\; < \; R \, + \, D \, - \, 2\epsilon_{1} \\
\;\;\;\;\;\;\;\;\;\;\;\;\;\;\;\;\;\;\;\;\;\;\;\;\;\;\;\;\;\;\;\;\;\;\;\;\;\;\;\;\;\;\;\;\;\;\;\;\;\;\;
\;\;\;\;\;\;\;\;\;\;\;\;\;\;\;\;\;\;\;\;\;\;\;\;\;\;\;\;\;\;\;
\displaystyle f\big({P\mathstrut}_{{\bf x}, \,{\bf y}, \, \hat{\bf x}}^{}\big)
\; \leq \; R \, - \, 2\epsilon_{1}}\bigg\}}
({P\mathstrut}_{{\bf x}, \,{\bf y}, \, \hat{\bf x}}^{})
\nonumber \\
&
\nonumber \\
& \overset{(b)}{\geq} \;
\exp\big\{-n\big[\,\,\widetilde{\!\!E}{\mathstrut}_{1}^{\,\text{types}}(R \, - \, 2\epsilon_{1}, \, D)\, + \, \epsilon_{2}\big]\big\}
\;\; \overset{(c)}{\geq} \;\; \exp\big\{-n\big[\,\,\widetilde{\!\!E}{\mathstrut}_{1}(R \, - \, 2\epsilon_{1} \, - \, \epsilon_{3}, \, D)\, + \, \epsilon_{2} \, + \, \epsilon_{3}\big]\big\}.
\label{eqTermOne}
\end{align}
Explanation of steps:\newline
($a$) follows by the lower bound on the probability of the joint type
\begin{equation} \label{eqTypeProbBelow}
\Pr\,\big\{({\bf X\mathstrut}_{m}, {\bf Y}) \, \in \, T({P\mathstrut}_{{\bf x}, \,{\bf y}}^{})\big\}
\; \geq \;
{(n\, + \, 1)\mathstrut}^{-|{\cal X}|\cdot|{\cal Y}|}\cdot\exp\Big\{-nD\big({P\mathstrut}_{{\bf x}, \,{\bf y}}^{}(x, y)\;\big\|\;Q(x)\cdot P(y \,|\, x)\big)\Big\},
\end{equation}
and (\ref{eqLemmaPartOne}), (\ref{eqDeG});\newline
($b$) holds for $n$ sufficiently large for a given $\epsilon_{2}\, > \, 0$, and uses the definition of the minimal exponent
similar to
(\ref{eqFirstSumExponentTypes}):
\begin{align}
\,\,\widetilde{\!\!E}{\mathstrut}_{1}^{\,\text{types}}(R, D) \;\; \triangleq \;\; & \min_{{P\mathstrut}_{{\bf x}, \,{\bf y}, \, \hat{\bf x}}^{}(x, \,y, \,\hat{x})} \, D\big({P\mathstrut}_{{\bf x}, \,{\bf y}}^{} \;\big\|\; Q\circ P\big)
\nonumber \\
& \;\;\;\;\;\;
\text{subject to:} \;\;\;\;\;\;\;\;
\mathbb{E}_{{P\mathstrut}_{{\bf x}, \,{\bf y}, \, \hat{\bf x}}^{}}\bigg[\ln \frac{P(Y \,|\, X)}{P(Y \,|\, \hat{X})}\bigg]
\, + \, D\big({P\mathstrut}_{{\bf x}, \,{\bf y}, \, \hat{\bf x}}^{} \,\big\|\, {P\mathstrut}_{{\bf x}, \,{\bf y}}^{} \!\times Q\big)
\; < \; R \, + \, D,
\nonumber \\
& \;\;\;\;\;\;\;\;\;\;\;\;\;\;\;\;\;\;\;\;\;\;\;\;\;\;\;\;\;\;\;\;\;\;\;\;\;\;\;\;\;\;\;\;\;\;\;\;\;\;\;\;\;\;\;
\;\;\;\;\;\;\;\;\;\;\;\;\;
D\big({P\mathstrut}_{{\bf x}, \,{\bf y}, \, \hat{\bf x}}^{} \,\big\|\, {P\mathstrut}_{{\bf x}, \,{\bf y}}^{} \!\times Q\big)
\; \leq \; R,
\nonumber
\end{align}
and the definition of $f\big({P\mathstrut}_{{\bf x}, \,{\bf y}, \, \hat{\bf x}}^{}\big)$ (\ref{eqDefF}).\newline
($c$) Let ${T\mathstrut}^{*}\circ {W\mathstrut}^{*}$ denote the joint distribution, achieving $\,\,\widetilde{\!\!E}_{1}(R \, - \, 2\epsilon_{1} \, - \, \epsilon_{3}, \, D)$,
defined by (\ref{eqFirstSumExponent}), for some $\epsilon_{3} \, > \, 0$.
This implies
\begin{align}
D\big({T\mathstrut}^{*} \;\|\; Q \circ P\big) \;\; & = \;\; \,\,\widetilde{\!\!E}_{1}(R \, - \, 2\epsilon_{1} \, - \, \epsilon_{3}, \, D),
\label{eqAchieves} \\
\mathbb{E}_{\,{T\mathstrut}^{*}\circ\, {W\mathstrut}^{*}}\bigg[\ln \frac{P(Y \,|\, X)}{P(Y \,|\, \hat{X})}\bigg]
\, + \, D\big({T\mathstrut}^{*} \circ {W\mathstrut}^{*} \,\|\, {T\mathstrut}^{*} \times Q\big)
\;\; & \leq \;\; R \, + \, D \, - \, 2\epsilon_{1} \, - \, \epsilon_{3},
\nonumber \\
D\big({T\mathstrut}^{*} \circ {W\mathstrut}^{*} \,\|\, {T\mathstrut}^{*} \times Q\big)
\;\; & \leq \;\; R \, - \, 2\epsilon_{1} \, - \, \epsilon_{3}.
\nonumber
\end{align}
Let ${T\mathstrut}_{n}^{*}\circ {W\mathstrut}_{n}^{*}$
denote a quantized version of the joint distribution ${T\mathstrut}^{*} \circ {W\mathstrut}^{*}$
with precision $\frac{1}{n}$, i.e. a joint type with denominator $n$.
Note, that the divergences, as functions of $T\circ W$, have bounded derivatives, and also the ratio $\ln \frac{P(y \,|\, x)}{P(y \,|\, \hat{x})}$
is bounded. Therefore, for any $\epsilon_{3}\,>\,0$ there exists $n$ large enough,
such that the quantized distribution ${T\mathstrut}_{n}^{*}\circ {W\mathstrut}_{n}^{*}$
satisfies
\begin{align}
D\big({T\mathstrut}_{n}^{*} \;\|\; Q \circ P\big) \;\; & \leq \;\; D({T\mathstrut}^{*} \;\|\; Q \circ P) \, + \, \epsilon_{3},
\label{eqContDist} \\
\mathbb{E}_{\,{T\mathstrut}_{n}^{*}\,\circ\, {W\mathstrut}_{n}^{*}}\bigg[\ln \frac{P(Y \,|\, X)}{P(Y \,|\, \hat{X})}\bigg]
\, + \, D\big({T\mathstrut}_{n}^{*} \circ {W\mathstrut}_{n}^{*} \,\|\, {T\mathstrut}_{n}^{*} \times Q\big)
\;\; & < \;\; R \, + \, D \, - \, 2\epsilon_{1},
\nonumber \\
D\big({T\mathstrut}_{n}^{*} \circ {W\mathstrut}_{n}^{*} \,\|\, {T\mathstrut}_{n}^{*} \times Q\big)
\;\; & \leq \;\; R \, - \, 2\epsilon_{1}.
\nonumber
\end{align}
It follows from the last two inequalities that for $n$ sufficiently large
\begin{equation} \label{eqQWay}
D\big({T\mathstrut}_{n}^{*} \;\|\; Q \circ P\big) \;\; \geq \;\; \,\,\widetilde{\!\!E}{\mathstrut}_{1}^{\,\text{types}}(R \, - \, 2\epsilon_{1}, \, D).
\end{equation}
The relations (\ref{eqQWay}), (\ref{eqContDist}), (\ref{eqAchieves}) together give
\begin{displaymath}
\,\,\widetilde{\!\!E}_{1}(R \, - \, 2\epsilon_{1} \, - \, \epsilon_{3}, \, D) \, + \, \epsilon_{3} \;\; \geq \;\; \,\,\widetilde{\!\!E}{\mathstrut}_{1}^{\,\text{types}}(R \, - \, 2\epsilon_{1}, \, D).
\end{displaymath}
This explains ($c$).

Now we return to the bound (\ref{eqBoundFH}).
If $f\big({P\mathstrut}_{{\bf x}, \,{\bf y}, \, \hat{\bf x}}^{}\big) \, \geq \, R$,
then also $I \, \geq \, R$ by (\ref{eqIandF}), and the second part of Lemma~\ref{lemma10} holds for the following:
\begin{align}
& \Pr\,\bigg\{E\big({P\mathstrut}_{\hat{\bf x} \, | \,{\bf x}, \, {\bf y}}^{}, {\bf X\mathstrut}_{m}, {\bf Y}\big)
\; \leq \; \underbrace{- \mathbb{E}_{{P\mathstrut}_{\hat{\bf x}, \,{\bf y}}^{}}[\ln P(Y \,|\, \hat{X})]}_{h({P\mathstrut}_{{\bf x}, \,{\bf y}, \, \hat{\bf x}}^{})}
\; \bigg| \;
({\bf X\mathstrut}_{m}, {\bf Y}) \, \in \, T({P\mathstrut}_{{\bf x}, \,{\bf y}}^{})\bigg\}
\nonumber \\
\overset{(a)}{=} \;\; & \Pr\,\bigg\{\sum_{m' \, \neq \, m}P({\bf Y} \,|\, {\bf X\mathstrut}_{m'})\cdot
\mathbbm{1}_{\displaystyle\big\{{\bf X\mathstrut}_{m'} \; \in \; T\big({P\mathstrut}_{\hat{\bf x} \, | \,{\bf x}, \, {\bf y}}^{}, \, {\bf X\mathstrut}_{m}, {\bf Y}\big)\big\}}(m')
\; \geq
\nonumber \\
& \;\;\;\;\;\;\;\;\;\;\;\;\;\;\;\;\;\;\;\;\;\;\;\;\;\;\;\;\;\;\;\;\;\;
\exp\big\{n\mathbb{E}_{{P\mathstrut}_{\hat{\bf x}, \,{\bf y}}^{}}[\ln P(Y \,|\, \hat{X})]\big\}
\; \bigg| \;
({\bf X\mathstrut}_{m}, {\bf Y}) \, \in \, T({P\mathstrut}_{{\bf x}, \,{\bf y}}^{})\bigg\}
\nonumber \\
= \;\; & \Pr\,\bigg\{\sum_{m' \, \neq \, m}
\mathbbm{1}_{\displaystyle\big\{{\bf X\mathstrut}_{m'} \; \in \; T\big({P\mathstrut}_{\hat{\bf x} \, | \,{\bf x}, \, {\bf y}}^{}, \, {\bf X\mathstrut}_{m}, {\bf Y}\big)\big\}}(m')
\; \geq \; 1
\; \bigg| \;
({\bf X\mathstrut}_{m}, {\bf Y}) \, \in \, T({P\mathstrut}_{{\bf x}, \,{\bf y}}^{})\bigg\}
\nonumber
\end{align}
\begin{align}
\overset{(b)}{>} \;\; &
{e\mathstrut}^{-n(I\,-\,R)}\,\cdot\,
{\Big(1\,-\,{e\mathstrut}^{-nR}\Big)\mathstrut}^{{e\mathstrut}^{nR}}
\nonumber \\
\overset{(c)}{\geq} \;\; &
{e\mathstrut}^{-n\left(f({P\mathstrut}_{{\bf x}, \,{\bf y}, \, \hat{\bf x}}^{})\,-\,R \, + \, \epsilon_{1}\right)}\,\cdot\,
{\Big(1\,-\,{e\mathstrut}^{-nR}\Big)\mathstrut}^{{e\mathstrut}^{nR}},
\;\;\;\;\;\;\;\;\;
\;\;\;\;\;\;\;\;\;
\frac{|{\cal X}||{\cal X}||{\cal Y}|\ln(n + 1)}{n}
\; \leq \; \epsilon_{1},
\label{eqLemmaPartTwo}
\end{align}
where ($a$) follows by the definition (\ref{eqTypeExp}), ($b$) follows by (\ref{eqLowSlopePart}) of the lemma,
($c$) holds for sufficiently large $n$, as in (\ref{eqNLarge}), because for such $n$, according to
(\ref{eqIandF}), we obtain
$I \; \leq \; f\big({P\mathstrut}_{{\bf x}, \,{\bf y}, \, \hat{\bf x}}^{}\big)\,+\,\epsilon_{1}$.
If we choose
\begin{equation} \label{eqDeH}
h\big({P\mathstrut}_{{\bf x}, \,{\bf y}, \, \hat{\bf x}}^{}\big)
\;\; \triangleq \;\;
- \mathbb{E}_{{P\mathstrut}_{\hat{\bf x}, \,{\bf y}}^{}}[\ln P(Y \,|\, \hat{X})],
\end{equation}
then with (\ref{eqLemmaPartTwo}) the lower bound (\ref{eqBoundFH}) becomes
\begin{align}
& \max_{
\;\;\;{P\mathstrut}_{{\bf x}, \,{\bf y}, \, \hat{\bf x}}^{}}
\,\Pr\,\big\{({\bf X\mathstrut}_{m}, {\bf Y}) \, \in \, T({P\mathstrut}_{{\bf x}, \,{\bf y}}^{})\big\}\,\cdot\,
\Pr\,\Big\{
E\big({P\mathstrut}_{\hat{\bf x} \, | \,{\bf x}, \, {\bf y}}^{}, {\bf X\mathstrut}_{m}, {\bf Y}\big) \, \leq \, h\big({P\mathstrut}_{{\bf x}, \,{\bf y}, \, \hat{\bf x}}^{}\big)\; \Big| \; {P\mathstrut}_{{\bf x}, \,{\bf y}}^{}\Big\}\,\times
\nonumber \\
& \;\;\;\;\;\;\;\;\;\;\;\;\;\;\;\;\;\;\;\;\;\;\;\;\;\;\;\;\;\;\;\;\;\;\;\;\;\;\;\;\;\;\;\;\;\;\;\;
\;\;\;\;\;\;\;\;\;\;\;\;\;\;\;\;\;
\mathbbm{1}_{\bigg\{\substack{\displaystyle h\big({P\mathstrut}_{{\bf x}, \,{\bf y}, \, \hat{\bf x}}^{}\big)
\; < \; - \, \mathbb{E}_{{P\mathstrut}_{{\bf x}, \,{\bf y}}^{}}[\ln P(Y \,|\, X)] \, + \, D \\
\displaystyle f\big({P\mathstrut}_{{\bf x}, \,{\bf y}, \, \hat{\bf x}}^{}\big)
\; \geq \;  R
\;\;\;\;\;\;\;\;\;\;\;\;\;\;\;\;\;\;\;\;\;\;\;\;\;\;\;\;\;\;\;\;\;\;\;\;\;
}\bigg\}}
({P\mathstrut}_{{\bf x}, \,{\bf y}, \, \hat{\bf x}}^{})
\nonumber \\
&
\nonumber \\
& \overset{(a)}{\geq} \; \!\!
\max_{
\;\;\;{P\mathstrut}_{{\bf x}, \,{\bf y}, \, \hat{\bf x}}^{}}
{(n\, + \, 1)\mathstrut}^{-|{\cal X}|\cdot|{\cal Y}|}\cdot
\exp\Big\{-nD\big({P\mathstrut}_{{\bf x}, \,{\bf y}}^{}\;\big\|\;Q \circ P\big)\Big\}\,\cdot \,
\exp\Big\{-n\left(f\big({P\mathstrut}_{{\bf x}, \,{\bf y}, \, \hat{\bf x}}^{}\big)\,-\,R \, + \, \epsilon_{1}\right)\Big\}\,\times
\nonumber \\
&
\;\;\;\;\;\;\;\;\;\;\;\;\;\;\;\;\;\;\;\;\;\;\;\;\;\;\;\;\;\;\;\;\;\;\;\;\;\;\;\;\;\;\;\;\;\;\;\;\;\;\;
\;\;\;\;\;\;\;\;\;\;\;\;\;\;\;\;\;\;\;\;\;\;\;\;\;\;\;\;\;\;\;\;\;\;\;\;\;\;\;
\underbrace{
{\Big(1\,-\,{e\mathstrut}^{-nR}\Big)\mathstrut}^{{e\mathstrut}^{nR}}}_{\rightarrow \, 1/e}\times
\nonumber \\
& \;\;\;\;\;\;\;\;\;\;\;\;\;\;\;\;\;\;\;\;\;\;\;\;\;\;\;\;\;\;\;\;\;\;\;\;\;\;\;\;\;\;\;\;\;\;\;\;\;\;\;\;\;\;\;\;
\mathbbm{1}_{\bigg\{\substack{\displaystyle\mathbb{E}_{{P\mathstrut}_{{\bf x}, \,{\bf y}}^{}}[\ln P(Y \,|\, X)]
\, - \,\mathbb{E}_{{P\mathstrut}_{\hat{\bf x}, \,{\bf y}}^{}}[\ln P(Y \,|\, \hat{X})]
\; < \; D  \\
\;\;\;\;\;\;\;\;\;\;\;\;\;\;\;\;\;\;\;\;\;\;\;\;\;\;\;\;\;\;\;\;\;\;\;\;\;\;\;\;\;\;\;\;\;\;\;\;\;\;\;
\;\;\;\;\;\;\;\;\;\;\;\;\;
\displaystyle f\big({P\mathstrut}_{{\bf x}, \,{\bf y}, \, \hat{\bf x}}^{}\big)
\; \geq \; R}\bigg\}}
({P\mathstrut}_{{\bf x}, \,{\bf y}, \, \hat{\bf x}}^{})
\nonumber \\
&
\nonumber \\
& \overset{(b)}{\geq} \;
\exp\big\{-n\big[E_{2}^{\,\text{types}}(R, D) \, + \, \epsilon_{1} \, + \, \epsilon_{4}\big]\big\}
\nonumber \\
& \overset{(c)}{\geq} \; \exp\big\{-n\big[E_{2}(R \, + \, \epsilon_{5}, \, D \, - \, \epsilon_{5})\, + \, \epsilon_{1} \, + \, \epsilon_{4} \, + \, 2\epsilon_{5}\big]\big\}.
\label{eqTermTwo}
\end{align}
Explanation of steps:\newline
($a$) follows by the lower bound on the probability of the joint type (\ref{eqTypeProbBelow})
and (\ref{eqLemmaPartTwo}), (\ref{eqDeH});\newline
($b$) holds for $n$ sufficiently large for a given $\epsilon_{4}\, > \, 0$, and uses
the definition of $f\big({P\mathstrut}_{{\bf x}, \,{\bf y}, \, \hat{\bf x}}^{}\big)$ (\ref{eqDefF})
with the definition of the minimal exponent
\begin{align}
E_{2}^{\,\text{types}}(R, D) \;\; \triangleq \;\; & \min_{{P\mathstrut}_{{\bf x}, \,{\bf y}, \, \hat{\bf x}}^{}(x, \,y, \,\hat{x})} \, \Big\{D\big({P\mathstrut}_{{\bf x}, \,{\bf y}}^{} \;\big\|\; Q\circ P\big)\, + \,
D\big({P\mathstrut}_{{\bf x}, \,{\bf y}, \, \hat{\bf x}}^{} \,\big\|\, {P\mathstrut}_{{\bf x}, \,{\bf y}}^{} \!\times Q\big) \, - \, R\Big\}
\label{eqE2TypesEq} \\
& \;\;\;\;\;\;
\text{subject to:} \;\;\;\;\;\;\;\;
\;\;\;\;\;\;\;\;\;\;\;\;\;\;\;\;\;\;\;\;\;\;\;\,
\mathbb{E}_{{P\mathstrut}_{{\bf x}, \,{\bf y}, \, \hat{\bf x}}^{}}\bigg[\ln \frac{P(Y \,|\, X)}{P(Y \,|\, \hat{X})}\bigg]
\; < \; D,
\nonumber \\
& \;\;\;\;\;\;\;\;\;\;\;\;\;\;\;\;\;\;\;\;\;\;\;\;\;\;\;\;\;\;\;\;\;\;\;\;\;\;\;\;\;\;\;\;\;\;\;\;\;\;\;\;\;
D\big({P\mathstrut}_{{\bf x}, \,{\bf y}, \, \hat{\bf x}}^{} \,\big\|\, {P\mathstrut}_{{\bf x}, \,{\bf y}}^{} \!\times Q\big)
\; \geq \; R,
\nonumber
\end{align}
which differs from (\ref{eqE2Types}) by the inequality symbols ``$<$'' and ``$\geq$''.\newline
($c$) Let ${T\mathstrut}^{*}\circ {W\mathstrut}^{*}$ denote the joint distribution, achieving $E_{2}(R \, + \, \epsilon_{5}, \, D \, - \, \epsilon_{5})$,
defined by (\ref{eqThirdSumExponent}), for some $\epsilon_{5} \, > \, 0$.
This implies
\begin{align}
D\big({T\mathstrut}^{*} \;\|\; Q \circ P\big)
\, + \, D\big({T\mathstrut}^{*} \circ {W\mathstrut}^{*} \,\|\, {T\mathstrut}^{*} \times Q\big) \, - \, R \, - \, \epsilon_{5}
\;\; & = \;\; E_{2}(R \, + \, \epsilon_{5}, \, D \, - \, \epsilon_{5}),
\label{eqAchievesEps5} \\
\mathbb{E}_{\,{T\mathstrut}^{*}\circ\, {W\mathstrut}^{*}}\bigg[\ln \frac{P(Y \,|\, X)}{P(Y \,|\, \hat{X})}\bigg]
\;\; & \leq \;\; D \, - \, \epsilon_{5},
\nonumber \\
D\big({T\mathstrut}^{*} \circ {W\mathstrut}^{*} \,\|\, {T\mathstrut}^{*} \times Q\big)
\;\; & \geq \;\; R \, + \, \epsilon_{5}.
\nonumber
\end{align}
Let ${T\mathstrut}_{n}^{*}\circ {W\mathstrut}_{n}^{*}$
denote a quantized version of the joint distribution ${T\mathstrut}^{*} \circ {W\mathstrut}^{*}$
with precision $\frac{1}{n}$, i.e. a joint type with denominator $n$.
Since the divergences, as functions of $T\circ W$, have bounded derivatives, and also the ratio $\ln \frac{P(y \,|\, x)}{P(y \,|\, \hat{x})}$
is bounded, for any $\epsilon_{5}\,>\,0$ there exists $n$ large enough,
such that the quantized distribution ${T\mathstrut}_{n}^{*}\circ {W\mathstrut}_{n}^{*}$
satisfies
\begin{align}
D\big({T\mathstrut}_{n}^{*} \;\|\; Q \circ P\big)
\, + \, D\big({T\mathstrut}_{n}^{*} \circ {W\mathstrut}_{n}^{*} \,\|\, {T\mathstrut}_{n}^{*} \times Q\big) \, - \, \epsilon_{5}
\;\; & \leq \;\; D\big({T\mathstrut}^{*} \;\|\; Q \circ P\big)
\, + \, D\big({T\mathstrut}^{*} \circ {W\mathstrut}^{*} \,\|\, {T\mathstrut}^{*} \times Q\big),
\label{eqContDistEps5} \\
\mathbb{E}_{\,{T\mathstrut}_{n}^{*}\,\circ\, {W\mathstrut}_{n}^{*}}\bigg[\ln \frac{P(Y \,|\, X)}{P(Y \,|\, \hat{X})}\bigg]
\;\; & < \;\; D,
\nonumber \\
D\big({T\mathstrut}_{n}^{*} \circ {W\mathstrut}_{n}^{*} \,\|\, {T\mathstrut}_{n}^{*} \times Q\big)
\;\; & \geq \;\; R.
\nonumber
\end{align}
It follows from the last two inequalities that for $n$ sufficiently large
\begin{equation} \label{eqQWayR}
D\big({T\mathstrut}_{n}^{*} \;\|\; Q \circ P\big)
\, + \, D\big({T\mathstrut}_{n}^{*} \circ {W\mathstrut}_{n}^{*} \,\|\, {T\mathstrut}_{n}^{*} \times Q\big) \, - \, R
\;\; \geq \;\;
E_{2}^{\,\text{types}}(R, D),
\end{equation}
where $E_{2}^{\,\text{types}}(R, D)$ is defined as in (\ref{eqE2TypesEq}).
The relations (\ref{eqQWayR}), (\ref{eqContDistEps5}), (\ref{eqAchievesEps5}) give
\begin{displaymath}
E_{2}(R \, + \, \epsilon_{5}, \, D \, - \, \epsilon_{5}) \, + \, 2\epsilon_{5}
\;\; \geq \;\;
E_{2}^{\,\text{types}}(R, D).
\end{displaymath}
This explains ($c$).

The lower bounds on the probability (\ref{eqBoundFG}), (\ref{eqBoundFH}) are replaced now by (\ref{eqTermOne}), (\ref{eqTermTwo}),
resulting in the upper bound on the error exponent:
\begin{align}
\limsup_{n \, \rightarrow \, \infty} \; \left\{-\frac{1}{n}\ln \Pr \, \{{\cal E}_{m}\}\right\}
\;\; & \leq \;\;
\min \Big\{\,\,\widetilde{\!\!E}_{1}(R \, - \, 2\epsilon_{1} \, - \, \epsilon_{3}, \, D)\, + \, \epsilon_{2} \, + \, \epsilon_{3},
\nonumber \\
& \;\;\;\;\;\;\;\;\;\;\;\;\;\;\;\;\;\,
E_{2}(R \, + \, \epsilon_{5}, \, D \, - \, \epsilon_{5})\, + \, \epsilon_{1} \, + \, \epsilon_{4} \, + \, 2\epsilon_{5}\Big\}.
\nonumber
\end{align}
Since $\epsilon_{1}$, $\epsilon_{2}$, $\epsilon_{3}$, $\epsilon_{4}$, $\epsilon_{5}$ are arbitrary,
they can be replaced with zeros and limits, as follows
\begin{thm} \label{thm7}
\begin{align}
& \limsup_{n \, \rightarrow \, \infty} \; \left\{-\frac{1}{n}\ln \Pr \, \{{\cal E}_{m}\}\right\}
\;\; \leq \;\;
\lim_{\epsilon\,\rightarrow\,0}\;
\min \big\{ \,\,\widetilde{\!\!E}_{1}(R \, - \, \epsilon, \, D),\;
E_{2}(R \, + \, \epsilon, \, D \, - \,\epsilon)\big\}
\label{eqEpsilonLim} \\
& = \;\;
\min \left\{
\inf_{\substack{T(x,\,y),\,W(\hat{x} \,|\, x,\,y):\\
\\
d(T\,\circ\, W)\,+\,D(T\, \circ\, W \;\|\; T \,\times\, Q)\;<\; D \, + \, R\\
\\
D(T\, \circ\, W \;\|\; T \,\times\, Q)\;<\;R
}}
\Big\{D(T \; \| \; Q \circ P)\Big\}
\right.,
\nonumber \\
& \;\;\;\;\;\;\;\;\;\;\;\;\;\;\;\;\;\;\;\;\;\;\;\;\,
\left.
\inf_{\substack{T(x,\,y),\,W(\hat{x} \,|\, x,\,y):\\
\\
d(T\,\circ\, W) \; < \; D\\
\\
D(T\, \circ\, W \;\|\; T \,\times\, Q) \; > \; R
}}
\Big\{D(T \; \| \; Q \circ P)\, + \,
D(T \circ W \;\|\; T \times Q) \, - \, R\Big\}
\right\},
\label{eqUpperBoundForney}
\end{align}
{\em where
$\;d(T\,\circ\, W) \; = \; \mathbb{E}_{\,T\,\circ\, W}\big[d\big((X, Y), \hat{X}\big)\big]\; = \;\mathbb{E}_{\,T\,\circ\, W}\Big[\ln \frac{P(Y \,|\, X)}{P(Y \,|\, \hat{X})}\Big]$.}
\end{thm}
Note, that the only difference of the upper bound of Theorem~\ref{thm7} from the lower bound of Theorem~\ref{thm5} is that here the conditions are {\em strict} inequalities.

\section{\bf Derivation of the explicit random coding error exponent of Forney's decoder} \label{S16}
In order to compare the bounds (\ref{eqUpperBoundForney}) and (\ref{eqLowerBoundForney}),
it is convenient to rewrite (\ref{eqEpsilonLim}) and (\ref{eqLowerBoundForney}), and replace the {\em second argument}
of $\min$ with another expression. Observe that the {\em first argument} in the minimum of (\ref{eqEpsilonLim}) can be upper-bounded as
\begin{align}
\,\,\widetilde{\!\!E}_{1}(R \, - \, \epsilon, \, D) \;\; & = \;\;
\min_{\substack{T(x,\,y),\,W(\hat{x} \,|\, x,\,y):\\
\\
d(T\,\circ\, W)\,+\,D(T\, \circ\, W \;\|\; T \,\times\, Q)\;\leq\; D \, + \, R \, - \, \epsilon \\
\\
D(T\, \circ\, W \;\|\; T \,\times\, Q) \; \leq \; R \, - \, \epsilon
}}
\Big\{D(T \; \| \; Q \circ P)\Big\}
\nonumber \\
& \leq \;\;
\min_{\substack{T(x,\,y),\,W(\hat{x} \,|\, x,\,y):\\
\\
d(T\,\circ\, W)\,+\,D(T\, \circ\, W \;\|\; T \,\times\, Q)\;\leq\; D \, + \, R \, - \, \epsilon\\
\\
d(T\,\circ\, W)\;\leq\; D \, - \, \epsilon \\
\\
D(T\, \circ\, W \;\|\; T \,\times\, Q) \; \leq \; R \, - \, \epsilon
}}
\Big\{D(T \; \| \; Q \circ P)\Big\}
\nonumber \\
& = \;\;
\;\;\;\;\;\;\;\;\;
\min_{\substack{T(x,\,y),\,W(\hat{x} \,|\, x,\,y):\\
\\
d(T\,\circ\, W) \; \leq \; D \, - \, \epsilon \\
\\
D(T\, \circ\, W \;\|\; T \,\times\, Q) \; \leq \; R \, - \, \epsilon
}}
\;\;\;\;\;\;\;\;\;
\Big\{D(T \; \| \; Q \circ P)\Big\} \;\; \triangleq \;\; E_{1}(R \, - \, \epsilon, \, D \, - \, \epsilon).
\label{eqE1Defined}
\end{align}
Therefore, the minimum of (\ref{eqEpsilonLim}) becomes
\begin{align}
& \min\big\{\,\,\widetilde{\!\!E}_{1}(R \, - \, \epsilon, \, D), \; E_{2}(R \, + \, \epsilon, \, D \, - \, \epsilon)\big\}
\nonumber \\
= \;\; & \min\big\{\,\,\widetilde{\!\!E}_{1}(R \, - \, \epsilon, \, D), \; E_{1}(R \, - \, \epsilon, \, D \, - \, \epsilon), \; E_{2}(R \, + \, \epsilon, \, D \, - \, \epsilon)\big\}
\nonumber \\
= \;\; & \min\Big\{\,\,\widetilde{\!\!E}_{1}(R \, - \, \epsilon, \, D), \;
\underbrace{
\min\big\{E_{1}(R \, - \, \epsilon, \, D \, - \, \epsilon), \; E_{2}(R \, + \, \epsilon, \, D \, - \, \epsilon)\big\}
}
\Big\}.
\label{eqNewRightTerm}
\end{align}
Similarly (using $\epsilon\,=\,0$), for the minimum in (\ref{eqLowerBoundForney}) we obtain
\begin{align}
\min\big\{\,\,\widetilde{\!\!E}_{1}(R, D), \, E_{2}(R, D)\big\}
\;\; = \;\; & \min\Big\{\,\,\widetilde{\!\!E}_{1}(R, D), \, \underbrace{\min\big\{E_{1}(R, D), \, E_{2}(R, D)\big\}}_{\triangleq \, \,\,\widetilde{\!\!E\mathstrut}_{2}(R, D)}\Big\}
\nonumber \\
= \;\; & \min\big\{\,\,\widetilde{\!\!E}_{1}(R, D), \, \,\,\widetilde{\!\!E}_{2}(R, D)\big\}.
\label{eqE2TildeDefined}
\end{align}
The next lemma serves to clarify the relationship between the new right argument in (\ref{eqNewRightTerm}) and $\,\,\widetilde{\!\!E}_{2}(R, D)$
defined in (\ref{eqE2TildeDefined}), as follows:
\begin{lemma} \label{lemma11}
{\em If $D \, > \,  D_{\min} \, = \, \min_{(x, \, y), \, \hat{x}} d\big((x, y), \hat{x}\big)$, then for $R\,>\,0$}
\begin{displaymath}
\lim_{\epsilon\,\rightarrow\,0}\;
\min\big\{E_{1}(R \, - \, \epsilon, \, D), \; E_{2}(R \, + \, \epsilon, \, D)\big\}
\; = \; \,\,\widetilde{\!\!E}_{2}(R, D).
\end{displaymath}
\end{lemma}
\begin{proof}
Consider the definition of $\,\,\widetilde{\!\!E}_{2}(R, D) \, = \, \min\big\{E_{1}(R, D), \, E_{2}(R, D)\big\}\,$:
\begin{align}
&
\,\,\widetilde{\!\!E}_{2}(R, D) \;\; = \;\;
\min \left\{
\min_{\substack{T(x,\,y),\,W(\hat{x} \,|\, x,\,y):\\
\\
d(T\,\circ\, W) \; \leq \; D \\
\\
D(T\, \circ\, W \;\|\; T \,\times\, Q) \; \leq \; R
}}
\Big\{D(T \; \| \; Q \circ P)\Big\}
\right.,
\nonumber \\
& \;\;\;\;\;\;\;\;\;\;\;\;\;\;\;\;\;\;\;\;\;\;\;\;\;\;\;
\;\;\;\;\;
\left.
\min_{\substack{T(x,\,y),\,W(\hat{x} \,|\, x,\,y):\\
\\
d(T\,\circ\, W) \; \leq \; D\\
\\
D(T\, \circ\, W \;\|\; T \,\times\, Q) \; \geq \; R
}}
\Big\{D(T \; \| \; Q \circ P)\, + \,
D(T \circ W \;\|\; T \times Q) \, - \, R\Big\}
\right\}.
\nonumber
\end{align}
Let ${T\mathstrut}^{*}\circ {W\mathstrut}^{*}$ be the joint distribution,
achieving $\,\,\widetilde{\!\!E}_{2}(R, D)$.

If
$D\big({T\mathstrut}^{*} \circ {W\mathstrut}^{*} \;\|\; {T\mathstrut}^{*} \times Q\big) \, \neq \, R$,
then
for sufficiently small $\epsilon \, > \, 0$ the minimum is achieved
\begin{displaymath}
\min\big\{E_{1}(R \, - \, \epsilon, \, D), \; E_{2}(R \, + \, \epsilon, \, D)\big\}
\; = \;
\min\big\{E_{1}(R, D), \; E_{2}(R, D)\big\}
\; = \; \,\,\widetilde{\!\!E}_{2}(R, D),
\end{displaymath}
and the statement of the lemma holds.

If exactly $D\big({T\mathstrut}^{*} \circ {W\mathstrut}^{*} \;\|\; {T\mathstrut}^{*} \times Q\big) \, = \, R \, > \, 0$,
then, given the condition of the lemma $D \, > \,  D_{\min} \, = \, \min_{(x, \, y), \, \hat{x}} d\big((x, y), \hat{x}\big)$,
there exists $W' \, \neq \, {W\mathstrut}^{*}$, such that\footnote{Note, that for $\,D\, = \, D_{\min}\,$ a {\em distinct} $\;W' \, \neq \, {W\mathstrut}^{*}$, satisfying (\ref{eqD}), may not exist.}
\begin{equation} \label{eqD}
d\big({T\mathstrut}^{*}\,\circ\, W'\big) \; \leq \; D.
\end{equation}
Since $D(T \circ W \;\|\; T \times Q)$ is strictly convex in $W$,
there exists ${W}_{\lambda} \, = \, \lambda W' \, + \, (1 - \lambda){W\mathstrut}^{*}$
such that $D\big({T\mathstrut}^{*} \circ {W}_{\lambda} \;\|\; {T\mathstrut}^{*} \times Q\big)\, \neq \, R\;$
and is arbitrarily close to $R$. Note also that $d\big({T\mathstrut}^{*}\,\circ\, {W}_{\lambda}\big) \; \leq \; D$.
The convergence in the lemma follows.
\end{proof}

Since the limit of the right argument in (\ref{eqNewRightTerm}) can be taken in two steps as
\begin{align}
\lim_{\epsilon\,\rightarrow\,0}\;
& \min\big\{E_{1}(R \, - \, \epsilon, \, D \, - \, \epsilon), \; E_{2}(R \, + \, \epsilon, \, D \, - \, \epsilon)\big\}
\nonumber \\
\; = \;
\lim_{\epsilon_{2}\,\rightarrow\,0}\;
\lim_{\epsilon_{1}\,\rightarrow\,0}\;
& \min\big\{E_{1}(R \, - \, \epsilon_{1}, \, D \, - \, \epsilon_{2}), \; E_{2}(R \, + \, \epsilon_{1}, \, D \, - \, \epsilon_{2})\big\},
\nonumber
\end{align}
Lemma~\ref{lemma11} implies, that for any\footnote{For $D\,\leq\,D_{\min}$ both sides of (\ref{eqNewRightTermSimplified}) are trivially $+\infty$.} $D$ 
\begin{equation} \label{eqNewRightTermSimplified}
\lim_{\epsilon\,\rightarrow\,0}\;
\min\big\{E_{1}(R \, - \, \epsilon, \, D \, - \, \epsilon), \; E_{2}(R \, + \, \epsilon, \, D \, - \, \epsilon)\big\}
\; = \;
\lim_{\epsilon\,\rightarrow\,0}\;
\,\,\widetilde{\!\!E}_{2}(R, \, D \, - \, \epsilon).
\end{equation}
Thus, 
with the help of (\ref{eqNewRightTerm}) and (\ref{eqNewRightTermSimplified}),
the bound (\ref{eqEpsilonLim})
can be rewritten as
\begin{equation} \label{eqEpsilonLimSimplified}
\lim_{\epsilon\,\rightarrow\,0}\;
\min \big\{ \,\,\widetilde{\!\!E}_{1}(R \, - \, \epsilon, \, D),\;
E_{2}(R \, + \, \epsilon, \, D \, - \,\epsilon)\big\}
\; = \;
\lim_{\epsilon\,\rightarrow\,0}\;
\min \big\{ \,\,\widetilde{\!\!E}_{1}(R \, - \, \epsilon, \, D),\;
\,\,\widetilde{\!\!E}_{2}(R, \, D \, - \, \epsilon)\big\}.
\end{equation}

Thus far, we have obtained the two alternative expressions (\ref{eqE2TildeDefined}) and (\ref{eqEpsilonLimSimplified})
for the bounds (\ref{eqLowerBoundForney}) and (\ref{eqEpsilonLim}), respectively, replacing the second argument of $\min$
in each one of the bounds.

In what follows, first we obtain the explicit formula for the expression (\ref{eqE2TildeDefined}),
which is equivalent to the lower bound (\ref{eqLowerBoundForney}).
Then we conclude about the values of $(R, D)$ for which the bounds (\ref{eqE2TildeDefined}) and (\ref{eqEpsilonLimSimplified})
may not agree, which may occur only at the points of discontinuity of the bounds as functions of $(R, D)$.

It is convenient to express the new second argument $\,\,\widetilde{\!\!E}_{2}(R, D)$ with the help of $R(T, Q, D)$, defined in (\ref{eqRTQDDefinition}),
which is written here with substitutions (\ref{eqSubstitutions1}) and (\ref{eqSubstitutions3}):
\begin{equation} \label{eqRTQDChannel}
R(T, Q, D) \;\; \triangleq \;\; \min_{W(\hat{x} \,|\, x, \, y): \;\; d(T\,\circ\, W) \; \leq \; D} D(T \circ W \; \| \; T \times Q).
\end{equation}
Using $R(T, Q, D)$, we can write the following two identities:
\begin{align}
& \,\,\widetilde{\!\!E}_{2}(R, D) \;\; = \;\;
\min\big\{E_{1}(R, D),\, E_{2}(R, D)\big\} \;\; =
\nonumber \\
&
\min \left\{
\min_{\substack{T(x,\,y),\,W(\hat{x} \,|\, x,\,y):\\
\\
d(T\,\circ\, W) \; \leq \; D \\
\\
D(T\, \circ\, W \;\|\; T \,\times\, Q) \; \leq \; R
}}
\Big\{D(T \; \| \; Q \circ P)\Big\}
\right.,
\nonumber \\
& \;\;\;\;\;\;\;\;\;\,
\left.
\min_{\substack{T(x,\,y),\,W(\hat{x} \,|\, x,\,y):\\
\\
d(T\,\circ\, W) \; \leq \; D\\
\\
D(T\, \circ\, W \;\|\; T \,\times\, Q) \; \geq \; R
}}
\Big\{D(T \; \| \; Q \circ P)\, + \,
D(T \circ W \;\|\; T \times Q) \, - \, R\Big\}
\right\} \;\; =
\nonumber
\end{align}
\begin{align}
&
\min \left\{
\min_{\substack{T(x,\,y),\,W(\hat{x} \,|\, x,\,y):\\
\\
d(T\,\circ\, W) \; \leq \; D \\
\\
D(T\, \circ\, W \;\|\; T \,\times\, Q) \; \leq \; R
}}
\Big\{D(T \; \| \; Q \circ P)\Big\}, \;\;
\min_{T(x,\,y):\;
R(T, Q, D) \; \geq \; R}
\Big\{D(T \; \| \; Q \circ P)\, + \,
R(T, Q, D) \, - \, R\Big\}
\right\},
\label{eqE1TildeE2}
\end{align}
\begin{equation} \label{eqE1Tilde}
E_{1}(R, D) \;\; = \;\;
\min_{\substack{T(x,\,y),\,W(\hat{x} \,|\, x,\,y):\\
\\
d(T\,\circ\, W) \; \leq \; D \\
\\
D(T\, \circ\, W \;\|\; T \,\times\, Q) \; \leq \; R
}}
\Big\{D(T \; \| \; Q \circ P)\Big\}
\;\; = \;\;
\min_{T(x,\,y):\;
R(T, Q, D) \; \leq \; R}
\Big\{D(T \; \| \; Q \circ P)\Big\}.
\end{equation}
Combining the two, we have
\begin{align}
& \,\,\widetilde{\!\!E}_{2}(R, D) \;\; =
\nonumber \\
&
\min \left\{
\min_{T(x,\,y):\;
R(T, Q, D) \; \leq \; R}
\Big\{D(T \; \| \; Q \circ P)\Big\}, \;\;
\min_{T(x,\,y):\;
R(T, Q, D) \; \geq \; R}
\Big\{D(T \; \| \; Q \circ P)\, + \,
R(T, Q, D) \, - \, R\Big\}
\right\}
\nonumber \\
&
\;\;\;\;\;\;\;\;\;\;\;\;\;\;\;\;
= \;\;
\min_{T(x, \, y)} \; \Big\{ D(T \; \| \; Q \circ P) \;\; + \;\; {\big| R(T, Q, D) \; - \; R   \big|\mathstrut}_{}^{+} \Big\}.
\label{eqE2Tilde}
\end{align}
Observe, that this is exactly the same as (\ref{eqErrorD}), the channel decoding error exponent for the decoding error event defined in (\ref{eqErrorEvent}).
Consequently, (\ref{eqE2Tilde}) also equals (\ref{eqDecErrorArbD}). Recall, that (\ref{eqDecErrorArbD}) is equivalent to (\ref{eqErrorD}),
and derives from (\ref{eqErrorD}) by exactly the same derivation as (\ref{eqES}) from (\ref{eqSuccess}), with substitutions (\ref{eqSubstitutions1})-(\ref{eqSubstitutions3}).

The {\em first argument} $\,\,\widetilde{\!\!E}_{1}(R, D)$ in the minimum (\ref{eqE2TildeDefined}) can also be expressed alternatively,
with the help of a similar function, a ``coupled'' version of $R(T, Q, D)$, defined as
\begin{equation} \label{eqRTQR}
{R\mathstrut}^{c}(T, Q, D) \;\; \triangleq \;\;
\min_{W(\hat{x} \,|\, x, \, y): \;\; d(T\,\circ\, W) \, + \, D(T \,\circ\, W \; \| \; T\, \times\, Q) \; \leq \; D} D(T \circ W \; \| \; T \times Q).
\end{equation}
This definition gives
\begin{align}
\,\,\widetilde{\!\!E}_{1}(R, D) \;\; & = \;\;
\min_{\substack{T(x,\,y),\,W(\hat{x} \,|\, x,\,y):\\
\\
d(T\,\circ\, W)\,+\,D(T\, \circ\, W \;\|\; T \,\times\, Q)\;\leq\; D \, + \, R\\
\\
D(T\, \circ\, W \;\|\; T \,\times\, Q)\;\leq\;R
}}
\Big\{D(T \; \| \; Q \circ P)\Big\}
\;\; = \;\;
\min_{\substack{T(x,\,y):\\
{R\mathstrut}^{c}(T, \, Q, \, D \, + \, R) \; \leq \; R}}
\Big\{D(T \; \| \; Q \circ P)\Big\}.
\label{eqE1RD}
\end{align}
An explicit formula for ${R\mathstrut}^{c}(T, Q, D)$ is given by
\begin{lemma} \label{lemma12}
\begin{equation} \label{eqRTQRs}
{R\mathstrut}^{c}(T, Q, D) \; = \; \sup_{\mu\,\geq\,0} \; \Bigg\{-\sum_{x, \, y}T(x, y) \ln {\bigg[\sum_{\hat{x}}Q(\hat{x})e^{-\frac{\mu}{1\,+\,\mu}\big[d((x,\,y), \, \hat{x}) \, - \, D \big]}\bigg]\mathstrut}^{1\,+\,\mu}\Bigg\}.
\end{equation}
\end{lemma}
\begin{proof}
\begin{align}
& {R\mathstrut}^{c}(T, Q, D) \;\; \triangleq \;\; \min_{W(\hat{x} \,|\, x, \, y): \;\; d(T\,\circ\, W) \, + \, D(T \,\circ\, W \; \| \; T \,\times\, Q) \; \leq \; D} D(T \circ W \; \| \; T \times Q)
\nonumber
\end{align}
\begin{align}
& = \footnotemark \;\;
\min_{W(\hat{x} \,|\, x, \, y)} \;\sup_{\mu \,\geq \,0} \;\big\{ D(T \circ W \; \| \; T \times Q) \; + \; \mu\big[d(T\circ W) \, + \, D(T \circ W \; \| \; T \times Q) \, - \, D\big]\big\}
\nonumber \\
& \overset{(*)}{=} \;\;\;\;\, \sup_{\mu \,\geq \,0} \; \min_{W(\hat{x} \,|\, x, \, y)} \;\big\{ D(T \circ W \; \| \; T \times Q) \; + \; \mu\big[d(T\circ W) \, + \, D(T \circ W \; \| \; T \times Q) \, - \, D\big]\big\}
\label{eqRTQREnvelope} \\
& = \;\;\;\;\; \sup_{\mu \,\geq \,0} \; \min_{W(\hat{x} \,|\, x, \, y)} \;\Bigg\{ (1+\mu)\sum_{x,\,y, \, \hat{x}}T(x, y)W(\hat{x} \,|\, x, y)\ln \frac{W(\hat{x} \,|\, x, y)}{Q(\hat{x})}
\nonumber \\
&
\;\;\;\;\;\;\;\;\;\;\;\;\;\;\;\;\;\;\;\;\;\;\;\;\;\;\;\;\;\;\;\;\;\;\,
 + \; \mu\bigg[\sum_{x,\,y, \,\hat{x}}T(x, y)W(\hat{x} \,|\, x, y)d\big((x, y), \hat{x}\big) \; - \; D\bigg]\Bigg\}
\nonumber \\
& = \;\;\;\;\; \sup_{\mu \,\geq \,0} \;
\min_{W(\hat{x} \,|\, x, \, y)} \; \Bigg\{ (1 + \mu)\sum_{x,\,y, \,\hat{x}}T(x, y)W(\hat{x} \,|\, x, y)\ln \frac{W(\hat{x} \,|\, x, y)}{Q(\hat{x})e^{-\frac{\mu}{1\,+\,\mu}d((x,\,y), \, \hat{x})}} \; - \; \mu D \Bigg\}
\nonumber \\
& = \;\;\;\;\; \sup_{\mu\,\geq\,0} \; \Bigg\{-\sum_{x, \, y}T(x, y) \ln {\bigg[\sum_{\hat{x}}Q(\hat{x})e^{-\frac{\mu}{1\,+\,\mu}d((x,\,y), \, \hat{x})}\bigg]\mathstrut}^{1\,+\,\mu} \; - \; \mu D\Bigg\},
\label{eqRTQRExpression}
\end{align}
\footnotetext{This also includes the case when the set $\{W(\hat{x} \,|\, x, \,y): \; d(T\,\circ\, W) \, + \, D(T \,\circ\, W \; \| \; T \times Q) \; \leq \; D\}$
is empty, then ${R\mathstrut}^{c}(T, Q, D)\,=\,+\infty$.}
where the equality ($*$) holds because ${R\mathstrut}^{c}(T, Q, D)$ is a convex ($\cup$) function of $D$ (checked directly) and (\ref{eqRTQREnvelope})
is its lower convex envelope, therefore they must coincide.
Alternatively, ($*$) follows by the minimax theorem, since the objective function is convex ($\cup$) in $W(\hat{x} \,|\, x, y)$
and concave (linear) in $\mu$.
\end{proof}

The last expression (\ref{eqRTQRExpression}) helps to recognize the following property of ${R\mathstrut}^{c}(T, \, Q, \, D)$:
\begin{lemma} \label{lemma13}
{\em ${R\mathstrut}^{c}(T, Q, D)$ is a convex ($\cup$) function of the pair $\,(T, D)$.}
\end{lemma}
\begin{proof}
The same as of Lemma~\ref{lemma2}.
\end{proof}

This, in turn, results in convexity of $\,\,\widetilde{\!\!E}_{1}(R, D)$:
\begin{lemma} \label{lemma14}
{\em $\,\,\widetilde{\!\!E}_{1}(R, D)$ is a convex ($\cup$) function of $\,(R, D)$.}
\end{lemma}
\begin{proof}
Using (\ref{eqE1RD}):
\begin{align}
& \lambda \,\,\widetilde{\!\!E}_{1}({R\mathstrut}_{1}, {D\mathstrut}_{1}) \, + \, (1 - \lambda) \,\,\widetilde{\!\!E}_{1}({R\mathstrut}_{2}, {D\mathstrut}_{2})
\nonumber \\
& = \;\;
\lambda\;\cdot\min_{\substack{T(x,\,y):\\
\\
{R\mathstrut}^{c}(T, \, Q, \, {D\mathstrut}_{1} \, + \, {R\mathstrut}_{1}) \; \leq \; {R\mathstrut}_{1}}}
\Big\{D(T \; \| \; Q \circ P)\Big\}
\; + \;
(1 - \lambda)\;\cdot\min_{\substack{T(x,\,y):\\
\\
{R\mathstrut}^{c}(T, \, Q, \, {D\mathstrut}_{2} \, + \, {R\mathstrut}_{2}) \; \leq \; {R\mathstrut}_{2}}}
\Big\{D(T \; \| \; Q \circ P)\Big\}
\nonumber \\
& = \;\;
\lambda D\big({T\mathstrut}_{1}^{*} \; \| \; Q \circ P\big) \, + \, (1 - \lambda) D\big({T\mathstrut}_{2}^{*} \; \| \; Q \circ P\big)
\nonumber \\
& \geq \;\;
D\big(\lambda{T\mathstrut}_{1}^{*} \, + \, (1 - \lambda){T\mathstrut}_{2}^{*} \; \| \; Q \circ P\big)
\nonumber \\
& \geq \;\;
\min_{\substack{T(x,\,y):\\
{R\mathstrut}^{c}\big(T, \; Q, \; \lambda({D\mathstrut}_{1} + {R\mathstrut}_{1}) \, + \, (1 - \lambda)({D\mathstrut}_{2} \, + \, {R\mathstrut}_{2})\big) \;\; \leq \\
{R\mathstrut}^{c}\big(\lambda{T\mathstrut}_{1}^{*} \, + \, (1 - \lambda){T\mathstrut}_{2}^{*}, \; Q, \; \lambda({D\mathstrut}_{1} + {R\mathstrut}_{1}) \, + \, (1 - \lambda)({D\mathstrut}_{2} \, + \, {R\mathstrut}_{2})\big)}}
\Big\{D(T \; \| \; Q \circ P)\Big\}
\nonumber
\end{align}
\begin{align}
& \overset{(*)}{\geq} \;\;
\;\;\,
\min_{\substack{T(x,\,y):\\
{R\mathstrut}^{c}\big(T, \; Q, \; \lambda({D\mathstrut}_{1} + {R\mathstrut}_{1}) \, + \, (1 - \lambda)({D\mathstrut}_{2} \, + \, {R\mathstrut}_{2})\big) \;\; \leq \\
\\
\lambda{R\mathstrut}^{c}({T\mathstrut}_{1}^{*}, \; Q, \; {D\mathstrut}_{1} \, + \, {R\mathstrut}_{1})
\, + \,
(1 - \lambda){R\mathstrut}^{c}({T\mathstrut}_{2}^{*}, \; Q, \; {D\mathstrut}_{2} \, + \, {R\mathstrut}_{2})
}}
\;\;\,
\Big\{D(T \; \| \; Q \circ P)\Big\}
\nonumber \\
& \geq \;\;
\;\;\;\;\;\;
\min_{\substack{T(x,\,y):\\
{R\mathstrut}^{c}\big(T, \; Q, \; \lambda({D\mathstrut}_{1} + {R\mathstrut}_{1}) \, + \, (1 - \lambda)({D\mathstrut}_{2} \, + \, {R\mathstrut}_{2})\big) \;\; \leq \\
\\
\lambda{R\mathstrut}_{1}
\, + \,
(1 - \lambda){R\mathstrut}_{2}
}}
\;\;\;\;\;\;
\Big\{D(T \; \| \; Q \circ P)\Big\}
\nonumber \\
& = \;\;
\,\,\widetilde{\!\!E}_{1}\big(\lambda{R\mathstrut}_{1} \, + \, (1 - \lambda){R\mathstrut}_{2}, \, \lambda{D\mathstrut}_{1} \, + \, (1 - \lambda){D\mathstrut}_{2}\big),
\nonumber
\end{align}
where ($*$) follows by Lemma~\ref{lemma13}.
\end{proof}

Lemma~\ref{lemma14} helps to prove the following Lagrangian duality:
\begin{lemma} \label{lemma15}
\begin{align}
\,\,\widetilde{\!\!E}_{1}(R, D) \; & = \;
\sup_{\alpha\,\geq\,0,\;\beta\,\geq\,0}\;
\min_{T(x,\,y),\,W(\hat{x} \,|\, x,\,y)}
\Big\{D(T \; \| \; Q \circ P)
\nonumber \\
& \;\;\;\;\;\;\;\;\;\;\;\;\;\;\;\;\;\;\;\;\;\;\;\;\;\;\;\;\;\;\;\;\;\;\;\;\;\;\;\;\;\;\;\;
+\,\alpha\big[D(T\, \circ\, W \;\|\; T \,\times\, Q)\,+\, d(T\,\circ\, W)\, - \, D \, - \, R\big]
\nonumber \\
& \;\;\;\;\;\;\;\;\;\;\;\;\;\;\;\;\;\;\;\;\;\;\;\;\;\;\;\;\;\;\;\;\;\;\;\;\;\;\;\;\;\;\;\;
+\,\beta\big[D(T\, \circ\, W \;\|\; T \,\times\, Q) \, - \, R\big]
\Big\}.
\label{eqE1Envelope}
\end{align}
\end{lemma}
{\em Proof A:}
This variant of the proof equates a convex function with its lower convex envelope.
\begin{align}
\,\,\widetilde{\!\!E}_{1}(R, D) \; &
\;\;\;\;\;
= \;\;
\;\;\;\;
\min_{\substack{T(x,\,y),\,W(\hat{x} \,|\, x,\,y):\\
\\
d(T\,\circ\, W)\,+\,D(T\, \circ\, W \;\|\; T \,\times\, Q)\;\leq\; D \, + \, R\\
\\
D(T\, \circ\, W \;\|\; T \,\times\, Q)\;\leq\;R
}}
\Big\{D(T \; \| \; Q \circ P)\Big\}
\nonumber \\
& \overset{\alpha\,\geq\,0,\;\beta\,\geq\,0}{\geq} \;
\min_{\substack{T(x,\,y),\,W(\hat{x} \,|\, x,\,y):\\
\\
d(T\,\circ\, W)\,+\,D(T\, \circ\, W \;\|\; T \,\times\, Q)\;\leq\; D \, + \, R\\
\\
D(T\, \circ\, W \;\|\; T \,\times\, Q)\;\leq\;R
}}
\Big\{D(T \; \| \; Q \circ P)
\nonumber \\
& \;\;\;\;\;\;\;\;\;\;\;\;\;\;\;\;\;\;\;\;\;\;\;\;\;\;\;\;\;\;\;\;\;\;\;\;\;\;\;\;\;\;\;\;
+\,\alpha\big[D(T\, \circ\, W \;\|\; T \,\times\, Q)\,+\, d(T\,\circ\, W)\, - \, D \, - \, R\big]
\nonumber \\
& \;\;\;\;\;\;\;\;\;\;\;\;\;\;\;\;\;\;\;\;\;\;\;\;\;\;\;\;\;\;\;\;\;\;\;\;\;\;\;\;\;\;\;\;
+\,\beta\big[D(T\, \circ\, W \;\|\; T \,\times\, Q) \, - \, R\big]
\Big\}
\nonumber \\
& \overset{\alpha\,\geq\,0,\;\beta\,\geq\,0}{\geq} \;\;
\;\;\;\;\;\;\;\;\;\;\,
\min_{T(x,\,y),\,W(\hat{x} \,|\, x,\,y)}
\;\;\;\;\;\;\;\;\;\;\;\,
\Big\{D(T \; \| \; Q \circ P)
\nonumber \\
& \;\;\;\;\;\;\;\;\;\;\;\;\;\;\;\;\;\;\;\;\;\;\;\;\;\;\;\;\;\;\;\;\;\;\;\;\;\;\;\;\;\;\;\;
+\,\alpha\big[D(T\, \circ\, W \;\|\; T \,\times\, Q)\,+\,d(T\,\circ\, W) \, - \, D \, - \, R\big]
\nonumber \\
& \;\;\;\;\;\;\;\;\;\;\;\;\;\;\;\;\;\;\;\;\;\;\;\;\;\;\;\;\;\;\;\;\;\;\;\;\;\;\;\;\;\;\;\;
+\,\beta\big[D(T\, \circ\, W \;\|\; T \,\times\, Q) \, - \, R\big]
\Big\}.
\label{eqPlaneRD}
\end{align}
Let ${T\mathstrut}^{*} \circ {W\mathstrut}^{*}$ achieve the minimum in (\ref{eqPlaneRD}). Then there exist ${R\mathstrut}^{*} \geq\, 0$ and ${D\mathstrut}^{*}$, such that
\begin{align}
&
D\big({T\mathstrut}^{*}\, \circ\, {W\mathstrut}^{*} \;\|\; {T\mathstrut}^{*} \,\times\, Q\big)\,+\, d\big({T\mathstrut}^{*}\,\circ\, {W\mathstrut}^{*}\big)\, - \, {D\mathstrut}^{*} \, - \, {R\mathstrut}^{*} \;\; = \;\; 0,
\nonumber \\
&
D\big({T\mathstrut}^{*}\, \circ\, {W\mathstrut}^{*} \;\|\; {T\mathstrut}^{*} \,\times\, Q\big) \, - \, {R\mathstrut}^{*}  \;\; = \;\; 0
.
\nonumber
\end{align}
It follows, that the two-dimensional plane (\ref{eqPlaneRD}) touches $\,\,\widetilde{\!\!E}_{1}(R, D)$.
Since (\ref{eqPlaneRD}) is also a lower bound on $\,\,\widetilde{\!\!E}_{1}(R, D)$, we conclude, that (\ref{eqPlaneRD}) is a supporting plane of the surface $\,\,\widetilde{\!\!E}_{1}(R, D)$,
for each pair $\alpha\,\geq\,0$, $\beta\,\geq\,0$.

Now consider the other possible pairs $(\alpha, \beta)$, with negative $\alpha$
or $\beta$, or both. Observe by the definition,
that $\,\,\widetilde{\!\!E}_{1}(R, D)$
is a nonincreasing function of $R$ when the sum $D\,+\,R$ is kept constant.
Consequently, there does not exist a supporting plane for $\,\,\widetilde{\!\!E}_{1}(R, D)$
given by
\begin{equation} \label{eqSupportingPlane}
{E\mathstrut}_{0} \, - \, \alpha(D \,+\, R) \, - \, \beta R,
\end{equation}
with a negative $\beta$.
Similarly,
$\,\,\widetilde{\!\!E}_{1}(R, D)$ is a nonincreasing function of the sum $D\,+\,R$ when $R$ is kept constant.
Consequently, there does not exist a supporting plane for $\,\,\widetilde{\!\!E}_{1}(R, D)$
given by (\ref{eqSupportingPlane})
with a negative $\alpha$.

We conclude, that the supremum of the two-dimensional planes over $\alpha\,\geq\,0$, $\beta\,\geq\,0$
on the RHS of (\ref{eqE1Envelope}) is the lower convex envelope of $\,\,\widetilde{\!\!E}_{1}(R, D)$.
On the other hand,
the LHS of (\ref{eqE1Envelope}) is a convex ($\cup$) function of $(R, D)$ by Lemma~\ref{lemma14}.
Therefore they must coincide.
$\square$

{\em Proof B:}
This is a proof by repeated application of the minimax theorem for convex-concave functions.
\begin{align}
\,\,\widetilde{\!\!E}_{1}(R, D) \; & = \; \min_{\substack{T(x,\,y),\,W(\hat{x} \,|\, x,\,y):\\
\\
d(T\,\circ\, W)\,+\,D(T\, \circ\, W \;\|\; T \,\times\, Q)\;\leq\; D \, + \, R\\
\\
D(T\, \circ\, W \;\|\; T \,\times\, Q)\;\leq\;R
}}
\Big\{D(T \; \| \; Q \circ P)\Big\}
\nonumber \\
& = \,\;\;
\min_{T(x,\,y)}
\;
\min_{W(\hat{x} \,|\, x,\,y)}
\;
\sup_{\alpha\,\geq\,0,\;\beta\,\geq\,0}
\Big\{D(T \; \| \; Q \circ P)
\nonumber \\
& \;\;\;\;\;\;\;\;\;\;\;\;\;\;\;\;\;\;\;\;\;\;\;\;\;\;\;\;\;\;\;\;\;\;\;\;\;\;\;\;\;\;\;\;
+\,\alpha\big[D(T\, \circ\, W \;\|\; T \,\times\, Q)\,+\,d(T\,\circ\, W) \, - \, D \, - \, R\big]
\nonumber \\
& \;\;\;\;\;\;\;\;\;\;\;\;\;\;\;\;\;\;\;\;\;\;\;\;\;\;\;\;\;\;\;\;\;\;\;\;\;\;\;\;\;\;\;\;
+\,\beta\big[D(T\, \circ\, W \;\|\; T \,\times\, Q) \, - \, R\big]
\Big\}
\nonumber \\
& \overset{(a)}{=} \;\;
\min_{T(x,\,y)}
\;
\sup_{\alpha\,\geq\,0,\;\beta\,\geq\,0}
\;
\min_{W(\hat{x} \,|\, x,\,y)}
\Big\{D(T \; \| \; Q \circ P)
\nonumber \\
& \;\;\;\;\;\;\;\;\;\;\;\;\;\;\;\;\;\;\;\;\;\;\;\;\;\;\;\;\;\;\;\;\;\;\;\;\;\;\;\;\;\;\;\;
+\,\alpha\big[D(T\, \circ\, W \;\|\; T \,\times\, Q)\,+\,d(T\,\circ\, W) \, - \, D \, - \, R\big]
\nonumber \\
& \;\;\;\;\;\;\;\;\;\;\;\;\;\;\;\;\;\;\;\;\;\;\;\;\;\;\;\;\;\;\;\;\;\;\;\;\;\;\;\;\;\;\;\;
+\,\beta\big[D(T\, \circ\, W \;\|\; T \,\times\, Q) \, - \, R\big]
\Big\}
\nonumber \\
& \overset{(b)}{=} \;\;
\sup_{\alpha\,\geq\,0,\;\beta\,\geq\,0}
\;
\min_{T(x,\,y)}
\;
\min_{W(\hat{x} \,|\, x,\,y)}
\Big\{D(T \; \| \; Q \circ P)
\nonumber \\
& \;\;\;\;\;\;\;\;\;\;\;\;\;\;\;\;\;\;\;\;\;\;\;\;\;\;\;\;\;\;\;\;\;\;\;\;\;\;\;\;\;\;\;\;
+\,\alpha\big[D(T\, \circ\, W \;\|\; T \,\times\, Q)\,+\,d(T\,\circ\, W) \, - \, D \, - \, R\big]
\nonumber \\
& \;\;\;\;\;\;\;\;\;\;\;\;\;\;\;\;\;\;\;\;\;\;\;\;\;\;\;\;\;\;\;\;\;\;\;\;\;\;\;\;\;\;\;\;
+\,\beta\big[D(T\, \circ\, W \;\|\; T \,\times\, Q) \, - \, R\big]
\Big\},
\nonumber
\end{align}
where\newline
($a$) follows by the minimax theorem, because the objective function is convex ($\cup$) in $W(\hat{x} \,|\, x, y)$
and concave (linear) in $(\alpha, \beta)$;\newline
($b$) follows by the minimax theorem, because the corresponding objective function
\begin{align}
& \min_{W(\hat{x} \,|\, x,\,y)}
\Big\{D(T \; \| \; Q \circ P)
\, + \, \alpha\big[D(T\, \circ\, W \;\|\; T \,\times\, Q)\,+\,d(T\,\circ\, W) \, - \, D \, - \, R\big]
\nonumber \\
& \;\;\;\;\;\;\;\;\;\;\;\;\;\;\;\;\;\;\;\;\;\;\;\;\;\;\;\;\;\;\;\;\;\;\;\;\;
+\,\beta\big[D(T\, \circ\, W \;\|\; T \,\times\, Q) \, - \, R\big]
\Big\}
\nonumber \\
& = \;\;
D(T \; \| \; Q \circ P) \, - \, \alpha D \, - \,(\alpha+\beta)R
\, + \,\min_{W(\hat{x} \,|\, x,\,y)}\big[(\alpha+\beta)D(T\, \circ\, W \;\|\; T \,\times\, Q)\,+\,\alpha d(T\,\circ\, W)\big]
\nonumber \\
& = \;\;
\sum_{x,\,y}T(x, y)\ln \frac{T(x, y)}{Q(x)P(y\,|\,x)} \, - \, \alpha D \, - \,(\alpha+\beta)R
\nonumber \\
& \;\;\;
+\,\min_{W(\hat{x} \,|\, x,\,y)}\bigg[(\alpha+\beta) \sum_{x,\,y, \,\hat{x}}T(x, y)W(\hat{x} \,|\, x, \, y)\ln \frac{W(\hat{x} \,|\, x, \, y)}{Q(\hat{x})}\,+\,\alpha \sum_{x, \, y, \,\hat{x}}T(x, y)W(\hat{x} \,|\, x, y)d\big((x, y), \hat{x}\big)\bigg]
\nonumber \\
& 
=
\;\;
\sum_{x,\,y}T(x, y)\ln \frac{T(x, y)}{Q(x)P(y\,|\,x)} \, - \, \alpha D \, - \,(\alpha+\beta)R
\nonumber \\
& \;\;\;\;\;\;\;\;\;\;\;\;\;\;\;\;\;\;\;\;\;\;\;\;\;\;\;\;\;\;\;\;\;\;\;\;\;\;\;\;\;\;\;\;\;
+\,\min_{W(\hat{x} \,|\, x,\,y)}\bigg[(\alpha+\beta) \sum_{x,\,y, \,\hat{x}}T(x, y)W(\hat{x} \,|\, x, \, y)\ln \frac{W(\hat{x} \,|\, x, \, y)}{Q(\hat{x})
{e\mathstrut}^{-\frac{\alpha}{\alpha\,+\,\beta}d((x, \,y), \,\hat{x})}
}\bigg]
\nonumber \\
& = \;\;
\sum_{x,\,y}T(x, y)\ln \frac{T(x, y)}{Q(x)P(y\,|\,x)} \, - \, \alpha D \, - \,(\alpha+\beta)R
\,-\,(\alpha+\beta) \sum_{x,\,y}T(x, y)\ln \sum_{\hat{x}} Q(\hat{x})
{e\mathstrut}^{-\frac{\alpha}{\alpha\,+\,\beta}d((x, \,y), \,\hat{x})}
\label{eqMinW}
\end{align}
is\newline
1) concave ($\cap$) in $(\alpha, \beta)$ as a minimum of affine functions of $(\alpha, \beta)$,\newline
2) convex ($\cup$) in $T(x, y)$.
$\square$

Continuing (\ref{eqE1Envelope}) with (\ref{eqMinW}) gives
\begin{align}
& \,\,\widetilde{\!\!E}_{1}(R, D)
\nonumber \\
& = \;\;
\sup_{\alpha\,\geq\,0,\;\beta\,\geq\,0}
\;\min_{T(x,\,y)}
\left\{\sum_{x,\,y}T(x, y)\ln \frac{T(x, y)}{Q(x)P(y\,|\,x)\left[\sum_{\hat{x}} Q(\hat{x})
{e\mathstrut}^{-\frac{\alpha}{\alpha\,+\,\beta}d((x, \,y), \,\hat{x})}\right]^{\alpha\,+\beta}} \, - \, \alpha D \, - \,(\alpha+\beta)R \right\}
\nonumber \\
& = \;\;
\sup_{\alpha\,\geq\,0,\;\beta\,\geq\,0}
\left\{-\ln \sum_{x,\,y}Q(x)P(y\,|\,x)\bigg[\sum_{\hat{x}} Q(\hat{x})
{e\mathstrut}^{-\frac{\alpha}{\alpha\,+\,\beta}\left[d((x, \,y), \,\hat{x})\,-\,D\right]}\bigg]^{\alpha\,+\beta} - \;\; (\alpha+\beta)R \right\}
\nonumber \\
& = \;\;
\sup_{\rho\,\geq\,0}\;\sup_{0\,\leq\,s\,\leq\,1} \;
\left\{-\ln \sum_{x,\,y}Q(x)P(y\,|\,x)\bigg[\sum_{\hat{x}} Q(\hat{x})
{e\mathstrut}^{-s\left[d((x, \,y), \,\hat{x})\,-\,D\right]}\bigg]^{\rho} - \; \rho R \right\},
\label{eqE1RDExplicit}
\end{align}
where we define $\;\rho \, \triangleq \, \alpha\,+\,\beta\;$ and $\;s \, \triangleq \, \frac{\alpha}{\alpha\,+\,\beta}$.\footnote{Note also, that the parameter in the explicit formula for ${R\mathstrut}^{c}(T, Q, D)$ (\ref{eqRTQRs}) is related to $(\alpha, \beta)$ as $\mu \, = \, \frac{\alpha}{\beta}$.}

With (\ref{eqE1RDExplicit}) for $\,\,\widetilde{\!\!E}_{1}(R, D)$ and (\ref{eqDecErrorArbD}) for $\,\,\widetilde{\!\!E}_{2}(R, D)$, the lower bound (\ref{eqLowerBoundForney}) now acquires its final form:
\begin{align}
& \min\big\{\,\,\widetilde{\!\!E}_{1}(R, D), \,\,\,\widetilde{\!\!E}_{2}(R, D)\big\} \;\; =
\nonumber \\
& \min \left\{
\;\;\;
\sup_{\rho\,\geq\,0}
\;\;
\;\sup_{0\,\leq\,s\,\leq\,1} \;
\Bigg\{-\ln \sum_{x,\,y}Q(x)P(y\,|\,x)\left[\sum_{\hat{x}} Q(\hat{x})
{e\mathstrut}^{-s\left[d((x, \,y), \,\hat{x})\,-\,D\right]}\right]^{\rho} - \; \rho R \Bigg\}
\right.,
\nonumber \\
& \;\;\;\;\;\;\;\;\;\;
\left.\sup_{0\,\leq\,\rho\,\leq\,1}
\;
\;\;
\sup_{s\,\geq\,0}\;
\;\;
\Bigg\{-\ln
\sum_{x,\,y}Q(x)P(y\,|\,x)\left[\sum_{\hat{x}} Q(\hat{x})
{e\mathstrut}^{-s\left[d((x, \,y), \,\hat{x})\,-\,D\right]}\right]^{\rho} - \; \rho R \Bigg\}
\;\right\}
\label{eqE1E2TildeFinal} \\
& \triangleq \;\; {E\mathstrut}_{\text{tradeoff decoder}}(Q, R, D),
\nonumber
\end{align}
where $d((x, y), \hat{x}) \, = \, \ln \frac{P(y \,|\, x)}{P(y \,|\, \hat{x})}$.

As can be seen, both $\,\,\widetilde{\!\!E}_{1}(R, D)$ and $\,\,\widetilde{\!\!E}_{2}(R, D)$
are suprema of affine functions of $(R, D)$ and, as such, are convex ($\cup$) in $(R, D)$.
Therefore, both $\,\,\widetilde{\!\!E}_{1}(R, D)$ and $\,\,\widetilde{\!\!E}_{2}(R, D)$ are,
basically, continuous. With the exception of boundary points where they switch to $+\infty$.
In this respect,
as can be verified from the expressions above,
the functions are {\em lower semi-continuous}, i.e. the convex sets of $(R, D)$,
on which the functions are finite, are closed sets.

Specifically, the second argument $\,\,\widetilde{\!\!E}_{2}(R, D)$ becomes $+\infty$ for
\begin{equation} \label{eqE2Dmin}
D \; < \; D_{\min} \; = \; \min_{(x, \, y), \, \hat{x}} d\big((x, y), \hat{x}\big)
\; = \;
\min_{x,\, y, \, \hat{x}}\;\ln \frac{P(y \,|\, x)}{P(y \,|\, \hat{x})}.
\end{equation}
The first argument
$\,\,\widetilde{\!\!E}_{1}(R, D)$, as can be seen from (\ref{eqE1RD}), equals $+\infty$ for
\begin{equation} \label{eqE1Rmin}
R \; < \; \min_{T(x, \, y)}{R\mathstrut}^{c}(T, \,Q, \,D \,+\, R).
\end{equation}
Note, that $\displaystyle\min_{T(x, \, y)}{R\mathstrut}^{c}(T, \,Q, \,D \, + \, R)$ itself is a nonincreasing right-continuous function of $R$
(in fact it is convex ($\cup$) and therefore lower semi-continuous).
In particular, given a sufficiently small $D$, like $D\, < \, D_{\min}$, the function
$\displaystyle f(R) \, = \,\min_{T(x, \, y)}{R\mathstrut}^{c}(T, \,Q, \,D \, + \, R)\;$ equals $+\infty$ for small $R$, then jumps from $+\infty$ to a finite value and decreases to $0$, with increase of $R$.
In any case, we can define
\begin{align}
R_{\min}(D) \;\; & \triangleq \;\; \min_{\min_{\,T(x, \, y)}{R\mathstrut}^{c}(T, \,Q, \,D \, + \, R) \; \leq \; R}\;\{R\,\}
\label{eqRmin} \\
& \geq \;\; \min_{T(x, \, y)}{R\mathstrut}^{c}\big(T, \,Q, \,D \, + \, R_{\min}(D)\big).
\nonumber
\end{align}
Thus, $\,\,\widetilde{\!\!E}_{1}(R, D)$ becomes $+\infty$ for
\begin{equation} \label{eqE1RminD}
R \; < \; R_{\min}(D).
\end{equation}

We conclude, that the only possible points, where the expressions
(\ref{eqE2TildeDefined}) and (\ref{eqEpsilonLimSimplified})
may not be equal, are the points with $D \, = \, D_{\min}$, and the points $(R, D) \, = \, \big(R_{\min}(D), D\big)$.
\begin{thm} \label{thm8}
\begin{align}
& \lim_{n \, \rightarrow \, \infty} \; \left\{-\frac{1}{n}\ln \Pr \, \{{\cal E}_{m}\}\right\}
\;\; = \;\;
\min \big\{ \,\,\widetilde{\!\!E}_{1}(R, D),\;
\,\,\widetilde{\!\!E}_{2}(R, D)\big\}
\nonumber
\end{align}
{\em for all $(R, D)$, with the possible exception of some points $(R, D_{\min})$ and $\big(R_{\min}(D), D\big)$, where still}
\begin{align}
\liminf_{n \, \rightarrow \, \infty} \; \left\{-\frac{1}{n}\ln \Pr \, \{{\cal E}_{m}\}\right\}
\;\; & \geq \;\;
\min \big\{ \,\,\widetilde{\!\!E}_{1}(R, D),\;
\,\,\widetilde{\!\!E}_{2}(R, D)\big\}
\nonumber \\
\limsup_{n \, \rightarrow \, \infty} \; \left\{-\frac{1}{n}\ln \Pr \, \{{\cal E}_{m}\}\right\}
\;\; & \leq \;\;
\lim_{\epsilon\,\rightarrow\,0}\;
\min \big\{ \,\,\widetilde{\!\!E}_{1}(R \, - \, \epsilon, \, D),\;
\,\,\widetilde{\!\!E}_{2}(R, \, D \, - \, \epsilon)\big\},
\nonumber
\end{align}
{\em with $\,\,\widetilde{\!\!E}_{1}(R, D)$ and $\,\,\widetilde{\!\!E}_{2}(R, D)$ given explicitly by (\ref{eqE1E2TildeFinal}).}
\end{thm}

\section{{\bf Comparison of decoding error exponents for arbitrary} \texorpdfstring{$D$}{\em D}} \label{S17}
For convenience, let us define
\begin{displaymath}
{E\mathstrut}_{0}(s, \rho, Q, D)
\;\; \triangleq \;\;
 -\ln\sum_{x, \,y}Q(x)P(y \,|\, x)\left[\sum_{\hat{x}}Q(\hat{x})\left[\frac{P(y \,|\, x)}{P(y \,|\, \hat{x})}\,e^{-D}\right]^{-s}\right]^{\rho}.
\end{displaymath}

We start with the exponent (\ref{eqDecErrorArbD}), which is the highest, and corresponds to the ``source duality'' decoding error event defined in (\ref{eqErrorEvent}):
\begin{align}
{E\mathstrut}_{e}(Q, R, D) \;\;\; & = \;\;\;\;\;\;\;\;\;\;\;\;
\sup_{0 \, \leq \,\rho \,\leq \,1}
\;\;\;
\sup_{s\,\geq\,0} \;
\;\;
\big\{ {E\mathstrut}_{0}(s, \rho, Q, D) \; - \; \rho R \big\}
\nonumber \\
& = \;\;
\min \Big\{
\;\;\,
\sup_{\rho\,\geq\,0}
\;\;
\;\;\;
\sup_{s\,\geq\,0} \;
\;\;
\big\{{E\mathstrut}_{0}(s, \rho, Q, D) \; - \; \rho R \big\},
\nonumber \\
& \;\;\;\;\;\;\;\;\;\;\;\;\;\;\;\;
\sup_{0\,\leq\,\rho\,\leq\,1}
\;\;\;
\sup_{s\,\geq\,0}\;
\;\;
\big\{{E\mathstrut}_{0}(s, \rho, Q, D) \; - \; \rho R \big\}
\;\;\Big\}
\nonumber \\
& \geq \;\;
\min \Big\{
\;\;\,
\sup_{\rho\,\geq\,0}
\;\;\;
\sup_{0\,\leq\,s\,\leq\,1} \;
\big\{{E\mathstrut}_{0}(s, \rho, Q, D) \; - \; \rho R \big\},
\nonumber \\
& \;\;\;\;\;\;\;\;\;\;\;\;\;\;\;\;
\underbrace{
\sup_{0\,\leq\,\rho\,\leq\,1}
\;\;\;
\sup_{s\,\geq\,0}\;
\;\;
\big\{{E\mathstrut}_{0}(s, \rho, Q, D) \; - \; \rho R \big\}}_{{E\mathstrut}_{e}(Q, R, D)}
\;\;\Big\}
\;\; = \;\; {E\mathstrut}_{\text{tradeoff decoder}}(Q, R, D)
\nonumber \\
& \geq \;\;
\min \Big\{
\;\;\,
\sup_{\rho\,\geq\,0}
\;\;\;
\sup_{0\,\leq\,s\,\leq\,1} \;
\big\{{E\mathstrut}_{0}(s, \rho, Q, D) \; - \; \rho R \big\},
\nonumber \\
& \;\;\;\;\;\;\;\;\;\;\;\;\;\;\;\;
\sup_{0\,\leq\,\rho\,\leq\,1}\;
\sup_{0\,\leq\,s\,\leq\,1}\;
\big\{{E\mathstrut}_{0}(s, \rho, Q, D) \; - \; \rho R \big\}
\;\;\Big\}
\nonumber \\
& = \;\;\;\;\;\;\;\;\;\;\;\;
\sup_{0\,\leq\,\rho\,\leq\,1}\;
\sup_{0\,\leq\,s\,\leq\,1}\;
\big\{{E\mathstrut}_{0}(s, \rho, Q, D) \; - \; \rho R \big\}
\;\;\;\;\;\; = \;\; {E\mathstrut}_{\text{bound}}(Q, R, D).
\nonumber
\end{align}
Thus we obtain
\begin{displaymath}
{E\mathstrut}_{e}(Q, R, D) \;\; \geq \;\; {E\mathstrut}_{\text{tradeoff decoder}}(Q, R, D) \;\; \geq \;\; {E\mathstrut}_{\text{bound}}(Q, R, D),
\end{displaymath}
where both ${E\mathstrut}_{\text{tradeoff decoder}}(Q, R, D)$ and ${E\mathstrut}_{\text{bound}}(Q, R, D)$
denote lower bounds on the random coding error exponent of Forney's decoder
(\ref{eqForneyErrorEvent}).
${E\mathstrut}_{\text{tradeoff decoder}}(Q, R, D)$ is our tight bound given by Theorem~\ref{thm8},
and ${E\mathstrut}_{\text{bound}}(Q, R, D)$ appears in \cite[eq.~(24)]{Forney68}
(subject to additional maximization over $Q$).
\begin{lemma} \label{lemma16}
{\em For $D \geq 0$}
\begin{displaymath}
{E\mathstrut}_{e}(Q, R, D) \;\; = \;\; {E\mathstrut}_{\text{tradeoff decoder}}(Q, R, D) \;\; = \;\; {E\mathstrut}_{\text{bound}}(Q, R, D).
\end{displaymath}
\end{lemma}
\begin{proof}
\begin{align}
{E\mathstrut}_{0}\big(s\, = \,\tfrac{1}{1\,+\,\rho}, \,\rho, \,Q, \,D\big)
\;\; & = \;\;
-\ln\sum_{x, \,y}Q(x)P(y \,|\, x)\left[\sum_{\hat{x}}Q(\hat{x})\left[\frac{P(y \,|\, x)}{P(y \,|\, \hat{x})}\right]^{-\frac{1}{1\,+\,\rho}}\right]^{\rho}
\; - \; \tfrac{1}{1 \, + \, \rho} \cdot \rho D
\nonumber \\
& \overset{(*)}{\geq} \;\;
-\ln\sum_{x, \,y}Q(x)P(y \,|\, x)\left[\sum_{\hat{x}}Q(\hat{x})\left[\frac{P(y \,|\, x)}{P(y \,|\, \hat{x})}\right]^{-s}\right]^{\rho}
\; - \; s \cdot \rho D
\nonumber \\
& = \;\;
{E\mathstrut}_{0}(s, \rho, Q, D),
\nonumber
\end{align}
where ($*$) holds by (\ref{eqEE}) for $s \, \geq \, \tfrac{1}{1\,+\,\rho}$ and $D\,\geq\,0$.
We conclude, that
\begin{align}
& \sup_{0\,\leq\,\rho\,\leq\,1}
\;\;\;
\sup_{s\,\geq\,0}\;
\;\;
\big\{{E\mathstrut}_{0}(s, \rho, Q, D) \; - \; \rho R \big\}
\; = \;
\sup_{0\,\leq\,\rho\,\leq\,1}\;
\sup_{0\,\leq\,s\,\leq\,1}\;
\big\{{E\mathstrut}_{0}(s, \rho, Q, D) \; - \; \rho R \big\}.
\nonumber
\end{align}
\end{proof}

\section{{\bf Maximization over} \texorpdfstring{$\,Q\,$}{\em Q} {\bf of the random coding error exponent of Forney's decoder}} \label{S18}

When we try to maximize the random coding exponent, given by Theorem~\ref{thm8}, over $Q$,
straightforward maximization, at first glance, is hampered by the special points where the true exponent is unknown:
\begin{align}
&
\begin{array}{r l}
\big(R, \; D_{\min}(Q)\big),
& \displaystyle \;\;\;\;\;\;\;\;\;
D_{\min}(Q) \; = \; \min_{y}\;\;\;
\min_{x:\;\; Q(x)\, > \, 0}
\;\;\;
\min_{\hat{x}:\;\; Q(\hat{x})\, > \, 0}
\;\ln \frac{P(y \,|\, x)}{P(y \,|\, \hat{x})},
\\
\big(R_{\min}(Q, D), \; D\big),
& \displaystyle \;\;\;\;\;\;\;\;\;
R_{\min}(Q, D) \;\; = \;\; \min_{\min_{\,T(x, \, y)}{R\mathstrut}^{c}(T, \,Q, \,D \, + \, R) \; \leq \; R}\;\{R\,\}.
\end{array}
\nonumber
\end{align}
The special points of the first kind $\big(R, \, D_{\min}(Q)\big)$ can be avoided by simply maximizing for $D\,\neq\,D_{\min}(Q)$,
leaving the finite set of lines $\{D \, = \, D_{\min}(Q)\}$
(whose size is bounded by the number of all possible subsets of the channel input alphabet $\cal X$)
unaddressed.
The second kind of the special points $\big(R_{\min}(Q, D), \, D\big)$
cannot be avoided that simple, but, better still,
can be almost completely circumvented, as shown by the next lemmas.
\begin{lemma} \label{lemma17}
\begin{equation} \label{eqLemmaMinTRc}
\min_{T(x, \, y)} \; {R\mathstrut}^{c}(T, \,Q, \,D \,+\, R)
\;\; \leq \;\; f^{*}(R, D)
\;\; \triangleq \;\; \Bigg\{
\begin{array}{r l}
0, & \;\;\; R \; \geq \; -D, \\
+\infty, & \;\;\; R \; < \; -D.
\end{array}
\end{equation}
\end{lemma}
\begin{proof}
\begin{align}
& 
\min_{T(x, \, y)} \; {R\mathstrut}^{c}(T, \,Q, \,D \,+\, R)
\nonumber \\
\overset{(a)}{=} \;\;
& \min_{T(x, \, y)} \;
\sup_{\mu\,\geq\,0} \; \Bigg\{-\sum_{x, \, y}T(x, y) \ln {\bigg[\sum_{\hat{x}}Q(\hat{x})e^{-\frac{\mu}{1\,+\,\mu}\big[d((x,\,y), \, \hat{x}) \, - \, D \, - \, R \big]}\bigg]\mathstrut}^{1\,+\,\mu}\Bigg\}
\nonumber \\
\leq \;\;
& \min_{T(x, \, y)} \; \sup_{\mu\,\geq\,0} \;
\Bigg\{-\sum_{x, \, y}T(x, y) \ln \;
\min_{Q(\hat{x})} \;
{\bigg[\sum_{\hat{x}}Q(\hat{x})e^{-\frac{\mu}{1\,+\,\mu}\big[d((x,\,y), \, \hat{x}) \, - \, D \, - \, R \big]}\bigg]\mathstrut}^{1\,+\,\mu}\Bigg\}
\nonumber \\
= \;\;
& \min_{T(x, \, y)} \; \sup_{\mu\,\geq\,0} \;
\Bigg\{-\sum_{x, \, y}T(x, y) \ln \;
{\bigg[e^{-\frac{\mu}{1\,+\,\mu}\big[\max_{\hat{x}} d((x,\,y), \, \hat{x}) \, - \, D \, - \, R \big]}\bigg]\mathstrut}^{1\,+\,\mu}\Bigg\}
\nonumber \\
= \;\;
& \min_{T(x, \, y)} \; \sup_{\mu\,\geq\,0} \;
\Bigg\{\mu\sum_{x, \, y}T(x, y)
\Big[\max_{\hat{x}} \; d((x,\,y), \, \hat{x}) \, - \, D \, - \, R \Big]\Bigg\}
\nonumber \\
\overset{(b)}{=} \;\;
& \sup_{\mu\,\geq\,0} \; \min_{T(x, \, y)} \;
\Bigg\{\mu\sum_{x, \, y}T(x, y)
\Big[\max_{\hat{x}} \; d((x,\,y), \, \hat{x}) \, - \, D \, - \, R \Big]\Bigg\}
\nonumber \\
= \;\;
& \sup_{\mu\,\geq\,0} \;
\bigg\{\mu
\Big[\underbrace{\min_{x, \, y}\;\max_{\hat{x}}\; d((x,\,y), \, \hat{x})}_{=\,0} \, - \, D \, - \, R \Big]\bigg\}
\nonumber \\
\overset{(c)}{=} \;\;
& \sup_{\mu\,\geq\,0} \;
\big\{\mu\,
[ - \, D \, - \, R ]\big\}
\;\; = \;\; \Bigg\{
\begin{array}{r l}
0, & \;\;\; R \; \geq \; -D, \\
+\infty, & \;\;\; R \; < \; -D,
\end{array}
\nonumber
\end{align}
where ($a$) follows by Lemma~\ref{lemma12},
($b$) follows by the minimax theorem for the objective function
convex (linear) in $T(x, y)$ and concave (linear) in $\mu$,
and ($c$) follows by the property of $d\big((x, y), \hat{x}\big)$, similar to (\ref{eqDistorionZero}):
\begin{align}
\min_{x, \,y} \;\max_{\hat{x}}\; d\big((x, y), \hat{x}\big) \;\; & = \;\; \min_{x} \; \min_{y} \; \max_{\hat{x}}\; \ln \frac{P(y \,|\, x)}{P(y \,|\, \hat{x})}
\nonumber \\
& \geq \;\; \min_{x} \; \min_{y} \;\;\;\;\;\;\;\;\; \ln \frac{P(y \,|\, x)}{P(y \,|\, x)} \;\; = \;\; 0
\;\; = \;\; \;\;\;\;\;\;\;\, \min_{y} \; \max_{\hat{x}}\; \ln \frac{P(y \,|\, \hat{x})}{P(y \,|\, \hat{x})}
\nonumber \\
& \;\;\;\;\;\;\;\;\;\;\;\;\;\;\;\;\;\;\;\;\;\;\;\;\;\;\;\;\;\;\;\;\;\;\;\;\;\;\;\;\;\;\;\;\;\;\;\;\;\;\;\;\;\;\;\;\;\;
\geq \;\; \min_{x} \; \min_{y} \; \max_{\hat{x}}\; \ln \frac{P(y \,|\, x)}{P(y \,|\, \hat{x})},
\nonumber \\
\min_{x, \,y} \;\max_{\hat{x}}\; d\big((x, y), \hat{x}\big) \;\; & = \;\; 0.
\label{eqDistorionZeroRev}
\end{align}
\end{proof}
\begin{lemma} \label{lemma18}
\begin{equation} \label{eqMaxRmin}
\max_{Q(\hat{x})} \; R_{\min}(Q, D) \;\; = \;\; \max\,\{0, \,-D\}.
\end{equation}
\end{lemma}
\begin{proof}
\begin{align}
R_{\min}(Q, D) \;\; & = \;\; \min_{\min_{\,T(x, \, y)}{R\mathstrut}^{c}(T, \,Q, \,D \, + \, R) \; \leq \; R}\;\{R\,\} \;\;
\overset{(*)}{\leq} \;\; \min_{{f\mathstrut}^{*}(R, \,D) \; \leq \; R}\;\{R\,\}
\;\; = \;\; \max\,\{0, \,-D\},
\nonumber
\end{align}
where ($*$) follows by Lemma~\ref{lemma17}.
This upper bound is achieved by any degenerate distribution
\begin{equation} \label{eqQDegen}
Q(\hat{x}) \; = \;
\Bigg\{
\begin{array}{r l}
1, & \hat{x} \; = \; a, \\
0, & \hat{x} \; \neq \; a.
\end{array}
\end{equation}
Substitution of such $Q$ in the explicit formula (\ref{eqE1E2TildeFinal})
gives
\begin{equation} \label{eqE1TildeDeg}
\,\,\widetilde{\!\!E}_{1}(Q, R, D) \;\; = \;\;
\sup_{\rho\,\geq\,0}
\;\;
\;\sup_{0\,\leq\,s\,\leq\,1} \;\{-\rho s D \, - \, \rho R\}
\;\; = \;\;
\Bigg\{
\begin{array}{r l}
0, & \;\;\; R \; \geq \; \max\,\{0, \,-D\}, \\
+\infty, & \;\;\; \text{else}. 
\end{array}
\end{equation}
\end{proof}

The conclusion of Lemma~\ref{lemma18} is that $\,\,\widetilde{\!\!E}_{1}(Q, R, D)$ is finite, and hence continuous in $R$,
for $R \, > \, \max\,\{0, \,-D\}$.

Observe also, that substitution of the degenerate distribution (\ref{eqQDegen}) in (\ref{eqE1E2TildeFinal}) gives
\begin{equation} \label{eqE2TildeDeg}
\,\,\widetilde{\!\!E}_{2}(Q, R, D) \;\; = \;\;
\sup_{0\,\leq\,\rho\,\leq\,1}
\;\;
\sup_{s\,\geq\,0}\;\;
\{-\rho s D \, - \, \rho R\}
\;\; = \;\;
\Bigg\{
\begin{array}{r l}
0, & \;\;\; D \; \geq \; 0, \\
+\infty, & \;\;\; D \; < \; 0.
\end{array}
\end{equation}
It follows from (\ref{eqE1TildeDeg}) and (\ref{eqE2TildeDeg}),
that, in the case of negative $D$ and $R \, < \, -D$, the maximum of the random coding exponent over $Q$ is
\begin{equation} \label{eqRndExpInf}
\sup_{Q(x)} \; \min\big\{\,\,\widetilde{\!\!E}_{1}(Q, R, D), \,\,\,\widetilde{\!\!E}_{2}(Q, R, D)\big\} \;\; = \;\; +\infty,
\;\;\;\;\;\;\;\;\; 0 \, < \, R \, < \, -D.
\end{equation}
Therefore, we can formulate the following
\begin{thm} \label{thm9}
\begin{align}
& \sup_{Q(x)} \; \lim_{n \, \rightarrow \, \infty} \; \left\{-\frac{1}{n}\ln \Pr \, \{{\cal E}_{m}\}\right\}
\;\; = \;\;
\sup_{Q(x)} \; \min \big\{ \,\,\widetilde{\!\!E}_{1}(Q, R, D),\;
\,\,\widetilde{\!\!E}_{2}(Q, R, D)\big\},
\nonumber
\end{align}
{\em for all $(R, D)$, with the possible exception of
points with $R = -D$,
and points with
$0 >D \in{\{D_{\min}(Q)\}\mathstrut}_{Q}\,$
(for $R >  -D$),
where still}
\begin{align}
\sup_{Q(x)} \; \liminf_{n \, \rightarrow \, \infty} \; \left\{-\frac{1}{n}\ln \Pr \, \{{\cal E}_{m}\}\right\}
\;\; & \geq \;\;
\sup_{Q(x)} \; \min \big\{ \,\,\widetilde{\!\!E}_{1}(Q, R, D),\;
\,\,\widetilde{\!\!E}_{2}(Q, R, D)\big\}
\nonumber \\
\sup_{Q(x)} \; \limsup_{n \, \rightarrow \, \infty} \; \left\{-\frac{1}{n}\ln \Pr \, \{{\cal E}_{m}\}\right\}
\;\; & \leq \;\;
\sup_{Q(x)} \; \lim_{\epsilon\,\rightarrow\,0}\;
\min \big\{ \,\,\widetilde{\!\!E}_{1}(Q, \, R \, - \, \epsilon, \, D),\;
\,\,\widetilde{\!\!E}_{2}(Q, \, R, \, D \, - \, \epsilon)\big\},
\nonumber
\end{align}
{\em with $\,\,\widetilde{\!\!E}_{1}(Q, R, D)$ and $\,\,\widetilde{\!\!E}_{2}(Q, R, D)$ given explicitly by (\ref{eqE1E2TildeFinal}).}
\end{thm}
Now, using Lemma~\ref{lemma16} with Theorem~\ref{thm9}, we obtain, that the original Forney's random coding exponent is tight at least for $D \, \geq \, 0\,$:
\begin{cor} \label{cor1}
{\em For $D\, \geq \, 0$}
\begin{align}
& \sup_{Q(x)} \; \lim_{n \, \rightarrow \, \infty} \; \left\{-\frac{1}{n}\ln \Pr \, \{{\cal E}_{m}\}\right\}
\;\; = \;\;
\sup_{Q(x)} \; {E\mathstrut}_{\text{bound}}(Q, R, D).
\nonumber
\end{align}
\end{cor}
\section{\bf Derivation of the encoding success exponent} \label{S19}
This will lead to both (\ref{eqSuccess}) and (\ref{eqErrorD}).

Upper bound on the probability of successful encoding:
\begin{align}
P_{s}
\; & \leq \sum_{\;\;\;{P\mathstrut}_{{\bf x}, \, \hat{\bf x}}^{}}
\,\Pr\,\big\{{\bf X} \, \in \, T({P\mathstrut}_{\bf x}^{})\big\}
\,\cdot\,\Pr\,\Big\{\exists \, m: \; \hat{\bf X\mathstrut}_{m} \; \in \; T\big({P\mathstrut}_{\hat{\bf x} \, | \,{\bf x}}^{}, \, {\bf X}\big) \,\Big|\, {P\mathstrut}_{\bf x}^{}\Big\}\,\cdot\,
\mathbbm{1}_{\displaystyle\big\{d\big({P\mathstrut}_{{\bf x}, \, \hat{\bf x}}^{}\big)
\; \leq \; D\big\}}({P\mathstrut}_{{\bf x}, \, \hat{\bf x}}^{})
\nonumber \\
& \leq \sum_{\substack{{P\mathstrut}_{{\bf x}, \, \hat{\bf x}}^{}\,:\\
D({P\mathstrut}_{{\bf x}, \, \hat{\bf x}}^{} \,\|\, {P\mathstrut}_{\bf x}^{} \times \,Q)
\; \leq \; R}}
\,\Pr\,\big\{{\bf X} \, \in \, T({P\mathstrut}_{\bf x}^{})\big\}
\,\cdot\,
\mathbbm{1}_{\displaystyle\big\{d\big({P\mathstrut}_{{\bf x}, \, \hat{\bf x}}^{}\big)
\; \leq \; D\big\}}({P\mathstrut}_{{\bf x}, \, \hat{\bf x}}^{})
\nonumber \\
&
\, + \sum_{\substack{{P\mathstrut}_{{\bf x}, \, \hat{\bf x}}^{}\,:\\
D({P\mathstrut}_{{\bf x}, \, \hat{\bf x}}^{} \,\|\, {P\mathstrut}_{\bf x}^{} \times \,Q)
\; \geq \; R}}
\,\Pr\,\big\{{\bf X} \, \in \, T({P\mathstrut}_{\bf x}^{})\big\}
\,\cdot\,\Pr\,\Big\{\exists \, m: \; \hat{\bf X\mathstrut}_{m} \; \in \; T\big({P\mathstrut}_{\hat{\bf x} \, | \,{\bf x}}^{}, \, {\bf X}\big) \,\Big|\, {P\mathstrut}_{\bf x}^{}\Big\}\,\times
\nonumber \\
&
\;\;\;\;\;\;\;\;\;\;\;\;\;\;\;\;\;\;\;\;\;\;\;\;\;\;\;\;\;\;\;\;\;\;\;\;\;\;
\;\;\;\;\;\;\;\;\;\;\;\;\;\;\;\;\;\;\;\;\;\;\;\;\;\;\;\;\;\;\,
\mathbbm{1}_{\displaystyle\big\{d\big({P\mathstrut}_{{\bf x}, \, \hat{\bf x}}^{}\big)
\; \leq \; D\big\}}({P\mathstrut}_{{\bf x}, \, \hat{\bf x}}^{})
\nonumber \\
& \overset{(a)}{\leq} \sum_{\substack{{P\mathstrut}_{{\bf x}, \, \hat{\bf x}}^{}\,:\\
D({P\mathstrut}_{{\bf x}, \, \hat{\bf x}}^{} \,\|\, {P\mathstrut}_{\bf x}^{} \times \,Q)
\; \leq \; R}}
\exp\big\{-nD({P\mathstrut}_{\bf x}^{}\;\|\;P )\big\}
\,\cdot\,
\mathbbm{1}_{\displaystyle\big\{d\big({P\mathstrut}_{{\bf x}, \, \hat{\bf x}}^{}\big)
\; \leq \; D\big\}}({P\mathstrut}_{{\bf x}, \, \hat{\bf x}}^{})
\nonumber \\
&
\, + \sum_{\substack{{P\mathstrut}_{{\bf x}, \, \hat{\bf x}}^{}\,:\\
D({P\mathstrut}_{{\bf x}, \, \hat{\bf x}}^{} \,\|\, {P\mathstrut}_{\bf x}^{} \times \,Q)
\; \geq \; R}}
\exp\big\{-nD({P\mathstrut}_{\bf x}^{}\;\|\;P )\big\}
\,\cdot\,
\exp\Big\{-n\left[D\big({P\mathstrut}_{{\bf x}, \, \hat{\bf x}}^{} \,\big\|\, {P\mathstrut}_{\bf x}^{} \!\times Q\big) \, - \, R\right]\Big\}
\,\times
\nonumber \\
&
\;\;\;\;\;\;\;\;\;\;\;\;\;\;\;\;\;\;\;\;\;\;\;\;\;\;\;\;\;\;\;\;\;\;\;\;\;\;
\;\;\;\;\;\;\;\;\;\;\;\;\;\;\;\;\;\;\;\;\;\;\;\;\;\;\;\;\;\;\,
\mathbbm{1}_{\displaystyle\big\{d\big({P\mathstrut}_{{\bf x}, \, \hat{\bf x}}^{}\big)
\; \leq \; D\big\}}({P\mathstrut}_{{\bf x}, \, \hat{\bf x}}^{})
\nonumber \\
& = \sum_{\;\;\;{P\mathstrut}_{{\bf x}, \, \hat{\bf x}}^{}}
\exp\big\{-nD({P\mathstrut}_{\bf x}^{}\;\|\;P )\big\}
\,\cdot\,
\mathbbm{1}_{\bigg\{\substack{\displaystyle
\;\;\;\;\;\;\;\;\;\;\;\;\;\;\;\,
d\big({P\mathstrut}_{{\bf x}, \, \hat{\bf x}}^{}\big)
\; \leq \; D\\
\displaystyle D\big({P\mathstrut}_{{\bf x}, \, \hat{\bf x}}^{} \,\big\|\, {P\mathstrut}_{\bf x}^{} \!\times Q\big)
\; \leq \; R}
\bigg\}}({P\mathstrut}_{{\bf x}, \, \hat{\bf x}}^{})
\nonumber \\
&
\, + \sum_{\;\;\;{P\mathstrut}_{{\bf x}, \, \hat{\bf x}}^{}}
\exp\Big\{-n\left[D({P\mathstrut}_{\bf x}^{}\;\|\;P )
\, + \,
D\big({P\mathstrut}_{{\bf x}, \, \hat{\bf x}}^{} \,\big\|\, {P\mathstrut}_{\bf x}^{} \!\times Q\big) \, - \, R\right]\Big\}
\,\times
\nonumber \\
&
\;\;\;\;\;\;\;\;\;\;\;\;\;\;\;\;\;\;\;\;\;\;\;\;\;\;\;\;\;\;\;\;\;\;\;\;\;\;
\;\;\;\;\;\;\;\;\;\;\;\;\,
\mathbbm{1}_{\bigg\{\substack{\displaystyle
\;\;\;\;\;\;\;\;\;\;\;\;\;\;\;\,
d\big({P\mathstrut}_{{\bf x}, \, \hat{\bf x}}^{}\big)
\; \leq \; D\\
\displaystyle D\big({P\mathstrut}_{{\bf x}, \, \hat{\bf x}}^{} \,\big\|\, {P\mathstrut}_{\bf x}^{} \!\times Q\big)
\; \geq \; R}
\bigg\}}({P\mathstrut}_{{\bf x}, \, \hat{\bf x}}^{})
\nonumber \\
& \overset{(b)}{\leq} \sum_{\;\;\;{P\mathstrut}_{{\bf x}, \, \hat{\bf x}}^{}}
\exp\big\{-nE_{1}(R, D)\big\}
\, + \, \sum_{\;\;\;{P\mathstrut}_{{\bf x}, \, \hat{\bf x}}^{}}
\exp\big\{-nE_{2}(R, D)\big\}
\nonumber \\
& \leq \; 2 {(n\, + \, 1)\mathstrut}^{|{\cal X}|\cdot|\hat{\cal X}|}\,\cdot\,
\exp\big\{-n\min\big\{E_{1}(R, D), \, E_{2}(R, D)\big\}\big\},
\nonumber
\end{align}
where ($a$) uses the bound analogous to (\ref{eqTypeProb})
\begin{equation} \label{eqTypeProbSource}
\Pr\,\big\{{\bf X} \, \in \, T({P\mathstrut}_{\bf x}^{})\big\}
\; \leq \;
\exp\big\{-nD\big({P\mathstrut}_{\bf x}^{}(x)\;\big\|\;P(x)\big)\big\},
\end{equation}
and the union bound analogous to (\ref{eqUBTypesR})
\begin{displaymath}
\Pr\,\Big\{\exists \, m: \; \hat{\bf X\mathstrut}_{m} \; \in \; T\big({P\mathstrut}_{\hat{\bf x} \, | \,{\bf x}}^{}, \, {\bf X}\big) \,\Big|\, {P\mathstrut}_{\bf x}^{}\Big\}
\; \leq \;
\exp\Big\{-n\left[D\big({P\mathstrut}_{{\bf x}, \, \hat{\bf x}}^{}(x, \hat{x}) \,\big\|\, {P\mathstrut}_{\bf x}^{}(x) \cdot Q(\hat{x})\big) \, - \, R\right]\Big\};
\end{displaymath}
($b$) uses the definitions
\begin{align}
&
E_{1}(R, D)
\;\; \triangleq \;\;
\min_{\substack{T(x),\,W(\hat{x} \,|\, x):\\
\\
d(T\,\circ\, W) \; \leq \; D \\
\\
D(T\, \circ\, W \;\|\; T \,\times\, Q) \; \leq \; R
}}
\big\{D(T \; \| \; P)\big\},
\label{eqE1Sources} \\
&
E_{2}(R, D)
\;\; \triangleq \;\;
\min_{\substack{T(x),\,W(\hat{x} \,|\, x):\\
\\
d(T\,\circ\, W) \; \leq \; D\\
\\
D(T\, \circ\, W \;\|\; T \,\times\, Q) \; \geq \; R
}}
\big\{D(T \; \| \; P)\, + \,
D(T \circ W \;\|\; T \times Q) \, - \, R\big\},
\label{eqE2Sources}
\end{align}
analogous to (\ref{eqE1Defined}) and (\ref{eqThirdSumExponent}).
Thus, we obtain the lower bound on the encoding success exponent:
\begin{thm} \label{thm10}
\begin{align}
& \liminf_{n \, \rightarrow \, \infty} \; \left\{-\frac{1}{n}\ln P_{s}\right\}
\;\; \geq \;\;
\min \big\{ E_{1}(R, D),\;
E_{2}(R, D)\big\}.
\label{eqLowerSuccess}
\end{align}
\end{thm}

Next, we construct two alternative lower bounds on the probability of successful encoding:
\begin{align}
P_{s} \;\; & \geq \;\;
\max_{\;\;\;{P\mathstrut}_{{\bf x}, \, \hat{\bf x}}^{}}
\,\Pr\,\big\{{\bf X} \, \in \, T({P\mathstrut}_{\bf x}^{})\big\}
\,\cdot\,\Pr\,\Big\{\exists \, m: \; \hat{\bf X\mathstrut}_{m} \; \in \; T\big({P\mathstrut}_{\hat{\bf x} \, | \,{\bf x}}^{}, \, {\bf X}\big) \,\Big|\, {P\mathstrut}_{\bf x}^{}\Big\}\,\times
\nonumber \\
& \;\;\;\;\;\;\;\;\;\;\;\;\;\;\;\;\;\;\;\;\;\;\;\;\;\;\;\;\;\;\;\;\;\;\;\;\;
\;\;\;\;\;\;\;\;\;\;\;\;\;\;\;\;\;\;\;\;
\mathbbm{1}_{\bigg\{\substack{\displaystyle
\;\;\;\;\;\;\;
d\big({P\mathstrut}_{{\bf x}, \, \hat{\bf x}}^{}\big)
\; \leq \; D \\
\displaystyle D\big({P\mathstrut}_{{\bf x}, \, \hat{\bf x}}^{} \,\big\|\, {P\mathstrut}_{\bf x}^{} \!\times Q\big)
\; \leq \; R \, - \, \epsilon_{1}}
\bigg\}}({P\mathstrut}_{{\bf x}, \, \hat{\bf x}}^{})
\nonumber \\
& = \;\;
\max_{\;\;\;{P\mathstrut}_{{\bf x}, \, \hat{\bf x}}^{}}
\,\Pr\,\big\{{\bf X} \, \in \, T({P\mathstrut}_{\bf x}^{})\big\}
\,\cdot\,\Pr\,\bigg\{
\sum_{m}
\mathbbm{1}_{\displaystyle\big\{\hat{\bf X\mathstrut}_{m} \; \in \; T\big({P\mathstrut}_{\hat{\bf x} \, | \,{\bf x}}^{}, \, {\bf X}\big)\big\}}(m)
\; \geq \; 1
\,\bigg|\, {P\mathstrut}_{\bf x}^{}\bigg\}\,\times
\nonumber \\
& \;\;\;\;\;\;\;\;\;\;\;\;\;\;\;\;\;\;\;\;\;\;\;\;\;\;\;\;\;\;\;\;\;\;\;\;\;
\;\;\;\;\;\;\;\;\;\;\;\;\;\;\;\;\;\;\;\;
\mathbbm{1}_{\bigg\{\substack{\displaystyle
\;\;\;\;\;\;\;
d\big({P\mathstrut}_{{\bf x}, \, \hat{\bf x}}^{}\big)
\; \leq \; D \\
\displaystyle D\big({P\mathstrut}_{{\bf x}, \, \hat{\bf x}}^{} \,\big\|\, {P\mathstrut}_{\bf x}^{} \!\times Q\big)
\; \leq \; R \, - \, \epsilon_{1}}
\bigg\}}({P\mathstrut}_{{\bf x}, \, \hat{\bf x}}^{})
\nonumber \\
& \overset{(a)}{\geq} \;\;
\max_{\;\;\;{P\mathstrut}_{{\bf x}, \, \hat{\bf x}}^{}}
\,\Pr\,\big\{{\bf X} \, \in \, T({P\mathstrut}_{\bf x}^{})\big\}
\,\cdot\,\Pr\,\Bigg\{\underbrace{\sum_{i \, = \, 1}^{{e\mathstrut}^{nR}}Z_{i} \; \geq \; 1}_{Z_{i}\;\sim\;\text{Ber}({e\mathstrut}^{-nR})}
\Bigg\}
\,\cdot \, 
\mathbbm{1}_{\bigg\{\substack{\displaystyle
\;\;\;\;\;\;\;
d\big({P\mathstrut}_{{\bf x}, \, \hat{\bf x}}^{}\big)
\; \leq \; D \\
\displaystyle D\big({P\mathstrut}_{{\bf x}, \, \hat{\bf x}}^{} \,\big\|\, {P\mathstrut}_{\bf x}^{} \!\times Q\big)
\; \leq \; R \, - \, \epsilon_{1}}
\bigg\}}({P\mathstrut}_{{\bf x}, \, \hat{\bf x}}^{})
\nonumber \\
& \overset{(b)}{\geq} \;\;
\max_{\;\;\;{P\mathstrut}_{{\bf x}, \, \hat{\bf x}}^{}}\;
\underbrace{{(n\, + \, 1)\mathstrut}^{-|{\cal X}|}\cdot
\exp\big\{-nD({P\mathstrut}_{\bf x}^{}\;\|\;P )\big\}}_{\leq\;\Pr\,\{{\bf X} \, \in \, T({P\mathstrut}_{\bf x}^{})\}}
\,\cdot\,
\underbrace{
{\Big(1\,-\,{e\mathstrut}^{-nR}\Big)\mathstrut}^{{e\mathstrut}^{nR}}}_{ \rightarrow \, 1/e}
\,\times
\nonumber \\
& \;\;\;\;\;\;\;\;\;\;\;\;\;\;\;\;\;\;\;\;\;\;\;\;\;\;\;\;\;\;\;\;\;\;\;\;\;
\;\;\;\;\;\;\;\;\;\;\;\;\;\;\;\;\;\;\;\;
\;\;\;\;\;\;\;\;\;\;\;\;\;\;\;\;\;\;\;\;\;\;\;\;\,
\mathbbm{1}_{\bigg\{\substack{\displaystyle
\;\;\;\;\;\;\;
d\big({P\mathstrut}_{{\bf x}, \, \hat{\bf x}}^{}\big)
\; \leq \; D \\
\displaystyle D\big({P\mathstrut}_{{\bf x}, \, \hat{\bf x}}^{} \,\big\|\, {P\mathstrut}_{\bf x}^{} \!\times Q\big)
\; \leq \; R \, - \, \epsilon_{1}}
\bigg\}}({P\mathstrut}_{{\bf x}, \, \hat{\bf x}}^{})
\nonumber \\
& \overset{(c)}{\geq} \;\;
\exp\big\{-n\big[E_{1}^{\,\text{types}}(R \, - \, \epsilon_{1}, \, D) \, + \, \epsilon_{2}\big]\big\}
\nonumber \\
& \overset{(d)}{\geq} \;\; \exp\big\{-n\big[E_{1}(R \, - \, \epsilon_{1} \, - \, \epsilon_{3}, \, D \, - \, \epsilon_{3})\, + \, \epsilon_{2} \, + \, \epsilon_{3}\big]\big\}.
\label{eqLowBoA}
\end{align}
Explanation of steps:\newline
($a$) holds for sufficiently large $n$, when
\begin{displaymath} 
\Pr\,\big\{\hat{\bf X\mathstrut}_{m} \; \in \; T\big({P\mathstrut}_{\hat{\bf x} \, | \,{\bf x}}^{}, \, {\bf X}\big)
\;\big|\;
{\bf X} \, \in \, T({P\mathstrut}_{\bf x}^{})\big\}
\;\geq\;
\exp\Big\{-n\left[D\big({P\mathstrut}_{{\bf x}, \, \hat{\bf x}}^{} \,\big\|\, {P\mathstrut}_{\bf x}^{} \!\times Q\big) \, + \, \epsilon_{1}\right]\Big\}
\;\geq\;
\exp\{-n R\},
\end{displaymath}
for
\begin{displaymath}
Z_{i} \;\; \sim \;\; \text{i.i.d}\;\;\text{Bernoulli}\left(\exp\{-n R\}\right).
\end{displaymath}
($b$) uses a lower bound on the probability of a type, and the second part of Lemma~\ref{lemma10} with $I \,=\, R$.\newline
($c$) holds for sufficiently large $n$, given $\epsilon_{2}\,>\,0$, with the exponent $E_{1}^{\,\text{types}}(R \, - \, \epsilon_{1}, \, D)$
defined as in (\ref{eqE1Sources}) with types in place of $T\circ W$.\newline
($d$) Analogous to the steps in (\ref{eqTermOne}). Let ${T\mathstrut}^{*}\circ {W\mathstrut}^{*}$ denote the joint distribution, achieving
\begin{displaymath}
E_{1}(R \, - \, \epsilon_{1} \, - \, \epsilon_{3}, \, D \, - \, \epsilon_{3}),
\end{displaymath}
defined by (\ref{eqE1Sources}), for some $\epsilon_{3} \, > \, 0$.
This implies
\begin{align}
D\big({T\mathstrut}^{*} \;\|\; P\big) \;\; & = \;\; E_{1}(R \, - \, \epsilon_{1} \, - \, \epsilon_{3}, \, D \, - \, \epsilon_{3}),
\label{eqAch} \\
d\big({T\mathstrut}^{*} \circ {W\mathstrut}^{*}\big)
\;\; & \leq \;\; D \, - \, \epsilon_{3},
\nonumber \\
D\big({T\mathstrut}^{*} \circ {W\mathstrut}^{*} \,\|\, {T\mathstrut}^{*} \times Q\big)
\;\; & \leq \;\; R \, - \, \epsilon_{1} \, - \, \epsilon_{3}.
\nonumber
\end{align}
Let ${T\mathstrut}_{n}^{*}\circ {W\mathstrut}_{n}^{*}$
denote a quantized version of the joint distribution ${T\mathstrut}^{*} \circ {W\mathstrut}^{*}$
with precision $\frac{1}{n}$, i.e. a joint type with denominator $n$.
Note, that the divergences, as functions of $T\circ W$, have bounded derivatives, and also the distortion measure $d(x, \hat{x})$
is bounded. Therefore, for any $\epsilon_{3}\,>\,0$ there exists $n$ large enough,
such that the quantized distribution ${T\mathstrut}_{n}^{*}\circ {W\mathstrut}_{n}^{*}$
satisfies
\begin{align}
D\big({T\mathstrut}_{n}^{*} \;\|\; P\big) \;\; & \leq \;\; D({T\mathstrut}^{*} \;\|\; P) \, + \, \epsilon_{3},
\label{eqOptDist} \\
d\big({T\mathstrut}_{n}^{*} \circ {W\mathstrut}_{n}^{*}\big)
\;\; & \leq \;\; D,
\nonumber \\
D\big({T\mathstrut}_{n}^{*} \circ {W\mathstrut}_{n}^{*} \,\|\, {T\mathstrut}_{n}^{*} \times Q\big)
\;\; & \leq \;\; R \, - \, \epsilon_{1}.
\nonumber
\end{align}
It follows from the last two inequalities that for $n$ sufficiently large
\begin{equation} \label{eqQW}
D\big({T\mathstrut}_{n}^{*} \;\|\; P\big) \;\; \geq \;\; E_{1}^{\,\text{types}}(R \, - \, \epsilon_{1}, \, D).
\end{equation}
The relations (\ref{eqQW}), (\ref{eqOptDist}), (\ref{eqAch}) result in
\begin{displaymath}
E_{1}(R \, - \, \epsilon_{1} \, - \, \epsilon_{3}, \, D \, - \, \epsilon_{3}) \, + \, \epsilon_{3}
\;\; \geq \;\; E_{1}^{\,\text{types}}(R \, - \, \epsilon_{1}, \, D).
\end{displaymath}
This explains ($d$).

The second bound:
\begin{align}
P_{s} \;\; & \geq \;\;
\max_{\;\;\;{P\mathstrut}_{{\bf x}, \, \hat{\bf x}}^{}}
\,\Pr\,\big\{{\bf X} \, \in \, T({P\mathstrut}_{\bf x}^{})\big\}
\,\cdot\,\Pr\,\Big\{\exists \, m: \; \hat{\bf X\mathstrut}_{m} \; \in \; T\big({P\mathstrut}_{\hat{\bf x} \, | \,{\bf x}}^{}, \, {\bf X}\big) \,\Big|\, {P\mathstrut}_{\bf x}^{}\Big\}\,\times
\nonumber \\
& \;\;\;\;\;\;\;\;\;\;\;\;\;\;\;\;\;\;\;\;\;\;\;\;\;\;\;\;\;\;\;\;\;\;\;\;\;
\;\;\;\;\;\;\;\;\;\;\;\;\;\;\;\;\;\;\;\;
\mathbbm{1}_{\bigg\{\substack{\displaystyle
\;\;\;\;\;\;\;\;\;\;\;\;\;\;\;\,
d\big({P\mathstrut}_{{\bf x}, \, \hat{\bf x}}^{}\big)
\; \leq \; D \\
\displaystyle D\big({P\mathstrut}_{{\bf x}, \, \hat{\bf x}}^{} \,\big\|\, {P\mathstrut}_{\bf x}^{} \!\times Q\big)
\; \geq \; R}
\bigg\}}({P\mathstrut}_{{\bf x}, \, \hat{\bf x}}^{})
\nonumber \\
& = \;\;
\max_{\;\;\;{P\mathstrut}_{{\bf x}, \, \hat{\bf x}}^{}}
\,\Pr\,\big\{{\bf X} \, \in \, T({P\mathstrut}_{\bf x}^{})\big\}
\,\cdot\,\Pr\,\bigg\{
\sum_{m}
\mathbbm{1}_{\displaystyle\big\{\hat{\bf X\mathstrut}_{m} \; \in \; T\big({P\mathstrut}_{\hat{\bf x} \, | \,{\bf x}}^{}, \, {\bf X}\big)\big\}}(m)
\; \geq \; 1
\,\bigg|\, {P\mathstrut}_{\bf x}^{}\bigg\}\,\times
\nonumber \\
& \;\;\;\;\;\;\;\;\;\;\;\;\;\;\;\;\;\;\;\;\;\;\;\;\;\;\;\;\;\;\;\;\;\;\;\;\;
\;\;\;\;\;\;\;\;\;\;\;\;\;\;\;\;\;\;\;\;
\mathbbm{1}_{\bigg\{\substack{\displaystyle
\;\;\;\;\;\;\;\;\;\;\;\;\;\;\;\,
d\big({P\mathstrut}_{{\bf x}, \, \hat{\bf x}}^{}\big)
\; \leq \; D \\
\displaystyle D\big({P\mathstrut}_{{\bf x}, \, \hat{\bf x}}^{} \,\big\|\, {P\mathstrut}_{\bf x}^{} \!\times Q\big)
\; \geq \; R}
\bigg\}}({P\mathstrut}_{{\bf x}, \, \hat{\bf x}}^{})
\nonumber \\
& \overset{(a)}{\geq} \;\;
\max_{\;\;\;{P\mathstrut}_{{\bf x}, \, \hat{\bf x}}^{}}
\,\Pr\,\big\{{\bf X} \, \in \, T({P\mathstrut}_{\bf x}^{})\big\}
\,\cdot\,\Pr\,\Bigg\{\sum_{i \, = \, 1}^{{e\mathstrut}^{nR}}Z_{i} \; \geq \; 1
\Bigg\}
\,\cdot \,
\mathbbm{1}_{\bigg\{\substack{\displaystyle
\;\;\;\;\;\;\;\;\;\;\;\;\;\;\;\,
d\big({P\mathstrut}_{{\bf x}, \, \hat{\bf x}}^{}\big)
\; \leq \; D \\
\displaystyle D\big({P\mathstrut}_{{\bf x}, \, \hat{\bf x}}^{} \,\big\|\, {P\mathstrut}_{\bf x}^{} \!\times Q\big)
\; \geq \; R}
\bigg\}}({P\mathstrut}_{{\bf x}, \, \hat{\bf x}}^{})
\nonumber \\
& \overset{(b)}{\geq} \;\;
\max_{\;\;\;{P\mathstrut}_{{\bf x}, \, \hat{\bf x}}^{}}\;
\overbrace{{(n\, + \, 1)\mathstrut}^{-|{\cal X}|}\cdot
\exp\big\{-nD({P\mathstrut}_{\bf x}^{}\;\|\;P )\big\}}^{\Pr\,\{{\bf X} \, \in \, T({P\mathstrut}_{\bf x}^{})\}\; \geq}
\,\cdot \,
\exp\Big\{-n\left[
D\big({P\mathstrut}_{{\bf x}, \, \hat{\bf x}}^{} \,\big\|\, {P\mathstrut}_{\bf x}^{} \!\times Q\big) \, - \, R
\, + \, \epsilon_{1}
\right]\Big\}\,\times
\nonumber \\
& \;\;\;\;\;\;\;\;\;\;\;\;\;\;\;\;\;\;\;\;\;\;\;\;\;\;\;\;\;\;\,
\;\;\;\;\;\;\;\;\;\;\;\;\;\;\;\;\;\;\;\;\;\;\;\;
\underbrace{
{\Big(1\,-\,{e\mathstrut}^{-nR}\Big)\mathstrut}^{{e\mathstrut}^{nR}}}_{ \rightarrow \, 1/e}
\cdot\,
\mathbbm{1}_{\bigg\{\substack{\displaystyle
\;\;\;\;\;\;\;\;\;\;\;\;\;\;\;\,
d\big({P\mathstrut}_{{\bf x}, \, \hat{\bf x}}^{}\big)
\; \leq \; D \\
\displaystyle D\big({P\mathstrut}_{{\bf x}, \, \hat{\bf x}}^{} \,\big\|\, {P\mathstrut}_{\bf x}^{} \!\times Q\big)
\; \geq \; R}
\bigg\}}({P\mathstrut}_{{\bf x}, \, \hat{\bf x}}^{})
\nonumber \\
& \overset{(c)}{\geq} \;\;
\exp\big\{-n\big[E_{2}^{\,\text{types}}(R, D) \, + \, \epsilon_{1} \, + \, \epsilon_{4}\big]\big\}
\nonumber \\
& \overset{(d)}{\geq} \;\; \exp\big\{-n\big[E_{2}(R \, + \, \epsilon_{5}, \, D \, - \, \epsilon_{5})\, + \,
\epsilon_{1} \, + \, \epsilon_{4} \, + \, 2\epsilon_{5}\big]\big\}.
\label{eqLowBoB}
\end{align}
Explanation of steps:\newline
($a$)  holds for sufficiently large $n$, when
\begin{displaymath} 
\Pr\,\big\{\hat{\bf X\mathstrut}_{m} \; \in \; T\big({P\mathstrut}_{\hat{\bf x} \, | \,{\bf x}}^{}, \, {\bf X}\big)
\;\big|\;
{\bf X} \, \in \, T({P\mathstrut}_{\bf x}^{})\big\}
\;\geq\;
\exp\Big\{-n\left[D\big({P\mathstrut}_{{\bf x}, \, \hat{\bf x}}^{} \,\big\|\, {P\mathstrut}_{\bf x}^{} \!\times Q\big) \, + \, \epsilon_{1}\right]\Big\},
\end{displaymath}
for
\begin{displaymath}
Z_{i} \;\; \sim \;\; \text{i.i.d}\;\;\text{Bernoulli}\left(\exp\Big\{-n\left[D\big({P\mathstrut}_{{\bf x}, \, \hat{\bf x}}^{} \,\big\|\, {P\mathstrut}_{\bf x}^{} \!\times Q\big) \, + \, \epsilon_{1}\right]\Big\}\right).
\end{displaymath}
($b$) uses the lower bound on the probability of a type, and the second part of Lemma~\ref{lemma10}.\newline
($c$) holds for sufficiently large $n$, given $\epsilon_{4}\,>\,0$, with the exponent $E_{2}^{\,\text{types}}(R, \, D)$
defined as in (\ref{eqE2Sources}) with types in place of $T\circ W$.\newline
($d$) parallels the analogous step in (\ref{eqTermTwo}) with $E_{2}(\cdot, \cdot)$ defined in (\ref{eqE2Sources}).

The two lower bounds on the probability (\ref{eqLowBoA}) and (\ref{eqLowBoB}) result in the upper bound on the exponent:
\begin{displaymath}
\lim_{\epsilon\,\rightarrow\,0}\;
\min \big\{ E_{1}(R \, - \, \epsilon, \, D \, - \, \epsilon),\;
E_{2}(R \, + \, \epsilon, \, D \, - \, \epsilon)\big\}
\end{displaymath}
Analogously to (\ref{eqNewRightTermSimplified}),
this limit can be simplified as
\begin{displaymath}
\lim_{\epsilon\,\rightarrow\,0}\;
\min \big\{ E_{1}(R \, - \, \epsilon, \, D \, - \, \epsilon),\;
E_{2}(R \, + \, \epsilon, \, D \, - \, \epsilon)\big\}
\;\; = \;\;
\lim_{\epsilon\,\rightarrow\,0}\;
\min \big\{ E_{1}(R, \, D \, - \, \epsilon),\;
E_{2}(R, \, D \, - \, \epsilon)\big\}.
\end{displaymath}
\begin{thm} \label{thm11}
\begin{align}
&
\limsup_{n \, \rightarrow \, \infty} \; \left\{-\frac{1}{n}\ln P_{s}\right\}
\;\; \leq \;\;
\lim_{\epsilon\,\rightarrow\,0}\;
\min \big\{ E_{1}(R, \, D \, - \, \epsilon),\;
E_{2}(R, \, D \, - \, \epsilon)\big\}.
\nonumber
\end{align}
\end{thm}

In order to combine the lower and upper bounds given by Theorem~\ref{thm10} and Theorem~\ref{thm11},
and determine the true exponent of successful encoding,
observe, that the lower bound (\ref{eqLowerSuccess}) of Theorem~\ref{thm10}
can be rewritten, analogously to (\ref{eqE2Tilde}), as the RHS of (\ref{eqSuccess}).
As we have shown previously, the {\em implicit} expression on the RHS of (\ref{eqSuccess}) equals the {\em explicit} expression (\ref{eqES}).
As can be seen from (\ref{eqES}), it is a convex ($\cup$) function of $(R, D)$,
and therefore it is continuous in $(R, D)$, except for the points where its result switches to the value $+\infty$.
This occurs for $D \, = \, D_{\min} \, = \, \min_{x, \, \hat{x}}d(x, \hat{x})$.
For this value of $D$, the upper bound of Theorem~\ref{thm11} is $+\infty$, while the lower bound,
given by Theorem~\ref{thm10}, is finite.
For all other values of $D$ the bounds of Theorem~\ref{thm10} and Theorem~\ref{thm11} coincide.
Therefore we have proved Theorem~\ref{thm1}.

\section{{\bf Maximization over} \texorpdfstring{$\,Q\,$}{\em Q} {\bf of the decoding error exponent for arbitrary} \texorpdfstring{$D$}{\em D}} \label{S20}
The decoding error exponent, corresponding to the ``source duality'' decoder (\ref{eqErrorEvent}), is given by
(\ref{eqDecErrorArbD}), with the possible exception of the points $\big(R, \; D_{\min}(Q)\big)$, where
\begin{displaymath}
D_{\min}(Q) \; = \; \min_{y}\;\;\;
\min_{x:\;\; Q(x)\, > \, 0}
\;\;\;
\min_{\hat{x}:\;\; Q(\hat{x})\, > \, 0}
\;\ln \frac{P(y \,|\, x)}{P(y \,|\, \hat{x})}.
\end{displaymath}
Let us assume
$D_{\min}(Q)\, < \, 0$
for all $Q$, except for the degenerate $Q$ given by (\ref{eqQDegen}),
for which $D_{\min}(Q)\, = \, 0$.
Otherwise, there exist distinct channel inputs with exactly the same $P(y \,|\, x)$,
as a function of $y$,
i.e. indistinguishable at the channel output.
Such input letters can be merged without loss of generality.

With this assumption, we obtain the following. For $D \, > \, 0$, the maximal random coding exponent over $Q$ is given by the supremum of (\ref{eqDecErrorArbD})
over $Q$.
For any $D \, < \, 0$ the expression (\ref{eqDecErrorArbD}) yields the true exponent for the
degenerate distibution (\ref{eqQDegen}),
which equals $+\infty$.
Therefore, for $D \, < \, 0$ the maximal random coding exponent over $Q$ is $+\infty$.
For $D \, = \, 0$, the true exponent for the
degenerate distibution (\ref{eqQDegen}) is $0$,
which can be inferred directly from the definition of the decoder (\ref{eqErrorEvent}) itself,
and the same is given by the expression (\ref{eqDecErrorArbD}).
For all other $Q$, in the case of $D \, = \, 0$, the true exponent is also given by (\ref{eqDecErrorArbD}).
We conclude, that, for $D \, = \, 0$, the maximal random coding exponent over $Q$ is given by the supremum of (\ref{eqDecErrorArbD})
over $Q$.

To summarize the above, we have
\begin{thm} \label{thm12}
\begin{align}
& \sup_{Q(x)}\;
\lim_{n \, \rightarrow \, \infty} \; \left\{-\frac{1}{n}\ln P_{e}\right\}
\;\; =
\nonumber \\
&\sup_{Q(x)}\;\sup_{0 \, \leq \,\rho \,\leq \,1} \; \Bigg\{ -\inf_{s\,\geq\,0} \; \ln \; \sum_{x, \,y}Q(x)P(y \,|\, x)\Bigg[\sum_{\hat{x}}Q(\hat{x})\left[\frac{P(y \,|\, x)}{P(y \,|\, \hat{x})}\,e^{-D}\right]^{-s}\Bigg]^{\rho} \; - \; \rho R \Bigg\}.
\label{eqDecErrorArbDMaxQ}
\end{align}
\end{thm}
Note, that this is equal $+\infty$ for $D\, < \, 0$.

Now, using Lemma~\ref{lemma16} with Theorem~\ref{thm12}, we obtain, that the original Forney's random coding exponent coincides with
the ``source duality'' exponent for $D \, \geq \, 0\,$:
\begin{cor} \label{cor2}
{\em For $D \, \geq \, 0$}
\begin{displaymath}
\sup_{Q(x)}\;
\lim_{n \, \rightarrow \, \infty} \; \left\{-\frac{1}{n}\ln P_{e}\right\}
\;\; = \;\;
\sup_{Q(x)} \; {E\mathstrut}_{\text{bound}}(Q, R, D).
\end{displaymath}
\end{cor}

\section{\bf Derivation of the encoding failure exponent} \label{S21}
This will lead to both (\ref{eqFailure}) and (\ref{eqCorrectExtended}).

Here another generic lemma is needed, similar to Lemma~\ref{lemma9} and Lemma~\ref{lemma10}.
\begin{lemma} \label{lemma19}
{\em Let $Z_{i}\,\sim\, \text{i.i.d}\;\text{Bernoulli}\left({e\mathstrut}^{-nI}\right)$,
$i \, = \, 1, \, 2, \, ... \, , \, {e\mathstrut}^{nR}$.}
{\em If $I\,\leq\,R \, - \, \epsilon$, with $\epsilon \, > \, 0$, then}
\begin{equation} \label{eqBoundForZero}
\Pr\,\Bigg\{\sum_{i \, = \, 1}^{{e\mathstrut}^{nR}}Z_{i} \, = \, 0\Bigg\}
\;\; < \;\; \exp\big\{-{e\mathstrut}^{n\epsilon}\big\}.
\end{equation}
\end{lemma}
\begin{proof}
\begin{align}
\Pr\,\Bigg\{\sum_{i \, = \, 1}^{{e\mathstrut}^{nR}}Z_{i} \; = \; 0\Bigg\}
\;\; & = \;\;
\prod_{i \, = \, 1}^{{e\mathstrut}^{nR}}
\Pr\,\big\{Z_{i} \; = \; 0\big\}
\nonumber \\
& = \;\;
{\Big[1\,-\,{e\mathstrut}^{-nI}\Big]\mathstrut}^{{e\mathstrut}^{nR}}
\nonumber \\
& = \;\;
\bigg[\underbrace{\left(1 \, - \, {e\mathstrut}^{-nI}\right)^{-{e\mathstrut}^{nI}}}_{>\,e}\bigg]^{-{e\mathstrut}^{-nI}\cdot\,{e\mathstrut}^{nR}}
\; \overset{(*)}{<} \;\; \exp\Big\{-{e\mathstrut}^{n(R\,-\,I)}\Big\}
\;\; \leq \;\; \exp\big\{-{e\mathstrut}^{n\epsilon}\big\},
\nonumber
\end{align}
where ($*$) holds because $\;(1 - x)^{-1/x} \, > \, e\;$ for $\;0\,<\,x\,<\,1$.
\end{proof}

Upper bound on the probability of encoding failure:
\begin{align}
P_{f}
\;\; \leq \;
& \sum_{\substack{{P\mathstrut}_{\bf x}^{}:\\
\\
{R\mathstrut}^{\,\text{types}}({P\mathstrut}_{\bf x}^{}, \, Q, \, D)\; \leq \; R \, - \, 2\epsilon_{1}
}}
\,\Pr\,\big\{{\bf X} \, \in \, T({P\mathstrut}_{\bf x}^{})\big\}
\,\times
\nonumber \\
&
\;\;\;\;\;\;\;\;\;\;\;\;\;\;\;
\min_{\substack{{P\mathstrut}_{\hat{\bf x} \, | \,{\bf x}}^{}:\\
\\
d({P\mathstrut}_{{\bf x},\,\hat{\bf x}}^{}) \; \leq \; D
}}
\,\Pr\,\bigg\{
\sum_{m}
\mathbbm{1}_{\displaystyle\big\{\hat{\bf X\mathstrut}_{m} \; \in \; T\big({P\mathstrut}_{\hat{\bf x} \, | \,{\bf x}}^{}, \, {\bf X}\big)\big\}}(m)
\; = \; 0
\,\bigg|\, {P\mathstrut}_{\bf x}^{}\bigg\}
\nonumber \\
&
\;\;\;\;\;\;\;\;\;\;\;\;\;\;\;\;\;\;\;\;\;\;\;\;\;\;\;\;\;\;\;\;\;\;\;\;\;\;\;\;\;\;
\;\;\;\;\;\;\;\;\;\;\;\;\;\;\;\;\;\;\;\;\;\;\;\;\;
+ \sum_{\substack{{P\mathstrut}_{\bf x}^{}:\\
\\
{R\mathstrut}^{\,\text{types}}({P\mathstrut}_{\bf x}^{}, \, Q, \, D)\; \geq \; R \, - \, 2\epsilon_{1}}}
\,\Pr\,\big\{{\bf X} \, \in \, T({P\mathstrut}_{\bf x}^{})\big\}
\nonumber \\
\overset{(a)}{\leq} \;
& \sum_{\substack{{P\mathstrut}_{\bf x}^{}:\\
\\
{R\mathstrut}^{\,\text{types}}({P\mathstrut}_{\bf x}^{}, \, Q, \, D)\; \leq \; R \, - \, 2\epsilon_{1}
}}
\,\underbrace{\Pr\,\big\{{\bf X} \, \in \, T({P\mathstrut}_{\bf x}^{})\big\}}_{\leq\, 1}
\,\times
\nonumber \\
&
\;\;\;\;\;\;\;\;\;\;\;\;\;\;\;
\min_{\substack{{P\mathstrut}_{\hat{\bf x} \, | \,{\bf x}}^{}:\\
\\
d({P\mathstrut}_{{\bf x},\,\hat{\bf x}}^{}) \; \leq \; D
}}
\,\Pr\,\Bigg\{\sum_{i \, = \, 1}^{{e\mathstrut}^{nR}}Z_{i} \, = \, 0\Bigg\}
\;\;\;\;\;\;\;\,
+ \sum_{\substack{{P\mathstrut}_{\bf x}^{}:\\
\\
{R\mathstrut}^{\,\text{types}}({P\mathstrut}_{\bf x}^{}, \, Q, \, D)\; \geq \; R \, - \, 2\epsilon_{1}}}
\,\Pr\,\big\{{\bf X} \, \in \, T({P\mathstrut}_{\bf x}^{})\big\}
\nonumber \\
&
\nonumber \\
\overset{(b)}{\leq} \;
& \sum_{\substack{{P\mathstrut}_{\bf x}^{}:\\
\\
{R\mathstrut}^{\,\text{types}}({P\mathstrut}_{\bf x}^{}, \, Q, \, D)\; \leq \; R \, - \, 2\epsilon_{1}
}}
\,\Pr\,\Bigg\{\sum_{i \, = \, 1}^{{e\mathstrut}^{nR}} B_{i} \, = \, 0\Bigg\}
\;\;\;\;\;\;\;\,
+ \sum_{\substack{{P\mathstrut}_{\bf x}^{}:\\
\\
{R\mathstrut}^{\,\text{types}}({P\mathstrut}_{\bf x}^{}, \, Q, \, D)\; \geq \; R \, - \, 2\epsilon_{1}}}
\,\Pr\,\big\{{\bf X} \, \in \, T({P\mathstrut}_{\bf x}^{})\big\}
\nonumber \\
&
\nonumber \\
\overset{(c)}{\leq} \; &
\;\;\;\;\;\;\;\;\;\;\;\;\,
\sum_{{P\mathstrut}_{\bf x}^{}}
\;\;\;\;\;\;\;\;\;\;\;\;\;
\exp\big\{-{e\mathstrut}^{n\epsilon_{1}}\big\}
\;\;\;\;
\;\;\;\;\;\;\;\;\;\;\;\,
+ \sum_{\substack{{P\mathstrut}_{\bf x}^{}:\\
\\
{R\mathstrut}^{\,\text{types}}({P\mathstrut}_{\bf x}^{}, \, Q, \, D)\; \geq \; R \, - \, 2\epsilon_{1}}}
\,\Pr\,\big\{{\bf X} \, \in \, T({P\mathstrut}_{\bf x}^{})\big\}
\nonumber \\
&
\nonumber \\
\overset{(d)}{\leq} \; &
\;\;\;\;\;\;\;\;\;\;\;\;
{(n\, + \, 1)\mathstrut}^{|{\cal X}|}
\cdot \;
\exp\big\{-{e\mathstrut}^{n\epsilon_{1}}\big\}
\;\;\;\;
\;\;\;\;\;\;\;\;\;\;\;\,
+ \sum_{\substack{{P\mathstrut}_{\bf x}^{}:\\
\\
{R\mathstrut}^{\,\text{types}}({P\mathstrut}_{\bf x}^{}, \, Q, \, D)\; \geq \; R \, - \, 2\epsilon_{1}}}
\,\exp\big\{-nD({P\mathstrut}_{\bf x}^{}\;\|\;P )\big\}
\nonumber \\
\overset{(e)}{\leq} \; &
\;\;\;\;\;\;\;\;\;\;\;\;
{(n\, + \, 1)\mathstrut}^{|{\cal X}|}
\cdot \;
\exp\big\{-{e\mathstrut}^{n\epsilon_{1}}\big\}
\;\;\;\;
\;\;\;\;\;\;\;\;\,
+ \sum_{\substack{{P\mathstrut}_{\bf x}^{}:\\
\\
R({P\mathstrut}_{\bf x}^{}, \, Q, \, D \, - \, \epsilon_{2})\; \geq \; R \, - \, 2\epsilon_{1} \, - \, \epsilon_{2}}}
\,\exp\big\{-nD({P\mathstrut}_{\bf x}^{}\;\|\;P )\big\}
\nonumber \\
\overset{(f)}{\leq} \; &
\;\;\;\;\;\;\;\;\;\;\;\;
{(n\, + \, 1)\mathstrut}^{|{\cal X}|}
\cdot \;
\exp\big\{-{e\mathstrut}^{n\epsilon_{1}}\big\}
\;\;
+ \;\; {(n\, + \, 1)\mathstrut}^{|{\cal X}|}
\cdot
\,\exp\big\{-n E_{f}^{\,\text{types}}(R \, - \, 2\epsilon_{1} \, - \, \epsilon_{2}, \, D \, - \, \epsilon_{2})\big\}
\nonumber \\
\overset{(g)}{\leq} \; &
\;\;\;\;\;\;\;\;\;\;\;\;
{(n\, + \, 1)\mathstrut}^{|{\cal X}|}
\cdot \;
\exp\big\{-{e\mathstrut}^{n\epsilon_{1}}\big\}
\;\;
+ \;\; {(n\, + \, 1)\mathstrut}^{|{\cal X}|}
\cdot
\,\exp\big\{-n E_{f}(R \, - \, 2\epsilon_{1} \, - \, \epsilon_{2}, \, D \, - \, \epsilon_{2})\big\}.
\label{eqPfUpperBound}
\end{align}
Explanation of steps:\newline
($a$)   holds for sufficiently large $n$, when
\begin{align} 
\Pr\,\big\{\hat{\bf X\mathstrut}_{m} \; \in \; T\big({P\mathstrut}_{\hat{\bf x} \, | \,{\bf x}}^{}, \, {\bf X}\big)
\;\big|\;
{\bf X} \, \in \, T({P\mathstrut}_{\bf x}^{})\big\}
\; & \geq \;
\exp\Big\{-n\left[D\big({P\mathstrut}_{{\bf x}, \, \hat{\bf x}}^{} \,\big\|\, {P\mathstrut}_{\bf x}^{} \!\times Q\big) \, + \, \epsilon_{1}\right]\Big\},
\nonumber
\end{align}
with
\begin{displaymath}
Z_{i} \;\; \sim \;\; \text{i.i.d}\;\;\text{Bernoulli}\left(
\exp\Big\{-n\left[D\big({P\mathstrut}_{{\bf x}, \, \hat{\bf x}}^{} \,\big\|\, {P\mathstrut}_{\bf x}^{} \!\times Q\big) \, + \, \epsilon_{1}\right]\Big\}
\right).
\end{displaymath}
($b$) holds for
\begin{displaymath}
B_{i} \;\; \sim \;\; \text{i.i.d}\;\;\text{Bernoulli}\left(
\exp\Big\{-n\big[{R\mathstrut}^{\,\text{types}}({P\mathstrut}_{\bf x}^{}, Q, D) \, + \, \epsilon_{1}\big]\Big\}
\right),
\end{displaymath}
where
\begin{equation} \label{eqRTypes}
{R\mathstrut}^{\,\text{types}}({P\mathstrut}_{\bf x}^{}, Q, D)
\;\; \triangleq \;
\min_{\substack{{P\mathstrut}_{\hat{\bf x} \, | \,{\bf x}}^{}:\\
\\
d({P\mathstrut}_{{\bf x},\,\hat{\bf x}}^{}) \; \leq \; D
}}
D\big({P\mathstrut}_{{\bf x}, \, \hat{\bf x}}^{} \,\big\|\, {P\mathstrut}_{\bf x}^{} \!\times Q\big).
\end{equation}
($c$) holds by Lemma~\ref{lemma19} for
\begin{displaymath}
I \;\; = \;\; {R\mathstrut}^{\,\text{types}}({P\mathstrut}_{\bf x}^{}, Q, D) \, + \, \epsilon_{1}
\;\; \leq \;\; R \, - \, 2\epsilon_{1} \, + \, \epsilon_{1}
\;\; = \;\; R \, - \, \epsilon_{1}.
\end{displaymath}
($d$) uses the upper bound on the probability of a type (\ref{eqTypeProbSource}).\newline
($e$) Let ${W\mathstrut}^{*}$ denote the conditional distribution, achieving $R({P\mathstrut}_{\bf x}^{}, \, Q, \, D \, - \, \epsilon_{2})\,<\,+\infty$
for some $\epsilon_{2}\,>\,0$.
This implies
\begin{align}
D({P\mathstrut}_{\bf x}^{} \circ {W\mathstrut}^{*} \,\|\, {P\mathstrut}_{\bf x}^{} \!\times Q)
\;\; & = \;\; R({P\mathstrut}_{\bf x}^{}, \, Q, \, D \, - \, \epsilon_{2}),
\label{eqComb1} \\
d({P\mathstrut}_{\bf x}^{} \circ {W\mathstrut}^{*}) \;\; & \leq \;\; D \, - \, \epsilon_{2}.
\nonumber
\end{align}
Let ${W\mathstrut}_{n}^{*}$
denote a quantized version of the conditional distribution ${W\mathstrut}^{*}$
with variable precision $1/\big(n {P\mathstrut}_{\bf x}^{}(x)\big)$, i.e. a set of types with denominators $n {P\mathstrut}_{\bf x}^{}(x)$,
such that the joint distribution ${P\mathstrut}_{\bf x}^{} \circ {W\mathstrut}_{n}^{*}$ is a type with denominator $n$.
Observe, that the differences between ${P\mathstrut}_{\bf x}^{} \circ {W\mathstrut}^{*}$ and ${P\mathstrut}_{\bf x}^{} \circ {W\mathstrut}_{n}^{*}$
do not exceed $\tfrac{1}{n}$.
Therefore,
since the divergence, as a function of ${P\mathstrut}_{\bf x}^{} \circ W$, has bounded derivatives, and also the distortion measure $d(x, \hat{x})$
is bounded, for any $\epsilon_{2}\,>\,0$ there exists $n$ large enough,
such that the quantized distribution ${W\mathstrut}_{n}^{*}$
satisfies
\begin{align}
D({P\mathstrut}_{\bf x}^{} \circ {W\mathstrut}_{n}^{*} \,\|\, {P\mathstrut}_{\bf x}^{} \!\times Q)
\;\; & \leq \;\; D({P\mathstrut}_{\bf x}^{} \circ {W\mathstrut}^{*} \,\|\, {P\mathstrut}_{\bf x}^{} \!\times Q) \, + \, \epsilon_{2},
\label{eqComb2} \\
d({P\mathstrut}_{\bf x}^{} \circ {W\mathstrut}_{n}^{*}) \;\; & \leq \;\; D.
\nonumber
\end{align}
The last inequality implies
\begin{align}
D({P\mathstrut}_{\bf x}^{} \circ {W\mathstrut}_{n}^{*} \,\|\, {P\mathstrut}_{\bf x}^{} \!\times Q)
\;\; & \geq \;\; {R\mathstrut}^{\,\text{types}}({P\mathstrut}_{\bf x}^{}, Q, D).
\label{eqComb3}
\end{align}
The relations (\ref{eqComb3}), (\ref{eqComb2}), (\ref{eqComb1}) together give
\begin{align}
R({P\mathstrut}_{\bf x}^{}, \, Q, \, D \, - \, \epsilon_{2}) \, + \, \epsilon_{2}
\;\; & \geq \;\; {R\mathstrut}^{\,\text{types}}({P\mathstrut}_{\bf x}^{}, Q, D).
\label{eqResult4}
\end{align}
This explains ($e$).\newline
($f$) uses the definition
\begin{equation} \label{eqEfTypes}
E_{f}^{\,\text{types}}(R, D)
\;\; \triangleq \;\;
\min_{\substack{{P\mathstrut}_{\bf x}^{}:\\
\\
R({P\mathstrut}_{\bf x}^{}, \,Q, \,D) \; \geq \; R}} \; D({P\mathstrut}_{\bf x}^{} \;\|\; P).
\end{equation}
($g$) $E_{f}^{\,\text{types}}(R, D)$ is bounded from below by $E_{f}(R, D)$ defined in (\ref{eqFailure}).

We conclude from (\ref{eqPfUpperBound}):
\begin{thm} \label{thm13}
\begin{align}
& \liminf_{n \, \rightarrow \, \infty} \; \left\{-\frac{1}{n}\ln P_{f}\right\}
\;\; \geq \;\;
\lim_{\epsilon\,\rightarrow\,0}\;
E_{f}(R \, - \, \epsilon, \, D \, - \, \epsilon).
\label{eqLowerFailure}
\end{align}
\end{thm}

Lower bound on the probability of encoding failure:
\begin{align}
P_{f}
\;\; \overset{(a)}{\geq} \;
& \max_{\substack{{P\mathstrut}_{\bf x}^{}:\\
\\
{R\mathstrut}({P\mathstrut}_{\bf x}^{}, \, Q, \, D)\; \geq \; R \, + \, \epsilon_{1}
}}
\,\Pr\,\big\{{\bf X} \, \in \, T({P\mathstrut}_{\bf x}^{})\big\}
\,\times
\nonumber \\
& \;\;\;\;\;\;\;\;\;\;\;
\left[1 \; - \;
\sum_{{P\mathstrut}_{\hat{\bf x} \, | \,{\bf x}}^{}:\;\;
d({P\mathstrut}_{{\bf x},\,\hat{\bf x}}^{}) \; \leq \; D
}
\;\;
\sum_{m \, = \, 1}^{{e\mathstrut}^{nR}}
\Pr\,\Big\{
\hat{\bf X\mathstrut}_{m} \; \in \; T\big({P\mathstrut}_{\hat{\bf x} \, | \,{\bf x}}^{}, \, {\bf X}\big)
\,\Big|\, {P\mathstrut}_{\bf x}^{}\Big\}
\right]
\nonumber \\
\overset{(b)}{\geq} \;
& \max_{\substack{{P\mathstrut}_{\bf x}^{}:\\
\\
{R\mathstrut}({P\mathstrut}_{\bf x}^{}, \, Q, \, D)\; \geq \; R \, + \, \epsilon_{1}
}}
\,\Pr\,\big\{{\bf X} \, \in \, T({P\mathstrut}_{\bf x}^{})\big\}
\,\times
\nonumber \\
& \;\;\;\;\;\;\;\;\;\;\;
\left[1 \; - \;
\sum_{{P\mathstrut}_{\hat{\bf x} \, | \,{\bf x}}^{}:\;\;
d({P\mathstrut}_{{\bf x},\,\hat{\bf x}}^{}) \; \leq \; D
}
\;\;
\exp\Big\{-n\left[D\big({P\mathstrut}_{{\bf x}, \, \hat{\bf x}}^{} \,\big\|\, {P\mathstrut}_{\bf x}^{} \!\times Q\big) \, - \, R\right]\Big\}
\right]
\nonumber \\
\overset{(c)}{\geq} \;
& \max_{\substack{{P\mathstrut}_{\bf x}^{}:\\
\\
{R\mathstrut}({P\mathstrut}_{\bf x}^{}, \, Q, \, D)\; \geq \; R \, + \, \epsilon_{1}
}}
\,\Pr\,\big\{{\bf X} \, \in \, T({P\mathstrut}_{\bf x}^{})\big\}
\,\cdot\,
\Bigg[1 \; - \;
\sum_{{P\mathstrut}_{\hat{\bf x} \, | \,{\bf x}}^{}
}
\exp\Big\{-n\left[{R\mathstrut}({P\mathstrut}_{\bf x}^{}, \, Q, \, D) \, - \, R\right]\Big\}
\Bigg]
\nonumber \\
\geq \;
& \max_{\substack{{P\mathstrut}_{\bf x}^{}:\\
\\
{R\mathstrut}({P\mathstrut}_{\bf x}^{}, \, Q, \, D)\; \geq \; R \, + \, \epsilon_{1}
}}
\,\Pr\,\big\{{\bf X} \, \in \, T({P\mathstrut}_{\bf x}^{})\big\}
\,\cdot\,
\Bigg[1 \; - \;
\sum_{{P\mathstrut}_{\hat{\bf x} \, | \,{\bf x}}^{}
}
\exp\{-n\epsilon_{1}\}
\Bigg]
\nonumber \\
\overset{(d)}{\geq} \;
& \max_{\substack{{P\mathstrut}_{\bf x}^{}:\\
\\
{R\mathstrut}({P\mathstrut}_{\bf x}^{}, \, Q, \, D)\; \geq \; R \, + \, \epsilon_{1}
}}
{(n\, + \, 1)\mathstrut}^{-|{\cal X}|}\cdot
\exp\big\{-nD({P\mathstrut}_{\bf x}^{}\;\|\;P )\big\}
\cdot
\Big[1 \; - \;
{(n\, + \, 1)\mathstrut}^{|{\cal X}|\cdot|\hat{\cal X}|}\cdot
\exp\{-n\epsilon_{1}\}
\Big]
\nonumber \\
\overset{(e)}{\geq} \;
& \max_{\substack{{P\mathstrut}_{\bf x}^{}:\\
\\
{R\mathstrut}({P\mathstrut}_{\bf x}^{}, \, Q, \, D)\; \geq \; R \, + \, \epsilon_{1}
}}
\exp\big\{-n\big[D({P\mathstrut}_{\bf x}^{}\;\|\;P )\, + \, \epsilon_{2}\big]\big\}
\;\; \overset{(f)}{=} \;\;
\exp\big\{-n \big[E_{f}^{\,\text{types}}(R \, + \, \epsilon_{1}, \, D) \, + \, \epsilon_{2}\big]\big\}
\nonumber \\
\overset{(g)}{\geq} \;
& \;
\exp\big\{-n \big[E_{f}(R \, + \, \epsilon_{1} \, + \, \epsilon_{3}, \, D) \, + \, \epsilon_{2} \, + \, \epsilon_{3}\big]\big\}.
\label{eqPfLowerBound}
\end{align}
Explanation of steps:\newline
($a$) uses the union bound for the probability of the {\em complementary} event of encoding success.\newline
($b$) uses the upper bound on the probability of a conditional type
\begin{align}
\Pr\,\big\{\hat{\bf X\mathstrut}_{m} \; \in \; T\big({P\mathstrut}_{\hat{\bf x} \, | \,{\bf x}}^{}, \, {\bf X}\big)
\;\big|\;
{\bf X} \, \in \, T({P\mathstrut}_{\bf x}^{})\big\}
\; & \leq \;
\exp\big\{-n D\big({P\mathstrut}_{{\bf x}, \, \hat{\bf x}}^{} \,\big\|\, {P\mathstrut}_{\bf x}^{} \!\times Q\big)\big\},
\nonumber
\end{align}
($c$) follows by the definition (\ref{eqRTypes}) and the property
$\,R({P\mathstrut}_{\bf x}^{}, Q, D) \, \leq \, {R\mathstrut}^{\,\text{types}}({P\mathstrut}_{\bf x}^{}, Q, D)$.\newline
($d$) uses the lower bound on the probability of a type and the polynomial upper bound on the number of conditional types.\newline
($e$) holds for sufficiently large $n$ for a given $\epsilon_{2} \, > \, 0$.\newline
($f$) follows by the definition (\ref{eqEfTypes}).\newline
($g$) Let ${T\mathstrut}^{*}$ denote the distribution achieving $E_{f}(R \, + \, \epsilon_{1} \, + \, \epsilon_{3}, \, D)\,<\,+\infty$.
Then
\begin{align}
D({T\mathstrut}^{*}\;\|\; P)
\;\; & = \;\; E_{f}(R \, + \, \epsilon_{1} \, + \, \epsilon_{3}, \, D),
\label{eqRel1} \\
R({T\mathstrut}^{*}\!, Q, D) \;\; & \geq \;\; R \, + \, \epsilon_{1} \, + \, \epsilon_{3}.
\nonumber
\end{align}
Let ${T\mathstrut}_{n}^{*}$
denote a quantized version of the distribution ${T\mathstrut}^{*}$
with precision $\frac{1}{n}$, i.e. a type with denominator $n$.
We note that both $D({T\mathstrut}^{*}\;\|\; P)$ and $R({T\mathstrut}^{*}\!, Q, D)$
are convex ($\cup$)
functions of $T$,
$R({T\mathstrut}^{*}\!, Q, D)$ is lower semi-continuous.
The latter property implies that if $R({T\mathstrut}^{*}\!, Q, D)$ is $+\infty$,
so is $R({T\mathstrut}_{n}^{*} , Q, D)$
for sufficiently large $n$.
Thus for any $\epsilon_{3} \, > \, 0$ there exists $n$ sufficiently large, such that
\begin{align}
D({T\mathstrut}_{n}^{*}\;\|\; P)
\;\; & \leq \;\; D({T\mathstrut}^{*}\;\|\; P) \, + \, \epsilon_{3},
\label{eqRel2} \\
R({T\mathstrut}_{n}^{*} , Q, D) \;\; & \geq \;\; R \, + \, \epsilon_{1}.
\nonumber
\end{align}
The last inequality implies
\begin{equation} \label{eqRel3}
D({T\mathstrut}_{n}^{*}\;\|\; P)
\;\; \geq \;\; E_{f}^{\,\text{types}}(R \, + \, \epsilon_{1}, \, D).
\end{equation}
The relations (\ref{eqRel3}), (\ref{eqRel2}), (\ref{eqRel1}) give
\begin{displaymath}
E_{f}(R \, + \, \epsilon_{1} \, + \, \epsilon_{3}, \, D) \, + \, \epsilon_{3}
\;\; \geq \;\; E_{f}^{\,\text{types}}(R \, + \, \epsilon_{1}, \, D).
\end{displaymath}
This explains ($g$).

We conclude from (\ref{eqPfLowerBound}):
\begin{thm} \label{thm14}
\begin{align}
& \limsup_{n \, \rightarrow \, \infty} \; \left\{-\frac{1}{n}\ln P_{f}\right\}
\;\; \leq \;\;
\lim_{\epsilon\,\rightarrow\,0}\;
E_{f}(R \, + \, \epsilon, \, D).
\label{eqUpperFailure}
\end{align}
\end{thm}

The bounds of Theorem~\ref{thm13} and Theorem~\ref{thm14} prove Theorem~\ref{thm2}.

\section{{\bf Proof of the identity} \texorpdfstring{$\;R(Q\circ P, \, Q, \, 0) \, = \, I(Q\circ P)$}{\em R(Q o P, Q, 0) = I(Q o P)}} \label{S22}

{\em Proposition 2:}
\begin{equation} \label{eqRQPQ0IQP}
R(Q\circ P, \, Q, \, 0) \, = \, I(Q\circ P).
\end{equation}
\begin{proof} This proof uses Lemma~\ref{lemma1}. Alternatively, it can be proved by the method of Lagrange multipliers.
We use the explicit expression of Lemma~\ref{lemma1}:
\begin{align}
R(Q\circ P, \, Q, \, 0) \; & = \; \sup_{s\,\geq\,0} \; \Bigg\{-\sum_{x, \, y}Q(x)P(y\,|\, x) \ln \sum_{\hat{x}}Q(\hat{x})\left[\frac{P(y \,|\, x)}{P(y \,|\, \hat{x})}\right]^{-s}\Bigg\}
\label{eqRQPQ0Expl} \\
& = \; \sup_{s \,\geq \,0} \; \min_{W(\hat{x} \,|\, x, \, y)} \;\big\{ D\big((Q\circ P) \circ W \; \| \; (Q\circ P) \times Q\big) \; + \; s\big[d\big((Q\circ P)\circ W\big) \; - \; D\big]\big\}.
\label{eqConcavityins}
\end{align}
In the above, (\ref{eqRQPQ0Expl}) is the same as (\ref{eqRTQD}), and (\ref{eqConcavityins}) is the same as (\ref{eqWs}).
Observe, that the expression inside the supremum of (\ref{eqRQPQ0Expl}) and (\ref{eqConcavityins}) is a concave ($\cap$) function of $s$,
as a minimum of affine functions of $s$.
We conclude, that in order to find the maximum over $s$, it suffices to find such $s$, for which the derivative of the expression in (\ref{eqRQPQ0Expl}) is zero.

Differentiation with respect to $s$ gives:
\begin{align}
& \frac{d}{d s}\;
\Bigg\{-\sum_{x, \, y}Q(x)P(y\,|\, x) \ln \sum_{\hat{x}}Q(\hat{x})\left[\frac{P(y \,|\, x)}{P(y \,|\, \hat{x})}\right]^{-s}\Bigg\}
\nonumber \\
& =
-\sum_{x, \, y}Q(x)P(y\,|\, x) \,\frac{1}{\sum_{a}Q(a)\left[\frac{P(y \,|\, x)}{P(y \,|\, a)}\right]^{-s}}
\,\sum_{\hat{x}}Q(\hat{x})\left[\frac{P(y \,|\, x)}{P(y \,|\, \hat{x})}\right]^{-s}
\left(-\ln \frac{P(y \,|\, x)}{P(y \,|\, \hat{x})}\right)
\nonumber \\
& =
\sum_{x, \, y}Q(x)P(y\,|\, x) \,\frac{1}{\sum_{a}Q(a)P^{s}(y \,|\, a)}
\,\sum_{\hat{x}}Q(\hat{x})P^{s}(y \,|\, \hat{x})
\ln \frac{P(y \,|\, x)}{P(y \,|\, \hat{x})}.
\nonumber
\end{align}
Now, let us substitute $s\, = \, 1$:
\begin{align}
& \sum_{x, \, y}Q(x)P(y\,|\, x) \,\frac{1}{\sum_{a}Q(a)P(y \,|\, a)}
\,\sum_{\hat{x}}Q(\hat{x})P(y \,|\, \hat{x})
\ln \frac{P(y \,|\, x)}{P(y \,|\, \hat{x})}
\nonumber \\
= \; &
\sum_{x, \, y}Q(x)P(y\,|\, x) \,\frac{\sum_{\hat{x}}Q(\hat{x})P(y \,|\, \hat{x})}{\sum_{a}Q(a)P(y \,|\, a)}
\,
\ln P(y \,|\, x) \, - \,
\sum_{y} \,\frac{\sum_{x}Q(x)P(y\,|\, x)}{\sum_{a}Q(a)P(y \,|\, a)}\sum_{\hat{x}}Q(\hat{x})P(y \,|\, \hat{x})
\,
\ln P(y \,|\, \hat{x})
\nonumber \\
= \; &
\sum_{x, \, y}Q(x)P(y\,|\, x)
\,
\ln P(y \,|\, x) \, - \,
\sum_{\hat{x}, \, y}Q(\hat{x})P(y \,|\, \hat{x})
\,
\ln P(y \,|\, \hat{x}) \; = \; 0.
\nonumber
\end{align}
We conclude, that $s^{*} \, = \, 1$. With $s^{*} \, = \, 1$ we obtain:
\begin{align}
R(Q\circ P, \, Q, \, 0) \; & = \; -\sum_{x, \, y}Q(x)P(y\,|\, x) \ln \sum_{\hat{x}}Q(\hat{x})\left[\frac{P(y \,|\, x)}{P(y \,|\, \hat{x})}\right]^{-1}
\nonumber \\
& = \; -\sum_{x, \, y}Q(x)P(y\,|\, x) \ln \frac{\sum_{\hat{x}}Q(\hat{x})P(y \,|\, \hat{x})}{P(y \,|\, x)}
\; = \; I(Q\circ P).
\nonumber
\end{align}
\end{proof}

Note, that the minimizing $W(\hat{x} \, | \, x, \, y)$ in (\ref{eqConcavityins}) is given by
\begin{displaymath}
W^{*}(\hat{x} \, | \, x, \, y) \; = \;
\frac{Q(\hat{x})P(y \,|\, \hat{x})}{\sum_{\hat{x}}Q(\hat{x})P(y \,|\, \hat{x})}.
\end{displaymath}

\section{{\bf Proof of the identity} \texorpdfstring{$\;R(Q\circ P, \, Q, \, 0) \, = \, R(Q\circ P, \, 0)$}{\em R(Q o P, Q, 0) = R(Q o P, 0)}} \label{S23}
The rate-distortion function of the effective source $Q\circ P$ is given by the minimum of the function $ R(Q\circ P, Q', D)$ over $Q'$:
\begin{align}
R(Q\circ P, \, 0) \; = \; \min_{Q'} \; R(Q\circ P, \, Q', \, 0) \;
& \overset{(\ref{eqRTQD})}{=} \; \min_{Q'} \; \sup_{s\,\geq\,0} \; \Bigg\{-\sum_{x, \, y}Q(x)P(y\,|\, x) \ln \sum_{\hat{x}}Q'(\hat{x})\left[\frac{P(y \,|\, x)}{P(y \,|\, \hat{x})}\right]^{-s}\Bigg\}
\nonumber \\
& \geq \; \min_{Q'} \;  \Bigg\{-\sum_{x, \, y}Q(x)P(y\,|\, x) \ln \sum_{\hat{x}}Q'(\hat{x})\left[\frac{P(y \,|\, x)}{P(y \,|\, \hat{x})}\right]^{-1}\Bigg\}
\nonumber \\
& = \; -\sum_{x, \, y}Q(x)P(y\,|\, x) \ln \frac{\sum_{\hat{x}}Q(\hat{x})P(y \,|\, \hat{x})}{P(y \,|\, x)}
\; = \; I(Q\circ P)
\nonumber \\
& \!\!\!\! \overset{(\ref{eqRQPQ0IQP})}{=} \; R(Q\circ P, \, Q, \, 0)
\; \geq \;
\min_{Q'} \; R(Q\circ P, \, Q', \, 0) \; = \; R(Q\circ P, \, 0).
\nonumber
\end{align}
We conclude, that the mutual information $I(Q\circ P)$ can be viewed also as the rate-distortion function $R(Q\circ P, \, D)$,
corresponding to the distortion measure (\ref{eqDmeasure}),
evaluated at $D\, = \, 0$.
The channel capacity therefore is given by the maximum of this rate-distortion function at $D \, = \, 0\,$\footnote{Not the lowest $D$, because the chosen distortion measure (\ref{eqDmeasure}) can have negative values.}:
\begin{equation} \label{eqCapacityRDF}
C(P) \; = \; \max_{Q} \; R(Q\circ P, \, 0).
\end{equation}

\bibliographystyle{IEEEtran}

\end{document}